\documentclass[a4paper,english]{lipics}

\usepackage{amsmath}
\usepackage{amssymb}
\usepackage{amsthm}
\usepackage{latexsym}
\usepackage{graphicx}
\usepackage{color}
\usepackage{amsfonts}
\usepackage{stmaryrd}
\usepackage{url}
\usepackage{ulem}
\usepackage{mathpartir}
\usepackage{enumerate}
\usepackage{prooftree}
\usepackage{framed}
\usepackage{rotating} 

\long\def\ignore#1{\relax}
\newcommand{\pierre}[1]{\red{#1}}

\newcommand{\fnsz}{\footnotesize}
\newcommand{\ssz}{\scriptsize}

\newcommand{\ssst}{\scriptscriptstyle}

\newcommand{\ie}{\textit{i.e.}\ }
\newcommand{\eg}{\textit{e.g.},\ }
\newcommand{\wrt}{w.r.t.\ }

\newcommand{\resp}{resp.\ }
\newcommand{\via}{\textit{via}\ }

\renewcommand{\S}{Sec.}

\newcommand{\leqs}{\leqslant}
\newcommand{\geqs}{\geqslant}

\newcommand{\blut}{\pierre{blut}\ }

\newcommand{\fublainv}[1]{}

\newcommand{\ct}{\mathord{\cdot}}
\newcommand{\set}[1]{\{ #1 \}}
\newcommand{\mult}[1]{ [ #1 ]}
\newcommand{\ovl}[1]{\overline{#1}}

\newcommand{\eset}{\emptyset}

\newcommand{\emul}{\mult{\,}}
\newcommand{\tv}{o}
\newcommand{\inter}{\wedge}

\newcommand{\rstr}[1]{|_{#1}}
\newcommand{\omin}{\ominus}

\newcommand{\phd}{\phantom{.}}

\newcommand{\subeq}{\subseteq}

\newcommand{\wdg}{\wedge}

\newcommand{\Bohm}{B\"ohm }

\newcommand{\bbN}{\mathbb{N}}

\newcommand{\ttB}{\mathtt{B}}
\newcommand{\ttC}{\mathtt{C}}
\newcommand{\ttE}{\mathtt{E}}

\newcommand{\ttL}{\mathtt{L}}

\newcommand{\ttO}{\mathtt{O}}
\newcommand{\ttP}{\mathtt{P}}
\newcommand{\ttR}{\mathtt{R}}
\newcommand{\ttS}{\mathtt{S}}

\newcommand{\ttSh}{\ttS_{\tth}}
\newcommand{\ttT}{\mathtt{T}}

\newcommand{\tte}{\mathtt{e}}

\newcommand{\tti}{\mathtt{i}}

\newcommand{\tth}{\mathtt{h}}

\newcommand{\ttr}{\mathtt{r}}

\newcommand{\scrM}{\mathscr{M}}     

\newcommand{\scrP}{\mathscr{P}}
\newcommand{\scrR}{\mathscr{R}}
\newcommand{\scrRo}{\scrR_0}

\newcommand{\scrV}{\mathscr{V}}

\newcommand{\calC}{\mathcal{C}}

\newcommand{\calG}{\mathcal{G}}

\newcommand{\calT}{\mathcal{T}}

\newcommand{\secu}{^{\prime \prime}}

\newcommand{\ovla}{\ovl{a}}
\newcommand{\ovlal}{\ovl{\al}}
\newcommand{\ovlalp}{\ovl{\al'}}

\newcommand{\trew}{\tilde{\rew}}    
\newcommand{\tilAx}{\tilde{\Ax}}

\newcommand{\Nmzo}{\bbN\setminus\set{0,1}}

\newcommand{\ttarg}{\mathtt{arg}}
\newcommand{\ttdom}{\mathtt{dom}}
\newcommand{\ttcodom}{\mathtt{codom}}

\newcommand{\tttr}{\mathtt{tr}}
\newcommand{\ttchild}{\mathtt{child}}
\newcommand{\ttid}{\mathtt{id}}
\newcommand{\ttpos}{\mathtt{pos}}
\newcommand{\ttref}{\mathtt{ref}}

\newcommand{\ttax}{\mathtt{ax}}
\newcommand{\ttabs}{\mathtt{abs}}
\newcommand{\ttapp}{\mathtt{app}}
\newcommand{\ttout}{\mathtt{out}} 

\newcommand{\ttAx}{\mathtt{Ax}}
\newcommand{\ttPol}{\mathtt{Pol}}
\newcommand{\ttSc}{\mathtt{Sc}}
\newcommand{\ttTg}{\mathtt{Tg}}

\newcommand{\iI}{{i \in I}}

\newcommand{\kK}{{k \in K}}

\newcommand{\jJi}{{j\in J(i)}}

\newcommand{\lam}{\lambda}
\newcommand{\Lam}{\Lambda}
\newcommand{\rew}{\rightarrow}
\newcommand{\rewsh}{\mathord{\rightarrow}}

\newcommand{\lx}{\lam x}
\newcommand{\ly}{\lam y}

\newcommand{\lxrs}{(\lx.r)s}
\newcommand{\rsx}{r\subx{s}}

\newcommand{\est}{(\,)}
\newcommand{\al}{\alpha}

\newcommand{\gam}{\gamma}

\newcommand{\Gam}{\Gamma}
\newcommand{\Gami}{\Gam_i}

\newcommand{\Del}{\Delta}
\newcommand{\Deli}{\Del_i}

\newcommand{\sig}{\sigma}
\newcommand{\sigi}{\sig_i}

\newcommand{\epsi}{\varepsilon}
\newcommand{\fom}{f^{\omega}}

\newcommand{\om}{\omega}
\newcommand{\Om}{\Omega}

\newcommand{\arob}{\symbol{64}}

\newcommand{\hearts}{\heartsuit}
\newcommand{\clubs}{\clubsuit}
\newcommand{\sheart}{{\ssst \hearts}}
\newcommand{\sclub}{{\ssst \clubs}}

\newcommand{\TermV}{\mathscr{V}}
\newcommand{\TypeV}{\mathscr{O}}

\newcommand{\bind}{\lam}            

\newcommand{\child}[1]{\ttchild(#1)}
\newcommand{\childp}[1]{\ttchild'(#1)}

\newcommand{\ju}[2]{#1 \vdash #2}

\newcommand{\msig}{\mult{\sig}}

\newcommand{\msigi}{\mult{\sigi}_{\iI}}

\newcommand{\msigpip}{\mult{\sig'_i}_{\iI'}}

\newcommand{\sSk}{(S_k)_{\kK}}
\newcommand{\sSpkp}{(S'_k)_{\kK'}}

\newcommand{\sSksh}{(\!S_k\!)_{\!k\!\in\!K}} 

\newcommand{\cu}{\mathtt{Y}}
\newcommand{\cuf}{\cu_f}

\newcommand{\Pex}{P_{\mathtt{ex}}}
\newcommand{\Piex}{\Pi_{\mathtt{ex}}}
\newcommand{\Sex}{S_{\mathtt{ex}}}

\newcommand{\loplus}{\!\,\mathord{^{\oplus}}\!}
\newcommand{\roplus}{\!\!\,^{\oplus}}
\newcommand{\lomin}{\!\,\mathord{^{\omin}}}
\newcommand{\romin}{\!\!\!\,\mathord{^{\omin}}\,}
\newcommand{\loast}{\!\,^{\oast}}

\newcommand{\Axl}[1]{\mathtt{Ax}_{\lam}(#1)}
\newcommand{\Trl}[1]{\mathtt{Tr}_{\lam}(#1)}

\newcommand{\Rep}{\mathtt{Rep}} 

\newcommand{\dom}[1]{\ttdom(#1)}
\newcommand{\codom}[1]{\ttcodom(#1)}

\newcommand{\Ax}{\ttAx}
\newcommand{\Hd}{\ttTg}
\newcommand{\Tl}{\ttSc}
\newcommand{\Pol}[1]{\ttPol(#1)}

\newcommand{\pos}[1]{\ttpos(#1)}
\newcommand{\posp}[1]{\ttpos'(#1)}
\newcommand{\Axa}{\Ax_{a}}
\newcommand{\AxP}{\Ax^P}

\newcommand{\Axal}{\Ax_\al}
\newcommand{\Axpalp}{\Ax'_{\al'}}

\newcommand{\tr}[1]{ \tttr(#1)}  
\newcommand{\trP}[1]{\tttr^P(#1)}  
\newcommand{\trp}[1]{ \tttr'(#1)}

\DeclareMathOperator{\Rt}{\mathtt{Rt}}

\newcommand{\Res}{\mathtt{Res}}
\DeclareMathOperator{\ResL}{\mathtt{ResL}}
\DeclareMathOperator{\ResR}{\mathtt{ResR}}

\newcommand{\out}{\mathtt{out}}

\newcommand{\Lamzzu}{\Lam^{001}}

\newcommand{\subb}[2]{[#1/#2]}
\newcommand{\subx}[1]{\subb{#1}{x}}

\newcommand{\rewa}{\!\rightarrow_{\! \mathtt{asc}}\!} 
\newcommand{\rewp}{\!\rightarrow_{\! \mathtt{pi}\!}}

\newcommand{\wer}{\leftarrow}

\newcommand{\iden}{\equiv}
\newcommand{\idena}{\iden_{\asc}}

\newcommand{\denot}[1]{{[} \! {[} #1 {]}  \! {]} }

\newcommand{\p}{\mathtt{p}}

\DeclareMathOperator{\RedTr}{\mathtt{RedTr}}

\newcommand{\Arg}{\mathtt{Arg}}

\newcommand{\ax}{\ttax}
\newcommand{\abs}{\ttabs}
\newcommand{\app}{\ttapp}

\newcommand{\apph}{\ttapp_\tth}

\DeclareMathOperator{\rbp}{rbp}

\newcommand{\Rst}{\ttR}
\newcommand{\Lst}{\ttL}

\newcommand{\ArgTr}{\mathtt{ArgTr}} 
\newcommand{\AxTr}{\mathtt{AxTr}}

\DeclareMathOperator{\QRes}{\mathtt{QRes}}
\DeclareMathOperator{\ResI}{\mathtt{ResI}}
\DeclareMathOperator{\Inter}{\mathtt{Inter}}

\newcommand{\red}[1]{\textcolor{red}{#1}}
\newcommand{\blue}[1]{\textcolor{blue}{#1}}

\newcommand{\purple}[1]{\textcolor{purple}{#1}}
\definecolor{greenone}{RGB}{0,200,140}

\definecolor{darkred}{RGB}{220,00,00}

\newcommand{\rewb}[1]{\stackrel{ #1}{\rightarrow}}

\newcommand{\ad}[1]{\mathtt{ad}(#1)}

\newcommand{\trck}[1]{\,\red{[#1]}}

\newcommand{\posPr}[1]{~ \red{\langle\!#1 \!\rangle}}

\newcommand{\rs}{\mathtt{rs}}

\newcommand{\rstra}{\rstr{a}}
\newcommand{\rstrb}{\rstr{b}}
\newcommand{\trb}{t\rstrb}

\newcommand{\tprb}{t'\rstrb}
\newcommand{\tra}{t\rstra}

\newcommand{\tri}{\rhd}

\newcommand{\asc}{\mathtt{asc}}

\newcommand{\Asc}[1]{\mathtt{Asc}(#1)}

\newcommand{\twoarewa}{(2\cdot  \tv)  \rew \tv}

\newcommand{\juGtt}{\ju{\Gam}{t:\tau}}

\newcommand{\juxkT}{\ju{x:(k\cdot T)}{x:T}}

\newcommand{\juCtt}{\ju{C}{t:T}}

\newcommand{\ttsupp}{\mathtt{supp}}
\newcommand{\ttbisupp}{\mathtt{bisupp}}
\newcommand{\supp}[1]{\ttsupp(#1)}
\newcommand{\bisupp}[1]{\ttbisupp(#1)}
\newcommand{\bisuppR}[1]{\pierre{\ttbisupp_{\ttL}}(#1)}
\newcommand{\bisuppL}[1]{\pierre{\ttbisupp_{\ttR}}(#1)}
\newcommand{\ttbisuppL}{\ttbisupp_\Lst}

\newcommand{\suppat}[1]{\ttsupp_{\arob}(#1)}

\newcommand{\ttlab}{\mathtt{lab}}
\newcommand{\lab}[1]{\ttlab(#1)}
\newcommand{\relab}{\theta} 
\newcommand{\Relab}{\Theta} 
\newcommand{\Relabarg}{\Relab_{\ttarg}} 
\newcommand{\Relabout}{\Relab_{\ttout}}

\newcommand{\Relabtr}{\Relab_{\tttr}}
\newcommand{\Relabref}{\Relab}
\newcommand{\Relabpout}{\Relab'_{\ttout}}
\newcommand{\Relabptr}{\Relab'_{\tttr}}
\newcommand{\Relabpref}{\Relab'}

\newcommand{\Apheart}{A'_{\hearts}}
\newcommand{\Apclub}{A'_{\clubs}}
\newcommand{\Apout}{A'_{\ttout}}
\newcommand{\Aparob}{A'_{\arob}}

\newcommand{\cL}{c_{\Lst}}
\newcommand{\cR}{c_{\Rst}}
\newcommand{\eL}{\tte_{\Lst}}
\newcommand{\eR}{\tte_{\Rst}}
\newcommand{\kL}{k_{\Lst}}
\newcommand{\kR}{k_{\Rst}}
\newcommand{\rL}{\ttr_{\Lst}}
\newcommand{\rR}{\ttr_{\Rst}}

\newcommand{\thL}{\theta_{\Lst}}
\newcommand{\thR}{\theta_{\Rst}}

\newcommand{\ttmut}{\mathtt{mut}}
\newcommand{\ttsuppmut}{\ttsupp_{\ttmut}}
\newcommand{\suppmut}[1]{\ttsuppmut(#1)}
\newcommand{\ttbisuppmut}{\ttbisupp_{\ttmut}}
\newcommand{\bisuppmut}[1]{\ttbisuppmut(#1)}

\newcommand{\LstP}{\Lst^P}
\newcommand{\RstP}{\Rst^P}

\newcommand{\Psisu}{\Psi_{\ttsupp}}
\newcommand{\Psitr}{\Psi_{\tttr}}

\newcommand{\Edg}[1]{\ttE(#1)}
\newcommand{\ttThr}{\mathtt{Thr}}
\newcommand{\Thrtr}[1]{\ThrE{#1}}
\newcommand{\ThrE}[1]{\ttThr_{\ttE}(#1)}

\newcommand{\Thr}[1]{\ttThr(#1)}
     
\newcommand{\ttthr}{\mathtt{thr}}
\newcommand{\thr}[1]{\ttthr(#1)}    
\newcommand{\thrp}[1]{\ttthr'(#1)}  

\newcommand{\ttVal}{\mathtt{Val}}
\newcommand{\Val}[1]{\ttVal(#1)}

\newcommand{\Resref}{\Res^{\ttref}}

\newcommand{\refe}[1]{\ttref(#1)}
\newcommand{\refp}[1]{\ttref'(#1)}

\newcommand{\lra}{\leftrightarrow}             
\newcommand{\lrab}[1]{\stackrel{#1}{\lra}}     

\newcommand{\rewc}[1]{\stackrel{#1}{\rew}}    

\newcommand{\rrewc}[1]{\tilde{\rew}_{#1}} 
\newcommand{\rrew}{\tilde{\rightarrow}} 

\newcommand{\werr}{\tilde{\wer}}
\newcommand{\werrc}[1]{\stackrel{#1}{\tilde{\wer}}}
\newcommand{\lrrew}{\tilde{\mathord{\leftrightarrow}}}
\newcommand{\lrrewc}[1]{\lrrew_{#1}}

\newcommand{\lsa}[1]{\stackrel{#1}{\mathord{\leftarrow}}} 
\newcommand{\rsa}[1]{\stackrel{#1}{\mathord{\rightarrow}}} 
\newcommand{\lrsa}[1]{\stackrel{#1}{\mathord{\leftrightarrow}}} 

\newcommand{\IM}{\mathcal{I}}
\newcommand{\IMi}{\IM_i}

\usepackage{verbatim}

\usepackage{tikz}

\usepackage[utf8]{inputenc}
\usepackage[T1]{fontenc}

\usepackage{wrapfig} 
\usepackage{xspace}

\usepackage{prooftree}

\usepackage{amsthm,amssymb,amsmath}

\theoremstyle{definition}
\renewcommand{\proofname}{Proof}
\theoremstyle{definition}

\usepackage{stmaryrd}

\newcommand{\trans}[3]{ 
  \begin{scope}[xshift=#1cm,yshift=#2cm]
     #3
  \end{scope}
}

\newcommand{\transh}[2]{ 
  \begin{scope}[xshift=#1cm]
     #2
  \end{scope}
}

\newcommand{\drawlabnode}[3]{  
     \trans{ #1 }{ #2 }{ 
       \draw (0,0) node {\small #3 };
       \draw (0,0) circle (0.23);
     }
}

\newcommand{\drawarob}[2]{   
  \trans{#1}{#2}{
       \draw (0,0) node {\small $\arob$};
       \draw (0,0) circle (0.23);
    }
}

\newcommand{\drawsmalltriin}[3]{   
  \trans{#1}{#2}{  
    \draw (0,0) --++ (0.35,0.65) --++ (-0.7,0) --++ (0.35,-0.65) ;
    \draw (0,0.46) node{\fnsz #3} ; 
  }
}

\newcommand{\drawlefttail}[2]{   
  \trans{#1}{#2}{
      \draw (0,0) --++ (-0.13,-0.18);
    }
  }

\newcommand{\blockunary}[3]{ 
  \drawlabnode{#1}{#2}{#3}
  \trans{#1}{#2}{
  \draw (0,0.23) --++ (0,0.44);
}
  }

\newcommand{\blocka}[2]{  
  \trans{#1}{#2}{  
    \drawarob{0}{0}
    \draw (-0.12,0.18) -- (-0.53,0.72); 
    \draw (0.12,0.18) -- (0.53,0.72);
  }
}

\newcommand{\paper}[1]{#1}
\newcommand{\techrep}[1]{\ignore{#1}}

\theoremstyle{plain}  
\newtheorem{proposition}{Proposition}
\newtheorem{property}{Property}
\newtheorem{rmk}{Remark}
\newtheorem*{remark*}{Remark}
\newtheorem*{lemma*}{Lemma}
\newtheorem*{theorem*}{Theorem}
\newtheorem{observation}{Observation}

\begin{document}


\title{Representing permutations without permutations, or the expressive power of sequence types}
\titlerunning{The Surjective Collapse of Sequential Types} 

\author{Pierre Vial}
\affil{Inria (LS2N CNRS)} 
\authorrunning{Pierre Vial} 

\Copyright{Pierre Vial}

\maketitle

\begin{abstract}
 Recent works by Asada, Ong and Tsukada have championed a rigid description of resources. Whereas in non-rigid paradigms (\eg standard Taylor expansion or non-idempotent intersection types), bags of resources are multisets and invariant under permutation, in the rigid ones, permutations must be processed explicitly and can be allowed or disallowed. Rigidity enables a fine-grained control of reduction paths and their effects on \eg typing derivations. We previously introduced a very constrained coinductive type system (system $\ttS$) in which permutation is completely disallowed. One may wonder in what extent the absence of permutations causes a loss of expressivity
\wrt reduction paths, compared to a usual multiset framework or a rigid one allowing permutations. We answer this question in the most general case \ie coinductive type grammars without validity conditions. 
 Our main contribution is to prove that not only every non-idempotent derivation can be represented by a rigid, permutation-free derivation, but also that any dynamic behavior may be captured in this way. In other words, we prove that system $\ttS$ 
 has full expressive power over multiset/permutation-inclusive intersection.
\ignore{
We present a type system (that we refer to as System $\ttS$) that features sequences as intersection types, in contrast with Gardner/de Carvalho's System $\scrR$, that uses multisets to represent intersection. Systems $\ttS$ and $\scrR$ are both linear and forbid contraction/duplication.
We explain why any reduction choice in System $\scrR$ (\ie any way to perform Subject Reduction) may be encoded in System $\ttS$ by resorting to suitable isomorphisms of labelled trees, that we call interfaces. 
The main theorem of this article is the following: every reduction choice may be encoded by a trivial interface. This means that not only the syntactical equality between labelled trees used in System $\ttS$ can represent the far more relaxed equality used in System $\scrR$ (relying upon nested permutations in multisets), but that it is also enough to capture convoluted reduction choices.

We actually consider coinductive type grammars. Contrary to most inductive cases, typability does not ensure here any form of normalizability. From the technical perspective, the chief contribution of this paper is the introduction of a method, based upon the study of a suitable first order theory, that does not rely on any notion of productive reduction, which is pivotal to prove the Representation Theorem.
We also briefly discuss the applications of this work and of coinductive intersection type grammars on the finitary or infinitary $\lam$-calculus and its semantics. 

}

\end{abstract}
\newcommand{\frakS}{\mathfrak{S}}

\section{Determinism, Rigidity and Reduction Paths}

\subsection{Resources for the $\lam$-Calculus}

The attempts at giving a quantitative account of resource consumption by functional programs originates from Girard's Linear Logic\;\cite{Girard87}. In his wake, several works were proposed to capture the same ideas. In this introduction, we chiefly focus on two of them:
\begin{itemize}
\item The \textbf{resource calculus}, notably by Ehrhard-Regnier\;\cite{EhrhardR08}: instead of having applications of the form $t\,u$ (one function, one argument),  applications of the  of the form $t\,\mult{u_1,\ldots,u_n}$ (one function, several argument), where $\mult{u_1,\ldots,u_n}$ is a \textbf{multiset} (also called a \textbf{bag}). Substitution of a variable $x$ with a bag $\mult{u_1,\ldots,u_n}$ in a term $t$ is \textbf{linearly} processed: it occurs only if there are exactly $n$ free occurrences of $x$ in $t$. In that case, each occurrence of $x$ is replaced by one $u_i$. Since the elements of a multiset can be permuted, this operation is non-deterministic \eg if $n=2$ and $t=x\,\mult{x,y}$, then the considered substitution can output $u_1\,\mult{u_2,y}$ or $u_2\,\mult{u_1,x}$. The usual untyped $\lam$-calculus can be embedded into the resource calculus \via the \textbf{Taylor expansion}, which can be understood as a \textit{linearization} of $\lam$-terms. However, since one does not now how many times a function may use its argument, Taylor expansion $v\mapsto \tilde{v}$ represents an application $t\,u$ of the usual $\lam$-calculus by a formal series involving  $\tilde{t}\,\emul$, $\tilde{t}\,\mult{\tilde{u}}$, $\tilde{t}\,\mult{\tilde{u},\tilde{u}}$, $\tilde{t}\,\mult{\tilde{u},\tilde{u},\tilde{u}}$\ldots
 The semantics of the resource calculus is compatible with that of the $\lam$-calculus in that, the computation of the Taylor expansion and that of the \Bohm tree of a term commute (\textit{adequation}).
\item \textbf{Non-idempotent intersection types}, notably by Gardner\;\cite{Gardner} and de Carvalho~\cite{Carvalho07}: intersection type systems (i.t.s.), introduced by Coppo-Dezani\;\cite{CoppoD80}, feature a type construction $\wdg$ (intersection). In the non-idempotent setting, $\wdg$ can be seen as a free operator. Gardner-de Carvalho's original system $\calG$ features an explicit permutation rule  allowing to permute types in the context:
  $$\infer{\ju{\Gam;x:A_1\wdg \ldots \wdg A_n}{t:B}\hspace{1cm}\rho \in \frakS_n}{\ju{\Gam;x:A_{\rho(1)} \wdg \ldots \wdg A_{\rho(n)}}{t:B}}\hspace{1cm} \text{($\frakS_n$: permutations of $\set{1,\ldots,n}$)} $$
A later presentation of system $\calG$, that we call system $\scrRo$ here (see \S\;\ref{ss:system-R-lipics-sC} for details), represents intersection with multisets \ie the lists 
$A_1\wdg \ldots \wdg A_n$ are (inductively) collapsed on $\mult{A_1,\ldots,A_n}$.
 Since $\mult{A_1,\ldots,A_n}=\mult{A_{\rho(1)},\ldots,A_{\rho(n)}}$, the $\mathtt{perm}$-rule becomes obsolete and system $\scrRo$ is \textit{syntax-directed}, meaning that for a given judgment is the conclusion of a unique rule. Syntax-direction makes type system more readable and closer to the syntax of the $\lam$-calculus. Indeed, system $\scrRo$ features only three typing rules, corresponding to the three constructors of the $\lam$-calculus ($x$, $\lx$ or $\arob$).

 Because of non-idempotency and relevance (no weakening), system $\scrRo$ is \textit{linear}. Moreover, it characterizes head normalization: a term is $\scrRo$-typable iff it is head normalizing (HN). Variants of system $\scrRo$ give characterizations of weak, strong or weak head normalization (see ~\cite{BucciarelliKV17} for a survey). Most interestingly, non-idempotency makes the termination property of system $\scrRo$ (if $t$ is $\scrRo$-typable, then $t$ is HN) straightforward to prove\techrep{: 
 when a \textit{typed} redex is reduced inside a $\scrRo$-derivation $\Pi$, we get a new derivation $\Pi'$ (typing the reduced term) of strictly smaller size. Thus, typed reduction must stop at some point. Moreover, the same feature entail that system $\scrRo$ 
 gives upper bounds on the length of normalization strategies.}\paper{ and the size of $\scrRo$-derivations gives upper bounds on the length of normalizing reduction sequences (quantitativity).}\techrep{

   } From the semantic point of view, the typing judgments of $\scrRo$ correspond to to the points of the relational model~\cite{BucciarelliEM07}\techrep{ (suggesting the letter $\scrR$)}. 
\end{itemize}
Contrary to \cite{TsukadaAO17}, we are mainly interested in non-idempotent i.t.s. and not to the Taylor expansion, but they are related in that, a $\scrRo$-typing of a term $t$ naturally gives an element of the Taylor expansion of $t$ that does not reduce to the null sum.

\subsection{Getting Rigid}
\label{ss:getting-rigid}

We now discuss the different features of the resource framework alluded to in the previous section. As noted above, the linear substitution of a variable with the elements of a multiset bag is non-deterministic. For system $\scrRo$, this entails that subject reduction 
is \textit{non-deterministic} (see Fig.\;\ref{fig:non-deter-srRo-lipics-1}): we say that system $\scrRo$ gives rise to different \textbf{reduction choices}.

In \cite{VialLICS17}, we characterized a form of infinitary weak normalization by using non-idempotent types. This required using a \textbf{coinductive} (meaning infinitary) grammar of types. However, coinductive grammars give rise to unsound derivations \eg typings of $\Om:=\Del\,\Del$ (with $\Del=\lx.x\,x$), the auto-autoapplication. To recover soundness, coinductive type grammars must come together with a validity criterion. We called the one that we introduced in \cite{VialLICS17} \textit{approximability}.  It turns out (\S\;IV.E. of  \cite{VialLICS17}) that approximability \textit{cannot} be defined with multiset intersection. This led us to introduce system $\ttS$ (see \S\;\ref{ss:S-types-lipics-sC}), a \textit{rigid} variant of system $\scrRo$: 
multisets are coinductively replaced by \textit{sequences} of types \ie families of types annotated with integers. For instance, $(2\ct S,3\ct T,8\ct S)$ is a \textbf{sequence type} that features two occurrences of type $S$\ignore{ on \textbf{tracks} 2 and 8} and one occurrence of\ignore{ type} $T$\ignore{ on track $3$}. We have a ``disjoint union'' operator $\uplus$ for sequences \eg $(2\ct S, 3\ct T)\uplus (8\ct S)=(2\ct S,3\ct T,8\ct S)$, so that the occurrences of $S$ annotated with 2 can be identified on each side of the equality (\textbf{tracking}).  
In contrast, $\mult{\sig,\tau}+\mult{\sig}=\mult{\sig,\sig,\tau}$, but, in this equality, we have no way to relate one occurrence of $\sig$ in $\mult{\sig,\sig,\tau}$ to $\mult{\sig,\tau}$ rather than $\mult{\sig}$ and \textit{vice versa} (tracking is impossible with multisets). The lack of tracking is the main cause of non-determinism. In other words, one may say that intuitively, a framework is \textbf{rigid} when it enables tracking and determinism.
 Multiset equality can be seen as \textit{lax} whereas sequential equality is \textit{syntactical} (and thus, highly constraining). The syntactic nature of sequences entails that subject reduction is \textit{deterministic} in system $\ttS$ (Fig.\;\ref{fig:deter-srS-lipics-1}).

In their recent works~\cite{OngT12,TsukadaAO17}, Asada, Ong and Tsukada proposed another rigid framework, featuring  \textbf{rigid resource calculus} \techrep{(associated with rigid bags)} and a \textbf{rigid Taylor expansion}, also satisfying adequation. Their approach has several aims, notably giving a precise account of the weight of reduction paths in probabilistic programming:  rigidity makes it possible to enumerate reduction paths \textit{statically} (what is referred to as the \textit{compositional enumeration problem}).
This rigid calculus is interpreted in generalized species of structures, which are a special case of cartesian closed bicategories featuring non-trivial isomorphisms. %
 These isomorphisms refine the $\mathtt{perm}$-rule  and enable the encoding of convoluted reduction paths. Moreover, without $\mathtt{perm}$-rule or isomorphisms of types, subject reduction fails in the presence of list types.


 Giving back our attention to type theory, in a companion paper~\cite{compUnsoundArxiv}, we prove that every $\lam$-term is $\scrR$-typable, where $\scrR$ is the coinductive extension of $\scrRo$. This entails that every term has a \textit{non-empty} interpretation in the infinitary relational model. Our proof of this result, which uses tracking, can be only obtained by working in system $\ttS$ instead of $\scrR$  and then, noticing that \paper{sequences collapse on multisets}\techrep{every sequence-based derivation collapses on a multiset-based derivation} when tracks are erased (\eg $(2\ct S,3\ct T,8\ct S)$ \techrep{collapses }onto $\mult{S,T,S}$). However, this does not exclude the fact that there are $\scrR$-derivations (\ie points of the infinitary relational model) that are \textit{not} the collapse of any $\ttS$-derivation.

 In this paper, we compare the expressive powers of the rigid and the non-rigid infinitary intersection type systems (systems $\ttS$ and $\scrR$). In particular, we prove that actually, every $\scrR$-derivation is the collapse of a $\ttS$-derivation. Due to the \textit{syntactical} equality used in the application rule of  system $\ttS$ and the absence of productivity, this turns out to be a difficult question (see \S\;\ref{ss:collapse-and-hybrid-cons}). Note that we do not assume any validity criterion or guard condition: this allows us to consider the problem with the maximal level of generality (at least when comparing sequential and multiset intersection), so that our results may both have consequences in semantics, in infinitary $\lam$-calculus or in  the more concrete domain of verification (\eg for Grellois-Melli\`es's infinitary model of Linear Logic~\cite{GrelloisM15,GrelloisM15a}).

 The comparison between systems $\scrR$ and $\ttS$ is refined in the following sense: the determinism of reduction in system $\ttS$ can been seen as restrictive compared to system $\scrR$, since, in the former case, one cannot choose a reduction path (which is imposed by the structure of the derivation) in the absence of permutation rule/isomorphism of types.  So it raises the following question: can we build a rigid representative $P$ of an $\scrR$-derivation $\Pi$ \wrt any $\scrR$-reduction choice on could perform on $\Pi$? We also give a positive answer to this question.

 \subsection{Contributions and Challenges}

Our main contributions are thus proving that (1) every $\scrR$-derivation is the collapse of a $\ttS$-derivation (``surjectivity of the multiset collapse'') and that (2) any $\scrR$-reduction path can be directly encoded in a $\ttS$-derivation. This means that system $\ttS$ has full expressive power over system $\scrR$ and the infinitary relational model.

Once again, subject reduction is not satisfied with list intersection (\ie without permutation/isomorphisms). On another hand, isomorphisms of types enables the description of reduction paths, as shown by \cite{TsukadaAO17}. Our article shows that sequence types provide the best of both worlds (multiset and list intersections): one can also endow sequence types with isomorphisms (we need to consider an intermediary system, namely system $\ttSh$, from \S\;\ref{ss:collapse-and-hybrid-cons} on, to conduct our proofs) but it is not necessary. Not only subject reduction is satisfied for sequential intersection without burdening derivations with isomorphisms of types, but point 2 above proves that no reduction path (taken individually) is lost without permutation.

The most difficult point is the first one \ie establishing every $\scrR$-derivation has a \textit{trivial} $\ttS$-representative. In a finitary/productive framework (for the notion of productivity, see \eg \cite{EndrullisHHP015}), this could be possible by studying first the derivations typing a (partial) normal form (for which representation is usually easily ensured
), and then proceeding by subject expansion (see \S\;\ref{ss:collapse-and-hybrid-cons}). However, as already noted, typability in system $\scrR$ does not imply any kind of normalization (every term can be $\scrR$-typed), even an infinitary one. The main technical contribution of this paper, that is also used in \cite{compUnsoundArxiv}, is to bypass this problem by resorting to a method inspired by first order model theory, along with some finite reduction strategy (see \S\;\ref{ss:bro-consum-lipics-sC} and \ref{ss:syntactic-polarity-and-collapsing-strat} for some high-level input). By lack of space, we can only sketch the main stages of the proof. All our results are valid for the infinitary calculus $\Lamzzu$~\cite{KennawayKSV97}. See the webpage of the author or Chapter 11 and 12 of \cite{VialPhd} for complete proofs and details.


\newcommand{\trckblck}[1]{\,[#1]}
\newcommand{\ddsh}{\mathord{:}}
\newcommand{\jush}[2]{#1\mathord{\vdash}#2}
\newcommand{\typeone}{\mathord{(8\ct \tv,9\ct \tv)\rewsh \tv'}}
\newcommand{\typetwo}{\mathord{(8\ct \tv)\rewsh \typeone}}
\newcommand{\typeonered}{\mathord{(\red{8}\ct \tv,\red{9}\ct \tv)\rewsh \tv'}}
\newcommand{\typetwored}{\mathord{(8\ct \tv)\rewsh \typeonered}}
\newcommand{\typeoneblue}{\mathord{(\blue{8}\ct \tv,\blue{9}\ct \tv)\rewsh \tv'}}
\newcommand{\typetwoblue}{\mathord{(8\ct \tv)\rewsh \typeoneblue}}
\newcommand{\typeonebis}{\mathord{(2\ct \tv,7\ct \tv)\rewsh \tv'}}
\newcommand{\typeonebispurple}{\mathord{(\purple{2}\ct \tv,\purple{7}\ct \tv)\rewsh \tv'}}
\newcommand{\typetwobis}{\mathord{(3\ct \tv)\rewsh \typeonebis}}
\newcommand{\typetwobispurple}{\mathord{(3\ct \tv)\rewsh \typeonebispurple}}
\newcommand{\orange}[1]{\textcolor{orange}{#1}}

\section{Non-Idempotent Intersection and Rigid Paradigms}
\label{s:rigid-trucs}

\subsection{Multiset Intersection}
\label{ss:system-R-lipics-sC}

Let us recall Gardner and de Carvalho's non-idempotent i.t.s.~\cite{Carvalho07,Gardner}  whose set of types is inductively defined by:
\begin{center}
  $\sigma,\,\tau :: = \tv\in \TypeV~ |~ \msigi \rew \tau $
\end{center}
We call $\IM:=\msigi$ a \textbf{multiset (intersection) type}. \paper{Intersection $\wdg$ corresponds to the multiset-theoretic sum.}\techrep{
The intersection operator $\inter$ is the multiset-theoretic sum: $\inter_\iI \IMi=+_\iI \IMi$ (\ie $\inter_{\iI} \mult{\sig^i_j}_\jJi  := +_{\iI} \mult{\sig^i_j}_{\jJi}$).} We assume $I$ to be finite, the empty multiset type is denoted by $\emul$. 

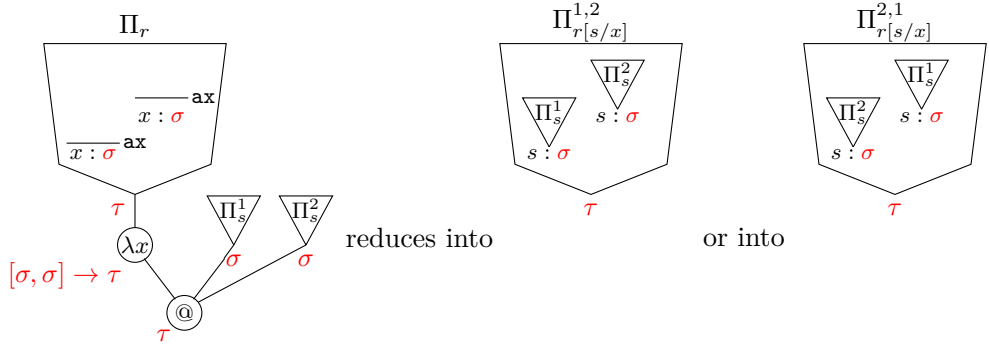
\begin{figure}

  \begin{tikzpicture}

{    
    \draw (-0.65,3.8) node {$\Pi_r$};
    \draw (-0.65,1.57) --++ (1,0.4) --++ (0.2,1.6) --++ (-2.4,0) --++ (0.2,-1.6) -- cycle;
    \draw (-0.65,1.57) node [below left] {$\red{\tau}$};
    \draw (-1.55,2.25) --++ (0.7,0);
    \draw (-0.65,2.25) node  {\fnsz $\ax$};
    \draw (-1.2,2.1) node  {\fnsz $x:\red{\sig}$};
    \draw (-0.65,2.85) --++ (0.7,0);
    \draw (0.25,2.85) node  {\fnsz $\ax$};
    \draw (-0.3,2.6) node {\fnsz $x:\red{\sig}$};
    
    \blockunary{-0.65}{0.9}{$\lx$}
    \draw (-0.7,0.8) node [below left] {$\red{\mult{\sig,\sig}\rew \tau}$};
    \blocka{0}{0}
    \draw (-0.05,-0.1) node [below left] {$\red{\tau}$};

\drawsmalltriin{0.65}{0.9}{$\Pi^1_s$}
\drawlefttail{0.65}{0.9}
\draw (0.65,0.9) node [below] {$\red{\sig}$};

\drawsmalltriin{1.6}{0.9}{$\Pi^2_s$}
\draw (0.18,0.14) -- (1.6,0.9);
\draw (1.6,0.9) node [below] {$\red{\sig}$};
}
\draw (2,1) node [right] {\large reduces into};

\transh{6}{    
    \draw (-0.65,3.85) node {$\Pi^{1,2}_{\rsx}$};
    \draw (-0.65,1.57) --++ (1,0.4) --++ (0.2,1.6) --++ (-2.4,0) --++ (0.2,-1.6) -- cycle;
    \draw (-0.65,1.57) node [below] {$\red{\tau}$};
    \drawsmalltriin{-1.2}{2.2}{$\Pi^1_s$};
    \draw (-1.2,2.1) node  {\fnsz $s:\red{\sig}$};
    \draw (-0.3,2.6) node {\fnsz $s:\red{\sig}$};
    \drawsmalltriin{-0.3}{2.7}{$\Pi^2_s$}; 
}

\draw (6.7,1) node [right] {\large or into};

\transh{10}{    
    \draw (-0.65,3.85) node {$\Pi^{2,1}_{\rsx}$};
    \draw (-0.65,1.57) --++ (1,0.4) --++ (0.2,1.6) --++ (-2.4,0) --++ (0.2,-1.6) -- cycle;
    \draw (-0.65,1.57) node [below] {$\red{\tau}$};
    \drawsmalltriin{-1.2}{2.2}{$\Pi^2_s$};
    \draw (-1.2,2.1) node  {\fnsz $s:\red{\sig}$};
    \draw (-0.3,2.6) node {\fnsz $s:\red{\sig}$};
    \drawsmalltriin{-0.3}{2.7}{$\Pi^1_s$}; 
}

  \end{tikzpicture}
  \vspace{-0.2cm}
  \caption{Non-Determinism of Subject Reduction (Multiset Intersection)}
  \label{fig:non-deter-srRo-lipics-1}
\vspace{-0.5cm}
\end{figure}

An $\scrRo$-\textbf{context} ($\Gam$ or $\Del$) is a \textit{total} function from $\TermV$ (the set of term variables) to the set of multiset types. The \textbf{domain of $\Gam$} is given by $\set{x\,|\,\Gam(x)\neq \emul}$. The intersection of contexts $+_{\iI} \Gami$ is defined point-wise. We write $\Gam;\Del$ instead of $\Gam + \Del$ when $\dom{\Gam}\cap \dom{\Del}=\eset$. Given a multiset type $\msigi$, we write $x:\msigi$ for the context $\Gam$ s.t. $\Gam(x)=\msigi$ and $\Gam(y)=\emul$ for all $y\neq x$.
 A $\scrRo$-\textbf{judgment} is a triple $\ju{\Gam}{t:\sig}$ where $\Gam$ is a $\scrRo$-context, $t$ a term and $\sig$ a $\scrRo$-type.
 The set of $\scrRo$-derivations is defined by the following rules:
 \begin{center}
  $
  \begin{array}[t]{c}\infer[\ax]{ \phd}{ \ju{x:\mult{\tau} }{x:\tau}}  \hspace{1.4cm} \infer[\abs]{ \ju{ \Gam;x:\msigi}{t:\tau}}{\ju{\Gam}{\lx.t:\msigi\rew \tau}} \\[3ex]
    \infer[\app]{\ju{\Gam}{t:\msigi\rew \tau}\hspace{0.3cm} (\ju{\Deli }{u:\sigi})_{\iI} }{\ju{\Gam+(+_{\iI} \Deli )}{t\,u:\tau}}
    \end{array}
$
\end{center}
 The typing judgments of system $\scrRo$  correspond to the points of the relational model $\scrM_{\mathtt{rel}}$ from ~\cite{BucciarelliEM07} in that, for any term $t$, $\set{\tri_{\scrRo} \juGtt\;|\; \Gam\;\text{context},\tau\;\text{type}}=\denot{t}_{\mathtt{rel}}$, where the left-hand side corresponds to the set of derivable judgments of system $\scrRo$ typing $t$ and the right-hand side, to the denotation of $t$ in $\scrM_{\mathtt{rel}}$.  Here is an example of $\scrRo$-derivation:
 \vspace{-0.3cm}
 \begin{center}
    {$\Piex=\infer[\app]{\infer[\app]{\infer[\ax]{\phd}{\ju{\ldots}{x:\mult{o,o',o}\rew o'}  }\hspace{0.2cm}
\infer[\ax]{\phd}{\ju{\ldots}{x:o}} \hspace{0.1cm}  
\infer[\ax]{\phd}{\ju{\ldots}{x:o'}} \hspace{0.1cm}
\infer[\ax]{\phd}{\ju{\ldots}{x:o}} 
    }{\ju{x:\mult{o',\mult{o,o',o}\rew o',o,o}}{xx:o'}}}{
  \ju{}{\lx.xx:\mult{o',\mult{o,o',o}\rew o',o,o} \rew o'}
    }$
    }
 \end{center}
 We write $\Pi \tri_{\scrRo} \ju{\Gamma}{t: \tau}$ to mean that the $\scrRo$-derivation $\Pi$ concludes with the judgment $\ju{\Gamma}{t:\tau}$ and $\tri \ju{\Gamma}{t: \tau}$ to mean that $\ju{\Gamma}{t:\tau}$ is derivable.  No weakening is allowed (\textit{relevance}) \eg $\lambda x.x$ (resp. $\lx.y$) can be typed with a type of the form $\mult{\tau}\rew \tau$ (resp. $\emul\rew \tau$), but \textit{not} with $\mult{\tau,\sig}\rew \tau$ (resp. $\mult{\tau}\rew \tau$). 
System $\scrRo$ enjoys both \textbf{subject reduction} and \textbf{expansion}, meaning that types are invariant under (anti)reduction (if $t\rew t'$, then $\tri \juGtt$ iff $\tri \ju{\Gam}{t':\tau}$). Moreover, a term is head normalizing iff it is $\scrRo$-typable.

The $\app$-rule of $\scrRo$ is based upon multiset equality and can be restated as follows:\\[-2ex]
 $$\infer{\ju{\Gam}{t:\msigi\rew \tau} \hspace{0.7cm} (\ju{\Del_i}{u:\sig'_i})_{\iI'} \hspace{0.7cm} \msigi=\msigpip  }{\ju{\Gam + (+_{\iI'} \Del_i)}{t\,u:\tau}} \app$$

Now, system $\scrRo$ has a coinductive version that we call system $\scrR$. For that, we just interpret the rules of the type grammar coinductively: in system $\scrR$, $\msigi$ can be of infinite cardinality ($I$ is assumed to be countable). Moreover, there may be an infinite number of nestings \eg  in the $\scrR$-type $\phi_\tv$ satisfying $\phi_\tv=\mult{\phi_\tv} \rew \tv$ for a given $\tv \in \TypeV$. Note that $\phi_\tv$ allows us to type $\Om$ with $\tv$: using $\ax$-rule concluding with $\ju{x:\mult{\phi_\tv}}{x:\phi_\tv}$, we easily obtain $\ju{x:\mult{\phi_\tv}_\om}{x\,x:\tv}$, then $\ju{}{\Del:\phi_\tv}$ and finally $\ju{}{\Om:\tv}$ ($\mult{\sig}_{\om}$ contains infinitely many occurrences of $\sig$). We will sketch a more rigorous definition of the set of $\scrR$-types in \S\;\ref{ss:S-types-lipics-sC}.

As claimed in \S\;\ref{ss:getting-rigid}, subject reduction is not deterministic in system $\scrRo$: in Fig.\;\ref{fig:non-deter-srRo-lipics-1}, a redex $\lxrs$ is typed with $\tau$ and in $r$, $x$ is assigned \textit{twice} the type $\sig$ in two different places. Then, to obtain a derivation $\Pi'$ typing the reduct $\rsx$, we just replace each axiom rule by an argument derivation $\Pi^i_s\ (i=1,2)$ concluding with $s:\sig$. There are two possibilities (\textbf{reduction choices}) represented on the right-hand side of the figure.

\subsection{Tracks and Labelled Trees}
\label{ss:tracks-sC}

\paper{
Let us now present the formalism that we use for labelled trees.  
Let $\bbN^*$ be the set of finite words on $\bbN^*$, the operator $\cdot$ denotes concatenation, $\epsi$ the empty word and $\leqs$ the prefix order \eg $2\cdot 1\cdot 3\cdot 7\in \bbN^*$, $2\cdot 1\leqs 2\cdot 1\cdot 3\cdot 7$.
 As a letter of $a\in \bbN^*$, a natural number is called a \textbf{track} and, for reasons to appear, a track $k\geqs 2$ is called an \textbf{argument track} (argument tracks are \textit{mutable},  \S\;\ref{ss:track-threads-and-representability}).
 The \textbf{collapse} $\ovl{k}$ of a track $k$ is defined by $\ovl{k}=\min(k,\,2)$. This notation is extended letter-wise on $\bbN^*$  \eg $\ovl{0 \ct 5 \ct 1 \ct 3\ct 2}=0\ct 2 \ct 1\ct 2\ct 2$. The \textbf{support of term} is defined by induction as expected: $\supp{x}=\set{\epsi}$, $\supp{\lx.t}=\set{\epsi}\cup 0\cdot \supp{t}$ and $\supp{t\,u}=\set{\epsi}\cup 1\cdot \supp{t} \cup 2\cdot \supp{u}$. If $a\in \bbN^*$ and $\ovla\in \supp{t}$, we denote by $\tra$ the subterm of $t$ rooted at position $\ovla$ whereas $t(a)$ is the constructor ($\arob$, $x$ or $\lx$) of $t$ at position $a$ \eg $t\rstr{0}=y\,x$ and $t(0\cdot 1)=y$ with $t=\lx.y\,x$. The \textbf{applicative depth} $\ad{a}$ of $a$ is the number of argument tracks that it contains  \eg $\ad{0\cdot 3 \cdot 2\cdot 1\cdot 1} =2$ and $\ad{0\cdot 1 \cdot 0\cdot 0\cdot 1}=0$.

}


A \textbf{tree} $A$ of $\bbN^*$ is a non-empty subset of $\bbN^*$ that is downward-closed for the prefix order ($a\leqslant a'\in A$ implies $a\in A$). The support of a term is a tree.
A \textbf{forest} is a set  of the form $A\setminus \set{\epsi}$ for some tree $A$ such that $0,\,1 \notin A$. 
Formally, a \textbf{labelled tree} $T$ (resp. \textbf{labelled forest} $F$) is a function to a set $\Sigma$, whose domain, called its \textbf{support} $\supp{T}$ (resp. $\supp{F}$), is a tree (resp. a forest). If $F$ is a labelled forest, the set of \textbf{roots} of $F$ is defined by $\Rt(F)=\supp{F}\cap \bbN$.
If $U=T$ or $U=F$, then $U\rstr{a}$ is the function defined on $\set{a_0\in \bbN^*\,|\, a\ct a_0\in \supp{U}}$ and $U\rstr{a}(a_0)=U(a\ct a_0)$. If $U$ is a labelled tree (resp. forest and $a \neq \epsi$), then $U\rstr{a}$ is a labelled tree.

\begin{figure}
\begin{center}
\begin{tikzpicture}

  \draw (0,0) node {$\rew$} ;
  \draw (0,0) circle (0.18) ;

  \draw (0.13,0.13) -- (0.57,0.57) ;
  \draw (0.45,0.25) node {$\tiny \red{1}$ };
  
  \draw (0.7,0.7) node {$\tv_1$} ;
  \draw (0.7,0.7) circle (0.18) ;
  
  \draw (-0.13,0.13) -- (-0.57,0.57) ;
  \draw (-0.45,0.25) node {$\tiny \red{4}$};
  \draw (-0.7,0.7) node {$\rew$} ;
  \draw (-0.7,0.7) circle (0.18) ;

  \draw (-0.57,0.83) -- (-0.13,1.27);
  \draw (-0.25,0.95) node { $\tiny \red{1} $ } ;
  \draw (0,1.4) node {$\tv_2$} ;
  \draw (0,1.4) circle (0.18);

  \draw (-0.83,0.83) -- (-1.27,1.27) ;
  \draw (-1.15,0.95) node {$\tiny \red{3}$};
  \draw (-1.4,1.4) node {$\tv_1$} ;
  \draw (-1.4,1.4) circle (0.18) ;

  \draw (-2.8,1.4) circle (0.18) ;
  \draw (-2.8,1.4) node {$\tv_3$} ;
  \draw (-2.2 ,0.95)  node {$\tiny \red{8}$};
  \draw (-0.87,0.73) -- (-2.7,1.26) ;

  \draw (-1.7,0.7) node {$\tv_2$};
  \draw (-1.7,0.7) circle (0.18);
  \draw (-1.68,0.53) -- (-0.15,0);
  \draw (-1.3,0.25) node {$\tiny \red{8}$};

  \draw (-3,-0.2) node [below right] {
     \parbox{5.2cm}{\small  $T_1=(8\ct \tv_2,4\ct(8\ct \tv_3,3\ct \tv_1)\rew \tv_2)\rew \tv_1$}};

\trans{7}{0}{
   \draw (0,0) node {$\rew$} ;
  \draw (0,0) circle (0.18) ;

  \draw (0.13,0.13) -- (0.57,0.57) ;
  \draw (0.45,0.25) node {$\tiny \red{1}$ };
  
  \draw (0.7,0.7) node {$\tv_1$} ;
  \draw (0.7,0.7) circle (0.18) ;

\trans{-1}{0}{

  \draw (-0.7,0.7) node {$\rew$} ;
  \draw (-0.7,0.7) circle (0.18) ;

  \draw (-0.57,0.83) -- (-0.13,1.27);
  \draw (-0.25,0.95) node { $\tiny \red{1} $ } ;
  \draw (0,1.4) node {$\tv_2$} ;
  \draw (0,1.4) circle (0.18);

  \draw (-0.83,0.83) -- (-1.27,1.27) ;
  \draw (-1.15,0.95) node {$\tiny \red{2}$};
  \draw (-1.4,1.4) node {$\tv_3$} ;
  \draw (-1.4,1.4) circle (0.18) ;

  \draw (-2.8,1.4) circle (0.18) ;
  \draw (-2.8,1.4) node {$\tv_1$} ;
  \draw (-2.2 ,0.95)  node {$\tiny \red{7}$};
  \draw (-0.87,0.73) -- (-2.7,1.26) ;}

 \draw (-1.68,0.53) -- (-0.15,0);
 \draw (-1.3,0.25) node {$\tiny \red{5}$};

%

  \draw (-0.13,0.13) -- (-0.57,0.57) ;
  \draw (-0.45,0.25) node {$\tiny \red{3}$};
  \draw (-0.7,0.7) node {$\tv_2$} ;
  \draw (-0.7,0.7) circle (0.18) ;
  
  \draw (-3,-0.2) node [below right] {
    \parbox{5.2cm}{\small  $T_2=
      (5\ct(7\ct \tv_1,2\ct \tv_3)\rew \tv_2,3\ct \tv_2)\rew \tv_1$}};
}

\end{tikzpicture}\\[-8ex]
\end{center}
\vspace{-0.5cm}
\caption{01-Isomorphic Labelled Trees ($\ttS$-Types)}
\label{fig:01-iso}
\vspace{-0.5cm}
\end{figure}

\begin{definition}
Let $U_1$ and $U_2$ be two (labelled) trees or forests.\\
A \textbf{01-isomorphism} $\phi$ from $U_1$ to 
$U_2$ is a bijection from $\supp{U_1}$ to $\supp {U_2}$ such that:
\begin{itemize}
\item $\phi$ is monotonic for the prefix order and preserves length.
\item If $a\ct k\in \supp{ U_1}$ with $a\in \mathbb{N}^*$ and $k=0,\, 1$, then $\phi(a\ct k)=\phi(a)\ct k$. 
\item In the labelled case: for all $a\in \supp {U_1}$, $U_2(\phi(a))=U_1(a)$.
\end{itemize}
\end{definition}

We write $U_1\equiv U_2$ when $U_1$ and $U_2$ are 01-isomorphic. In Fig.\;\ref{fig:01-iso}, we see two isomorphic labelled trees 
(which are two $\ttS$-types, \S\;\ref{ss:S-types-lipics-sC})
\wrt $\phi$ defined by $\phi(\epsi)=\epsi$, $\phi(1)=1$, $\phi(4)=5$, $\phi(4\ct 1)=5\ct 1$, $\phi(4\ct 3)=5\ct 7$, $\phi(4\ct 8)=5\ct 2$, $\phi(8)=3$.
\techrep{
If $U$ is a (labelled  or not) tree or forest and $\phi$ is a monotonic, length-preserving injection from $\supp{U}$ to $\mathbb{N}^*$ s.t., in the labelled case, $\phi(a
\ct k)=\phi(a)\ct k$ whenever $k=0,\,1$, we write $\phi(U)$ for the
\textit{unique} (labelled) tree or forest s.t. $\supp {\phi(U)}=\set{
\phi(a)\,|~ a\in \supp {U}}$ and $\phi(U)(a')=U(\phi^{-1}(a'))$. In that case, $\phi(U)\equiv U$. Such a function $\phi$ then is called a \textbf{01-resetting of $U$}.

  \begin{remark}
Alternatively, we can define $U_1\equiv U_2$ for types and sequence types by coinduction, without reference to 01-stable isomorphisms:
\begin{itemize}
\item $\tv \equiv \tv$
\item $\sSk\equiv \sSpkp$ if there is a bijection $\rho:K\rew K'$ such that, for all $\kK$, $S_k\equiv S'_{\rho(k)}$.
\item $\sSk\rew T\equiv \sSpkp\rew T'$  if $\sSk\equiv \sSpkp$ and $T\equiv T'$.\\
\end{itemize}
\end{remark}

  Actually, 01-resetting a labelled (or not) tree or forest $U$ consists in assigning new track values to the edges labelled with $k$ tracks $\geqslant 2$. In that case, we may (a bit abusively) use the position $a\cdot k$ of $U$ to stand for the edge from $a$ to $a\cdot k$ and we set $\lab{a\cdot k}=k$: the number $k$ is indeed the track that labels the edge from $a$ to $a\cdot k$.

This motivates the notion of \textit{mutable support} and of \textit{01-relabelling}:

  \begin{definition}
    \label{def:01-relab}
    Let $U$ be a labelled (or not) tree or forest.
    \begin{itemize}
    \item The \textbf{mutable support} $\suppmut{U}=\set{a\cdot k\in \supp{U}\,|\, a\in \bbN^*,\,k\in \Nmzo}$. We understand $\supp{U}$ as $U$ in the unlabelled case.
      \item A \textbf{01-relabelling} of $U$ is a function $\relab$ from $\suppmut{U}$ to $\Nmzo$ such that, for all $a\in \bbN,\, k_1,k_2\in \Nmzo$ such that $a\cdot k_1,a\cdot k_2\in \suppmut{U}$, $k_1\neq k_2$ implies $\relab(a\cdot k_1)\neq \relab(a\cdot k_2)$.
    \end{itemize}  
  \end{definition}

  Intuitively, $\suppmut{U}$ is the set of edges of $U$ whose track value may be changed and $\relab(a\cdot k)$ is the new track (given by $\relab$) of the edge between position $a$ and the mutable position $a\cdot k$. Indeed, from a 01-relabelling $\relab$ of $U$, we can define a 01-resetting of $U$ denoted $\phi^{\relab}$ by induction:
  \begin{itemize}
  \item $\phi^{\relab}(\epsi)=\epsi$.
  \item $\phi^{\relab}(a\cdot k)=\phi^{\relab}(a)\cdot k$ if $k=0,1$
  \item $\phi^{\relab}(a\cdot k)=\phi^{\relab}(a)\cdot \relab(a\cdot k)$ if $k\in \Nmzo$.
  \end{itemize}
Thanks to Definition\;\ref{def:01-relab}, the function $\phi^{\relab}$ is indeed a 01-resetting of $U$.\\
  
\noindent \textit{Example:} The transformation from $T_1$ to $T_2$ in Fig.\;\ref{fig:01-iso} can be seen as the relabelling $\relab$ defined on $\suppmut{T_1}=\set{4,4\cdot 3,4\cdot 8,8}$ by $\relab(4)=5$ (the edge labelled 4 receives the new label 5), $\relab(4)=5$, $\relab(4\cdot 3)=7$, $\relab(4\cdot 8)=2$, $\relab(8)=3$.\\

} Let $F$ and $F'$ be
two 01-isomorphic (labelled) forests. A \textbf{root isomorphism} is a function $\rho$ from $\Rt(F)$ to $\Rt(F')$ such that, for all $k\in \Rt(F)$, $F\rstr{k}\equiv F'\rstr{\rho(k)}$ \ie a function from $\Rt(F)$ to $\Rt(F')$ that can be extended to a 01-isomorphism from $F$ to $F'$. Conversely, every isomorphism $\phi$ from $F$ to $F'$ induces a root isomorphism from $F$ to $F'$, that we denote $\Rt(\phi)$.

\techrep{

\subsection{Infinitary Terms}
\label{s:terms-short}

The main result of this article (Theorem\;\ref{th:rep}) is true for the $\lam$-calculus, but it actually also holds for an infinitary calculus referred to as  \cite{KennawayKSV97}. Let us then give a short presentation of this calculus.

Let $\Del_f=\lambda x.f(xx)$ and $\cuf=\Del_f\Del_f$. Notice that $\cuf\rew f(\cuf)$, so $\cuf \rew^{k} f^k(\cuf)$. Intuitively, if $k$ tends toward $\infty$, the redex disappear and  we get $\cuf \rew^{\infty} \fom$ where $\fom$ is the (infinite) term $f(f(f(\ldots)))$, satisfying $\fom=f(\fom)$ and containing a rightward infinite branch. \\

\noindent Let $\TermV$ be a countable set of term variables.
The set of terms $\Lambda^{111}$ is defined \textit{coinductively}:\\[-3ex]
\begin{center}
$t,\,u~ ::=~ x~ \|~ \lambda x.t~ \|~ t\,u $\\[-3ex]
\end{center}
  We can still consider terms of $\Lambda^{111}$ as (possibly infinite) labelled trees.
 The abstraction $\lambda x$ binds $x$ in $t$ and $\alpha$-equivalence can be defined properly\;\cite{KennawayKSV97}.


The relation $t \rewb{b} t'$ is defined by induction on $b\in 
\set{0,\,1,\,2}^*$: $(\lambda x.r)s\rewb{\epsi} r[s/x]$, $\lambda 
x.t\rewb{0\cdot b} \lambda x.t'$ if $t\rewb{b}t'$, $t_1t_2\rewb{1\cdot 
b} t'_1t_2$ if $t_1\rewb{b} t_1,~ t_1t_2 \rewb{2\cdot b} t_1t_2 '$ if 
$t_2\rewb{b} t'_2$. We define \textbf{$\beta$-reduction} by $\rew  =\cup_{b\in \set{0,\,1,\,2}^* } \rewb{b}$.

If $b=(b_i)_{i\in \bbN}$ is a word of \textit{infinite} length, we extend $\ad{b}$ as $|\set{i \in \mathbb{N}~|~
b_i\geqslant 2}|$ and we say that $b$ is an infinite branch of $t\in
\Lambda^{111}$ if, for all $n\in \mathbb{N},~ b^{<n}\in \supp{t}$, where $b^{<n}= b_0\ct b_1\ct \ldots \ct b_n$.  The calculus
$\Lambda^{001}$ is the set of terms $t\in \Lambda^{111}$ such that,
for every infinite branch $b$ of $t$, $\ad{b}=\infty$ . Thus, for
\textbf{001-terms} (terms of $\Lam^{001}$), infinity is allowed, provided we visit arguments infinitely many times.


For instance, $\fom$ is in $\Lambda^{001}$, because $2^{\omega}$ (the infinite word holding only 2) is its unique infinite branch and $\ad{2^{\omega}}=\infty$. However, the term $t\in \Lambda^{111}$ coinductively defined by $u=u\,f$ (where $f$ is a term variable) is not a $001$-term: it also has a unique infinite branch, $1^{\omega}$ but $\ad{1^{\omega}}=0$.

\begin{remark}
\label{rmk:001-trees-and-induction}
Let us call a tree $A\subset \bbN^*$ such that, for all infinite branch $a$ of $A$, $\ad{a}=\infty$, a \textbf{001-tree}. Let $\scrP$ be a predicate on such a tree $A$. In order to prove for all $a\in A,\, \scrP(a)$, we can reason by \textbf{001-induction}: we prove that $\scrP(a)$ for all leaf $a$ of $A$ and then, for all $a\in \bbN^*$ such $a\cdot 0$ or $a\cdot 1$ is in $A$, we prove that $\scrP(a\cdot 0)$ or $\scrP(a\cdot 1)$ implies $\scrP(a)$.
\end{remark}

}

\techrep{\pierre{MAL PLACE}For instance, 
$(7\ct \tv_1,3\ct \tv_2,2\ct \tv_1)\rew \tv$ is represented by:\\[-3ex]
\begin{center}
\begin{tikzpicture}
\draw (2,0) node {\small $\rew$};
\draw (2,0) circle (0.18);

\draw (3,1)  node{\small $\tv$};
\draw (3,1) circle (0.18);
\draw (2.87,0.87) -- (2.13,0.13) ;
\draw (2.7,0.55) node {\small \red{$1$} };

\draw (1.3,1) node{\small $\tv_1$} ;
\draw (1.3,1) circle (0.18);
\draw (1.91,0.161) -- (1.38,0.85) ;
\draw (1.8,0.55) node {\small \red{$2$} };

\draw (0.6,1) node{\small $\tv_2$} ;
\draw (0.6,1) circle (0.18);
\draw (1.86,0.13) -- (0.68,0.84);
\draw (1.4,0.55) node {\small \red{$3$} };

\draw (-0.4,1) node{\small $\tv_1$} ;
\draw (-0.4,1) circle (0.18);
\draw (1.83,0.08) -- (-0.3,0.85);
\draw (0.1,0.55) node {\small \red{$7$} };
\end{tikzpicture}
\end{center}
Thus, for types, track 1 is dedicated to the target of an arrow (see also Figure\;\ref{fig:01-iso}).  We write $\est$ for the sequence type whose support is empty and a sequence type of the form $(k\cdot T)$ is called a \textit{singleton sequence type}.}

\subsection{Rigid Types}
\label{ss:S-types-lipics-sC}

Formally, the set of $\ttS$-types is defined coinductively by:
\newcommand{\skSk}{(k\cdot S_k)_{\kK}}
\begin{center}$ 
\begin{array}{cll} 
  T,\,S_k & ::= &  \tv ~ \|~ F \rightarrow T\\
  F & ::=& \skSk\ (K\subeq \Nmzo)
\end{array}$\\[-2ex] 
\end{center}
The empty sequence type is denoted $\est$. 
The set of top-level tracks of a sequence type is called its set of \textbf{roots} and we write \eg $\Rt(F)=\set{2,5,8}$ when $F=(2\cdot \tv,5\cdot \tv',8\cdot \tv)$. Note that the disjoint union operator can lead to \textbf{track conflict} \eg, if $F_1=(2\cdot \tv,3\cdot \tv')$ and  $F_2=(3\cdot \tv',8\cdot \tv)$, the union $F_1\uplus F_2$ is not defined,  since $\Rt(F_1)\cap \Rt(F_2)=\set{3}\neq \eset$. Thus, $\uplus$ is not total, but it is associative and commutative.

The \textbf{support of a type} (resp. \textbf{a sequence type}), which is a tree of $\bbN^*$ (resp. a forest),  is defined by mutual coinduction: $\supp{\tv}=\set{\epsi},~ \supp{F\rew T}=\set{\epsi}\cup \supp{F}\cup 1\cdot \supp{T}$ and $\supp{(T_k)_{k\in K}}= \cup_{k\in K} k\cdot \supp{T_k}$. For instance, $\supp{\Sex}=\set{\epsi,1,2,7}$. Thus, types (\resp sequence types) can naturally be  seen as labelled trees (\resp labelled forests) in the sense of \S\;\ref{ss:tracks-sC}. See Fig.\;\ref{fig:01-iso} for two examples. 

When we forget about tracks, a finite rigid type collapses on a $\scrRo$-type \eg $(7\cdot \tv_1,3\cdot \tv_2,2\cdot \tv_1)\rew \tv$,  $(9\cdot \tv_2,7\cdot \tv_1,6\cdot \tv_1)\rew \tv$ or $(7\cdot \tv_2,3\cdot \tv_1,2\cdot \tv_1)\rew \tv$ all  collapse on $\mult{\tv_1,\tv_2,\tv_1}\rew \tv$.
We can now define $\scrR$-types (\S\;\ref{ss:system-R-lipics-sC}) as collapses of rigid types:

\begin{definition}
  \label{def:type-iso-sC}
Let $U_1$ and $U_2$ be two (sequence) types. A \textbf{(sequence) type isomorphism}  from $U_1$ to $U_2$ is a 01-isomorphism from $U_1$ to $U_2$
\end{definition}
For instance, in Fig.\;\ref{fig:01-iso}, both $T_1$ and $T_2$ collapse on $\mult{\tv_2,\mult{\tv_1,\tv_3}\rew \tv_2}\rew \tv_1$. 
A similar notion of isomorphism exists in \cite{TsukadaAO17}, III.A., which is presented coinductively.
\techrep{
  \begin{itemize}
\item If $S=F\rightarrow T$, where $F$ is a sequence type and $T$ a type, we write $\Tl(S)$ for the sequence type $F$ (the \textbf{source} of $S$) and $\Hd(S)$ for the type $T$ (the \textbf{target} of $S$).
\item If $\psi:\,F_1\longrightarrow F_2$ is a sequence type isomorphism and $\phi:\,T_1\rew T_2$ is a type isomorphism, then $\psi\rew \phi:\, (F_1\rew T_1)\rew (F_2\rew T_2)$ is the type isomorphism defined by: $(\psi\rew \phi)(k\cdot c)=\phi(k\cdot c)$ when $k\geqslant 2$ and $(\psi\rew \phi)(1\cdot c)=1\cdot \phi(c)$.  If $F$ is a sequence type, then $F\rew \phi$ denotes $\ttid_F\rew \phi$.
\item Conversely, if $\phi$ is a type isomorphism from $F_1\rew T_1$ to $ F_2\rew T_2$, then $\Hd(\phi)$ and $\Tl(\phi)$ are respectively the type isomorphism from $T_1$ to $T_2$ and the sequence type isomorphism from $F_1$ to $F_2$ induced by $\phi$. Thus, $\phi=\Tl(\phi)\rew \Hd(\phi)$ and $\Hd(\psi\rew \phi)=\phi$ and $\Tl(\psi\rew \phi)=\phi$. 
  \end{itemize}

\noindent Recalling the notion of 01-resetting and of 01-relabelling (Sec.\;\ref{s:tracks}), we define:
  
\begin{definition}
A 01-resetting (resp. 01-relabelling) whose domain is a rigid (sequence) type is called a \textbf{(sequence) type resetting} (resp. a \textbf{(sequence) type relabelling}).
\end{definition}
}


The set of $\scrR$-types can be formally defined as the quotient set of the set of $\ttS$-types by $\equiv$.
Remember that $\scrR$ allows infinite multiset types and infinite nesting of types inside multisets. Multisets types of system $\scrR$ are the $\equiv$-equivalence classes of sequence types. System $\scrR$ is defined as $\scrRo$ (\S\;\;\ref{ss:system-R-lipics-sC}), except that we use types of $\scrR$ instead of just finite types of $\scrRo$\techrep{ and we allow coinductive derivations to type the infinite terms of $\Lam^{001}$}. 
A countable version of the binary operator $+$ can be easily defined. 
\techrep{See \S\;\;\furef{s:hybrid-cons-techrep} for more details.} 

\subsection{Trivial Rigid Derivations}
\label{ss:system-S-lipics-sC}

An $\ttS$-context $C$ (or $D$) is a total function from $\TermV$ to the set of $\ttS$ types. The operator $\uplus$ is extended point-wise.
An $\ttS$-judgment is a triple $\juCtt$, where $C$, $t$ and $T$ are respectively an $\ttS$-context, a term and $T$ an $\ttS$-type. 
A \textbf{sequence judgment} is a sequence of judgments  $(k\cdot (C_k\vdash t:T_k))_{k\in K}$ with $K\subeq \Nmzo$, often just written $(\ju{C_j}{t:T})_{\kK}$. For instance, if $5 \in K$, then the judgment on track 5 is $\ju{C_5}{t:S_5}$.

The set of $\ttS$-derivations is defined inductively by:
\begin{center}
 \vspace{-0.1cm}
$\infer[\ax]{\phd}{\juxkT} \hspace{1.4cm}
  \infer[\abs]{\ju{C;x:\sSk}{t:T}}{\ju{C}{\lx.t:\sSk \rew T}}$\\[0.2cm]
  $\infer[\app]{\ju{C}{t:\sSk\rew T} \hspace{0.7cm} (\ju{D_k}{u:S_k})_{\kK}  }
        {\ju{C\uplus(\uplus_{\kK} D_k)}{t\,u:T}}$
        \end{center}
The $\app$-rule can be applied only if there are no track conflicts in the context $C\uplus(\uplus_{\kK} D_k)$. In an $\ax$-rule concluding with $\juxkT$, the track $k$ is called the \textbf{axiom track} of this axiom rule. We refer to \S\;III and IV of \cite{VialLICS17} for additional examples and figures for all what concerns the basics of system $\ttS$ (bipositions, bisupports, residuals\ldots). 
In order to gain space, we write $\ju{k}{x:T}$ (with $k\geqs 2$, $x\in \TermV$, $T$ $\ttS$-type) instead of $\ju{x:(k\cdot T)}{x:T}$ in $\ax$-rules.
We also indicate the track of argument derivations between brackets \eg $\ju{2}{x:\tv'}\trck{3}$ means that judgment $\ju{x:(2\ct \tv')}{x:\tv'}$ is on track 3. Consider the $\ttS$-derivation $\Pex$:
          {\small
  $$\Pex=\infer[\abs]{\infer[\app]{\infer[\ax]{\phd}{\ju{4}{x:(8\ct o,3\ct o',2\ct o)\rew o' } }\hspace{0.4cm}
      \infer[\ax]{\phd}{\ju{9}{x: \tv}}\hspace{-0.56cm}\trck{2} \hspace{0.2cm}  
\infer[\ax]{\phd}{\ju{2}{x: \tv'}}\hspace{-0.56cm}\trck{3} \hspace{0.2cm}
\infer[\ax]{\phd}{\ju{5}{x:\tv}}\hspace{-0.56cm}\trck{8} 
    }{\ju{x:(2\ct \tv',4\ct (8\ct \tv,3\ct \tv',2\ct \tv)\rew \tv',5\ct \tv,9\ct \tv)}{xx:\tv'}}}{
  \ju{}{\lx.xx:(2\ct \tv',4\ct (8\ct \tv,3\ct \tv',2\ct \tv)\rew \tv',5\ct \tv,9\ct \tv)\rew \tv'}
              }$$
}
          In the $\ax$-rule concluding with $\ju{x:(5\ct o)}{x:o}$, the axiom track is 5. As hinted at before, $\ttS$-derivations collapse on $\scrR$-derivations\techrep{ \pierre{(this collapse will be more precisely described in Sec.\;\ref{s:hybrid-cons-techrep})}} \eg $\Pex$  collapses on $\Piex$ (from \S\;\ref{ss:system-R-lipics-sC}). Moreover:
          

\begin{proposition}
Systems $\scrR$ and $\ttS$  enjoy  subject reduction and expansion.
  \end{proposition}

\techrep{
  We write $P\rhd_{\ttS} \juCtt$ to mean that the conclusion of the trivial derivation $P$ is the judgment $\juCtt$. }
The $\app$-rule of system $\ttS$ can be restated as follows:
 $$\infer[\app]{\ju{C}{t:\sSk\rew T} \hspace{0.7cm} (\ju{D_k}{u:S'_k})_{\kK'}\hspace{0.7cm} \sSk=\sSpkp }
{\ju{C\uplus(\uplus_{\kK} D_k)}{t\,u:T}}$$
Thus, in system $\ttS$, the $\app$-rule requires that, for an application $t\,u$ to be typed, $\sSk$, the source of the arrow type given to $t$ must be \textit{(syntactically) equal} to the sequence type $\sSpkp$ given to $u$ (compare with \S\;\;\ref{ss:system-R-lipics-sC}). A consequence of the rigidity is that subject reduction is \textit{deterministic} in system $\ttS$, as hinted at by Fig.\;\ref{fig:deter-srS-lipics-1} (compare with Fig.\;\ref{fig:non-deter-srRo-lipics-1}).

Let us give $\ttS$ some high-level inputs on pointers in system $\ttS$ (again, see \S\;III. and IV. of \cite{VialLICS17} for complements). We define the support of a $\ttS$-derivation, and also the key notions of biposition and bisupport: the \textbf{support} of $P\rhd C\vdash t:T$ is defined inductively by $\supp{P}=\set{\epsi}$ if $P$ is an $\ax$-rule, $\supp{P}=\set{\epsi}\cup 0\cdot \supp{P_0}$ if $t=\lambda x.t_0$ and $P_0$ is the subderivation typing $t_0$, $\supp{P}=\{\epsi\}\cup 1\cdot \supp{P_1} \cup_{\kK} k\cdot \supp{P_k}$ if $t=t_1\,t_2,~ P_1$ is the left subderivation typing $t_1$ and $P_k$ the subderivation typing $t_2$ on track $k$. The $P_k$ ($k\in K$) are called \textbf{argument derivations}. For instance, $\supp{\Pex}=\set{\epsi,0,0\ct 1,0\ct 2,0\ct 3,0\ct 8}$, $P(0\cdot 1)=\ju{x:4\ct ((8\ct o,3\ct o',2\ct o)\rew o' )}{x:(8\ct o,3\ct o',2\ct o)\rew o' }$ and $P(0\ct 3)=\ju{x:(2\ct \tv')}{x:\tv'}$.
A (right) biposition of a $\ttS$-derivation $P$ is a pair $(a,c)$, where $a\in \supp{P}$ and $c\in \supp{T}$, where the judgment of $P$ at position $a$ is $\ju{C}{u:T}$. For this article, we just need to think of bipositions as pairs pointing to type symbols ($\tv\in \TypeV$ or $\rew$) or to labelled edges that are nested in  a given rigid derivation. We will do so when informally discussing Fig.\;\ref{fig:two-bro-threads-verbose}.

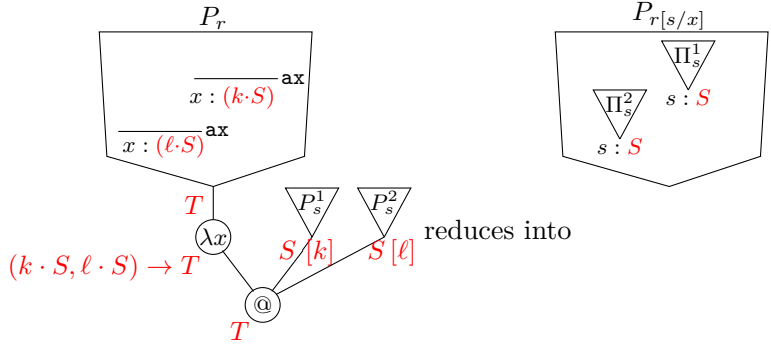
\begin{figure}
  \begin{tikzpicture}

{    
    \draw (-0.65,3.8) node {$P_r$};
    \draw (-0.65,1.57) --++ (1.2,0.4) --++ (0.1,1.65) --++ (-2.8,0) --++ (0.1,-1.65) -- cycle;
    \draw (-0.65,1.57) node [below left] {$\red{T}$};
    \draw (-1.9,2.3) --++ (1.1,0);
    \draw (-0.6,2.3) node  {\fnsz $\ax$};
    \draw (-1.3,2.1) node  {\fnsz $x:\red{(\ell \ct S)}$};
    \draw (-0.9,3) --++ (1.1,0);
    \draw (0.4,3) node  {\fnsz $\ax$};
    \draw (-0.4,2.75) node {\fnsz $x:\red{(k\ct S)}$};
    
    \blockunary{-0.65}{0.9}{$\lx$}
    \draw (-0.7,0.8) node [below left] {$\red{(k\cdot S,\ell \cdot S)\rew T}$};
    \blocka{0}{0}
    \draw (-0.05,-0.1) node [below left] {$\red{T}$};

\drawsmalltriin{0.65}{0.9}{$P^1_s$}
\drawlefttail{0.65}{0.9}
\draw (0.58,1) node [below] {$\red{S}\,\trck{k}$};

\drawsmalltriin{1.6}{0.9}{$P^2_s$}
\draw (0.18,0.14) -- (1.6,0.9);
\draw (1.7,1) node [below] {$\red{S}\trck{\ell}$};

}
\draw (2,1) node [right] {\large reduces into};


\transh{6}{    
    \draw (-0.65,3.83) node {$P_{\rsx}$};
  \draw (-0.65,1.57) --++ (1.2,0.4) --++ (0.1,1.65) --++ (-2.8,0) --++ (0.1,-1.65) -- cycle;
  \draw (-1.3,2.1) node  {\fnsz $s:\red{S}$};
  \drawsmalltriin{-1.3}{2.2}{$\Pi^2_s$};

  \draw (-0.4,2.75) node {\fnsz $s:\red{S}$};
    \drawsmalltriin{-0.4}{2.85}{$\Pi^1_s$};
}

  \end{tikzpicture}
  \vspace{-0.2cm}
  \caption{Deterministic Subject Reduction (Sequential Intersection, $\ttS$-case)}
  \label{fig:deter-srS-lipics-1}
  \vspace{-0.5cm}
\end{figure}



\techrep{
\subsection{Bipositions and Bisupport}
\label{s:bisupp}

Thanks to rigidity, we can identify and point to every part of a derivation, thus allowing to formulate many useful notions.


If $a\in \supp{P}$, then $a$ points to a judgment inside $P$ typing $t\rstr{\ovla}$. We write this judgment $\ju{\ttC(a)}{t\rstr{\ovl{a}}:\ttT(a)}$\paper{.}\techrep{: we say $a$ is an \textbf{(outer) position} of $P$.} The context function $\ttC$ and the type function $\ttT$ should be written $\ttC_P$ and $\ttT_P$ but we omit $P$. From now on, we shall also write $\tra$ and $t(a)$ instead of $t\rstr{\ovl{a}}$ and $t(\ovl{a})$.

In the example above, $\Pex(01)=\ju{x:(4\ct S)}{x:S}$, so
$\ttC(01)=x:(4\ct S)$ \ie $\ttC(01)(x)=(4\ct S)$. 
Since $S=(8\ct o,3\ct o',2\ct o)\rew o'$, we have $\ttC(01)(x)(4)=\rew$, $\ttC(01)(x)(43)=o',\ \ttT(01)(\epsi)=\rew$, $\ttT(01)(1)=o'$.
Likewise, $\Pex(03)=\ju{x:\!(2\ct \tv')}{\tv'}$, so that
 $\ttC(03)=x:(2\ct \tv')$ and $\ttT(03)=\tv'$. Thus, $\ttC(03)(x)(2)=\tv'$ and $\ttT(03)(\epsi)=o'$. We also have $\ttC(0)(x)=(2\ct \tv',4\ct (8\ct \tv,3\ct \tv',2\ct \tv)\rew \tv',5\ct \tv,9\ct \tv)$, so that $\ttC(0)(x)(2)=\tv'$ and $\ttC(0)(x)(42)=\tv$.\\

This example motivates the notion of \textbf{bipositions}: a biposition (metavariable $\p$) is a pointer into a type nested in a judgment of a derivation. A pair $(a,c)$ is a \textbf{right biposition} of $P$ if $a\in \supp{P}$ and $c\in \supp{\ttT(a)}$. A triple $(a,x,k\ct c)$ is a \textbf{left biposition} if $a\in \supp{P},\, x\in \scrV$ and $k\ct c\in \supp{\ttC(a)(x)}$.

\begin{definition}
The \textbf{bisupport} of a derivation $P$, written $\bisupp{P}$, is the set of its (right or left) bipositions.
\end{definition}

As it will turn out, we will take more time with right bipositions than with left ones.
We may consider a derivation as a function from its bisupport to the set $\TypeV\cup \set{\rew}$ and write now $P(a,c)$ for $\ttT(a)(c)$ and $P(a,x,k\ct c)$ for $\ttC(a)(x)(k\ct c)$.\\

A \textbf{bound occurrence} of a variable $x\in \TermV$ in a trivial derivation $P$ is a position $a\in \supp{P}$ such that $t(a)=x$ and there exists $a_*< a$ with $t(a_*)=\lx$. If $a$ is a bound occurrence of variable $x$ in $P$, the \textbf{binding position}  of this occurrence is the maximal $a_*\leqslant a$ such that $t(a_*)=\lx$. We then say that $a$ is bound by position $a_*$ in $t$ and  we write $a_*=\bind^P(b)$. 

If $a\in \supp{P}$ and $x\in \TermV$, we set $\AxP_a(x)=\set{a_0 \in \supp{P}\,|\,a\leqslant a_0,\,t(a)=x,\nexists a_0',\,a\leqslant a_0'\leqslant a_0,\,t(a'_0)=\lx }$ (positions of $\ax$-rules in $P$ above $a$ typing occurrences of $x$ that are not bound w.r.t. $a$).
If $a_0\in A$ is an axiom, we write $\trP{a_0}$ for its associated axiom track \eg $\trP{08}=5$ in the example above. Usually, $P$ is implicit and
we write only $\Axa(x)$ and $\tr{a}$.

\subsection{Quantitativity and Coinduction} 
\label{s:quant-deriv}

If $a\in A:=\supp{P}$ and $x\in \TermV$, we set $\AxP_a(x)=\set{a_0 \in A\,|\,a\leqslant a_0,\,t(a)=x,\nexists a_0',\,a\leqslant a_0'\leqslant a_0,\,t(a'_0)=\lx }$ (positions of $\ax$-rules in $P$ above $a$ typing occurrences of $x$ that are not bound w.r.t. $a$).
If $a_0\in A$ is an axiom, we write $\trP{a_0}$ for its associated axiom track \eg $\trP{08}=5$ in the example above. Usually, $P$ is implicit and
we write only $\Axa(x)$ and $\tr{a}$.

The presence of an infinite branch inside a derivation makes it possible that a type in a context is not created in an axiom rule. For instance, we set, for all $k\geqslant 2,~ j_k=\ju{f:(i\cdot \twoarewa)_{i\geqslant k},\,x:8\cdot\tv'}{\fom:o}$ and we coinductively define a family $(P_k)_{k\geqslant 2}$ of $\ttS$-derivations by:\\[-3ex]
  $$P_k=\infer{\infer*{\phd}{\ju{f:k\cdot\twoarewa}{f:\twoarewa}} \\
   P_{k+1}\tri j_{k+1}
  }{j_k=\ju{f:(i\cdot \twoarewa)_{i\geqslant k},\,x:8\cdot\tv'}{\fom:o}}  $$
We observe that the $P_k$ are indeed correct derivations of $\ttS$. However, notice that $x$ is typed (using track 8) whereas $x$ does not appear in the typed term $\fom$ and the part of the context assigned to $x$ cannot be traced back to any axiom rule typing $x$ with $o'$ (using axiom track 8).
This yields the notion of \textbf{quantitative derivation}, in which this does not happen:

\begin{definition}
A derivation $P$ is \textbf{quantitative} when, for all $a\in A$ and $x\in
\TermV$, $C(a)(x)=\uplus_{a'\in \AxP_a(x)} (\tr{a'}\cdot T(a'))$.
\end{definition}



Now, assume $P$ is quantitative. 
Then $\Rt(C(a)(x))=\set{\tr{a_0}\,|\, a_0\in \Axa(x)}$ and for all $a\in A,~x\in \TermV$ and $k\in \Rt(C(a)(x))$, we write $\pos{a,\,x,\, k}$ for the unique position $a'\in \Axa(x)$ such that $\tr{a'}=k$.

Actually, 
$\pos{a,\,x,\,k}$ can be defined by a downward induction on $a$ as follows:
\begin{itemize} 
\item If $a\in \Ax$, then actually $a\in \Ax(x)$ and $\tr{a}=k$ and we set 
  $\pos{a,\,x,\,k}=a$.
\item If $a:1\in A$, we set $\pos{a,\,x,\,k}=\pos{a:\ell,\,x,\,k}$, where $\ell$ is the necessarily unique (by typing  constraint) positive integer s.t. $k\in \Rt(C(a\cdot \ell)(x))$. 
\item If $a:0\in A$, we set $\pos{a,\, x,\,k}=\pos{a:0,\,x, \, k}$ 
\end{itemize}

}

\subsection{\techrep{More about $\ttS$-Representability}\paper{The Question of Representability}}
\label{ss:collapse-and-hybrid-cons}


\noindent From the introduction, we recall that we will give a positive answer to this question:\\

\noindent \textbf{Question:} Is every $\scrR$-derivation the collapse of an $\ttS$-derivation ?\\

\noindent \textbf{Attempt 1 (induction on proof structure):}
if we try to proceed by induction on the structure of a $\scrR$-derivation $\Pi$, we may be easily stuck. For instance, assume that
\begin{itemize}
  \item $\Pi=\infer{\Pi_{\lx.r}\tri \ju{\Gam}{\lx.r:\msigi\rew \tau} \hspace{1cm} (\Pi_i\tri \ju{\Deli}{s:\sigi})_{\iI}}{\ju{\Gam+(+_{\iI} \Deli)}{(\lx.r)s:\tau}}$.
  \item $\Pi_{\lx.r}$ and the $\Pi_i$ are all $\ttS$-representable.
\end{itemize}
The second assumption means that there are $\ttS$-derivations $P_{\lx.r}\tri \ju{C}{\lx.r:\sSk\rew T}$ and $(P_i\tri \ju{D_i}{s:S'_i})_{\iI'}$ that respectively collapse on $\Pi_{\lx.r}$ and the $\Pi_i$. But $P_{\lx.r}$ and the $P_i$ can be used to represent $\Pi$ only if (1) we can ensure that the $S_k$ and the $S'_i$ are syntactically equal (modulo some permutations) and (2) that we can avoid track conflicts.

For point (1), notice that, given a term $t$ that is $\scrR$-typable with type $\tau$, there may be some $\ttS$-types $T$ that collapse on $\tau$ such that $t$ is \textit{not} $\ttS$-typable with $T$ \eg if $\ovl{S}=\sig=\ovl{ S'}$, but $S\neq S'$, then $T:=(2\ct S)\rew S'$ collapses on $\tau:=\msig \rew \sig$ and $I:=\lx.x$ can be $\scrR$-typed with $\tau$ and $\ttS$-typed with $(2\ct S)\rew S$ or $(5\ct S')\rew S'$, but \textit{not} with $T$.

Thus, some typing constraints inside $r$ or $s$ could \textit{a priori} forbid that we can equalize the $S_k$ and the $S'_i$ since it is difficult to describe the $\ttS$-types of $r$ and $s$ that collapse on $\tau$ or $\sig_i$. \\

\noindent \textbf{Attempt 2 (expanding normal forms):} a fundamental method of intersection types is to type/study normal forms and then use subject expansion to extend the properties of the type system to every  term that can be reduced to a normal form. Can we do the same here ? 
As hinted as in \S\;\ref{ss:system-R-lipics-sC}, system $\scrR$ makes $\Om$ typable and thus, does not ensure any form of normalization\techrep{ (in a companion paper, we have proved that actually, every $\lam$-term is $\scrR$-typable)}: system $\scrR$ does not entail any form of normalization/productivity, and the method alluded to cannot work in the general case (it does though in the \textit{finite} one).\\

\noindent The remainder of this article is dedicated to answering positively to  the above question \ie proving:

\begin{theorem}[Representation]
\label{th:rep-R-S-naive}
Every $\scrR$-derivation is the collapse of an $\ttS$-derivation.
\end{theorem}

The discussion above
 suggests that system $\ttS$ is too constraining for the question of representability to be addressed directly and that we should relax the $\app$-rule. This motivates the definition of \textbf{hybrid derivations} by replacing, in system $\ttS$, the $\app$-rule by: 
{
\begin{center}
    $\infer{\ju{C}{t:\sSk \rew T} \hspace{0.8cm} (\ju{D_k}{u:S'_k})_{\kK'} \hspace{0.8cm} \sSk\equiv(S'_k)_{\kK'}  }{\ju{C\uplus(\uplus_{\kK} D_k)}{t\,u:T}} \apph$
\end{center}
  }
We call this modified system  $\ttSh$. We thus relax the syntactic equality of $\app$ into equality up to sequence type isomorphism (Def.\;\ref{def:type-iso-sC}).
If, in a $\ttSh$-derivation $P$, the above $\apph$-rule occurs at position $a$, one denotes $\sSk$ by $\LstP(a)$ (standing for ``Left'') and $\sSpkp$ by $\RstP(a)$ (``Right'').
 Then, the $\apph$-rule requires that \techrep{, for an application $t\,u$ to be typed, $\sSk$, the source of the arrow type of $t$ must be \textit{01-isomorphic} to the sequence type $\sSpkp$ given to $u$ \ie $\apph$ demands that} $\LstP(a)\equiv \RstP(a)$ \ie $\RstP(a)$ and $\LstP(a)$ must be isomorphic.   Thus, a hybrid derivation $P$ is trivial when for all $a\in \suppat{P}:=\set{a\in\supp{A}\,|\,t(a)=\arob}$ (\ie for all $\app$-rules of $P$), $\LstP(a)=\RstP(a)$.

Fig.\;\ref{fig:two-bro-threads-verbose} gives an example of $\ttSh$-derivation\techrep{, which will be discussed a lot}.  Note that all the contexts replaced by $\ldots$ can be computed from the given information \eg $\ju{x:(5\cdot \typetwo),y:(3\cdot \tv)}{x\,y:\typeone}$ and that the $\apph$-rules are correct. Indeed, $\Lst(\epsi)=(8\ct \tv,9\ct \tv)\rew \tv'\equiv (3\ct \tv,5\ct \tv)=\Rst(\epsi)$ and 
$\Lst(1)=(5\ct \typetwo) \equiv (6\ct \typetwobis) =\Rst(1)$.

 {
\techrep{
   \subsection{The Hybrid Construction}}
\label{s:hybrid-cons-techrep}


  \paper{Observe that rule $\apph$ is as constraining as the $\app$-rule of system $\scrR$ is. Indeed, $\sSk \equiv \sSpkp$ iff the multiset types $\ovl{\sSk}$ and $\ovl{\sSpkp}$ are equal, according to the end of \S\;\ref{ss:S-types-lipics-sC}. Thus, it is not a surprise that one easily proves:

    \begin{proposition}
      Let $\Pi$ be $\scrR$-derivation. Then, there \paper{is}\techrep{exists} a $\ttSh$-derivation $P$ collapsing on $\Pi$.
      \end{proposition}
  
  }




\section{Pseudo-Subject Reduction and Interfaces}
\label{s:reduction}

In this short section, we explore the effect of reduction on system $\ttSh$, which satisfies a relaxed form of subject reduction and expansion:

\begin{property}[Pseudo-Subject Reduction and Expansion]\mbox{}
  \label{prop:pseudo-sr-Sh}
  \begin{itemize}
  \item  If $t\rew t'$ and   $\tri_{\ttSh} \juCtt$, then $\tri_{\ttSh} \ju{C}{t':T'}$ for some $T'\equiv T$.
  \item If $t\rew t'$ and  $\tri_{\ttSh}\ju{C}{t':T'}$, then  $\tri_{\ttSh} \ju{C}{t:T}$ for some $T\equiv T'$.
  \end{itemize}  
\end{property}

Let us explain now how subject reduction is handled with hybrid derivations and why $T$ may be replaced by an isomorphic type $T'$ in $\ttSh$ (and \textit{vice versa}).\ignore{ In system $\ttSh$, we retrieve some determinism by considering \textbf{(root) interfaces} \ie sequence type isomorphisms that constrain \paper{how reduction is processed in a derivation.}}\techrep{how axiom leaves typing $x$, the variable or a redex, should be replaced by argument derivations typing while respecting the rules of $\ttSh$. This is explained in Sec.\;\ref{s:encoding-red-choices} from an high-level perspective. In Sec.\;\ref{s:rep-lemma-discussion}, we give a few intuitions on the means to capture sequences of reduction choices. From Sec.\;\ref{s:res-deriv} to 
  Sec.\;\ref{s:res-interface}, we formally prove that interfaces provide a sound way to produce a derivation typing a reduct (Corollary\;\ref{corol:pseudo-subj-red-Sh}). 
  In Sec.\;\ref{s:rep-lemma-proof}, we formally prove that every sequence of reduction choices can be built-in inside an interface.
}

\subsection{Encoding Reduction Choices with Interfaces}
\label{s:encoding-red-choices}

In this section, we explain Prop.\;\ref{prop:pseudo-sr-Sh} \ie reduction in system $\ttSh$. We assume that $\trb=\lxrs,\; t\rewb{b} t'$ (so that $t' \rstrb=\rsx$), each axiom rule concluding with $\ju{x:(k\ct S_k)}{x:S_k}$ will be replaced by a subderivation $P_{k'}\tri \ju{D_{k'}}{s:S'_{k'}}$ satisfying $S'_{k'}\equiv S_k$. However, the $\apph$-rule states that there exists an isomorphism but does not specify this isomorphism. 
This entails that there may be many ways to produce $P'$ typing $t'$ from $P$ typing $t$ \techrep{(this occurs whenever  $\sSk$ (and so, $\sSpkp$ as well) contains some 01-isomorphic types)} \ie in system $\ttSh$, there are also reduction choices, as in system $\scrR$ (\S\;\ref{ss:system-R-lipics-sC} and Fig.\;\ref{fig:non-deter-srRo-lipics-1}).

\begin{figure*}
  \begin{tikzpicture}


  \draw (-3.6,6.8) node [below right] {
\parbox{6cm}{\small
  \textbf{Assumptions:}  $\rho_a(2)=8$, $\rho_a(7)=5$ (so that $S_2\equiv S'_8$, $S_7\equiv S'_5$) 
    
  }
    };

 \draw (3.8,6.8) node [below right] {
\parbox{6cm}{\small
  \textbf{Comment:} since $\rho_a(2)=8$, $\rho_a(7)=5$, the argument subderivation $P_8$ (resp. $P_5$) on track 8 (resp. 5) will replace the axiom rule using track 2 (resp. 7).
  }
    };


  \draw (-0.7,1.03) --++ (0,0.38);
  \draw (-2.3,1.6) node [right] {\small $\ju{C;x:\sSk\!}{\!t\!:\!T} $};
  
  \draw (-0.8,1.25) node [right]{\small $\red{0}$};
  \draw (-2.6,0.9) node [right] {\small $\ju{C\!}{\!\lx.r\!:\!(S_k)_{\kK}\!\rew\! T}$};

  \draw (0.22,0.13)--(-0.57,0.67) ;
  \draw (-0.5,0.3) node [right] {\small $\red{1}$};
  \draw (-2.3,-0.1) node [right] {\small $\ju{C\uplus D_5\uplus D_8}{\lxrs}:T \posPr{a}$};
  
  \draw (-0.7,5.2) node{\Large $P_r$};
  \draw (-0.7,1.8) --++(2.8,4) --++ (-5.6,0) --++ (2.8,-4) ;
\techrep{
  \draw (1.3,5.3) node {\ssz $\heartsuit$};  
  \red{
    \draw (2.3,4.9) node {$\posPr{a\ct 1\ct 0\ct \al_{\scriptscriptstyle \heartsuit}}$};
    \draw [->,>=stealth] (1.8,5.1) --++(-0.35,0.2) ;
  }
}  


  \draw (-2.5,4.6) node [right] {\small $\ovl{\ju{7}{x:S_7}}$ }; 
  \draw [dotted] (-1.9,4.4) --++(0,-0.21);
  \red{
    \draw (-2.8,3.5) node {$ \posPr{a\ct 1\ct 0\ct a_7}$};
    \draw [>=stealth,->](-2.4,3.7) --++ (0.25,0.7);
  }

  \draw (-0.6,4.3) node [right] {\small $\ovl{\ju{2}{x: S_2}}$};
  \draw [dotted] (0.3,4.1) --++(0,-0.21);
\red{
    \draw (1.7,3.8) node {$ \posPr{a\ct 1\ct 0\ct a_2}$};
    \draw [>=stealth,->] (1,3.9) --++ (-0.3,0.2);
  }
  

  \draw (1.5,1.1) --++ (0.5,0.8) --++(-1,0) --++ (0.5,-0.8) ;
  \draw (1.5,1.5) node{\small $P_5$};
  \draw (0.9,0.9) node [right] {\small $\ju{D_5\!}{\!s\!:\!S'_5}$};
  \draw (0.4,0.13) -- (1.5,0.7);
  \draw (0.8,0.3) node [right] {\small $\red{5}$};

  \draw (3,1.1) --++ (0.5,0.8) --++(-1,0) --++ (0.5,-0.8) ;
  \draw (3,1.5) node{\small $P_8$};
  \draw (2.4,0.9) node [right] {\small $\ju{D_8\!}{\!s\!:\!S'_8}$};
  \draw (1.1,0.13) -- (2.8,0.7); 
  \draw (1.9,0.3) node [right] {\small $\red{8}$};

\techrep{
    \draw (2.86,1.77) node {\ssz $\clubs$}; 
  \draw (2.7,2.4) node {$\posPr{a\ct 8\ct \al_{\sclub}}$}; 
\red{  \draw [->,>=stealth] (2.7,2.2) --++ (0.15,-0.3) ;}
}
\draw (-2.4,-0.8) node [right] {\bf Subderivation typing the redex};

\trans{-2.9}{0}{


\draw (9.8,3.6) node{\Large $P_r$};
  \draw (9.8,0.3) --++(2.8,4) --++ (-5.6,0) --++ (2.8,-4) ;
\techrep{
  \draw (11.8,3.8) node {\ssz $\hearts$};
    \red{
    \draw (12.7,3.4) node {$\posPr{a\ct \al_{\sheart}}$};
    \draw [->,>=stealth] (12.3,3.6) --++(-0.35,0.2) ;
  }}

  \draw (7.7,-0.1) node [right] {\small $\ju{C\uplus  D_5\uplus D_8}{\rsx}:T' \posPr{a}$};

  \draw (8.6, 3.2) --++ (0.5,0.8) --++(-1,0) --++ (0.5,-0.8) ;
  \draw (8.6,3.6) node{\small $P_5$};
  \draw (8,3.1) node [right] {\small $\ju{D_5\!}{\!s\!:\!S'_5}$ }; 
  \draw [dotted] (8.6,2.8) --++(0,-0.21);
  \red{
    \draw (7.7,2.2) node {$ \posPr{a\ct a_7}$};
    \draw [>=stealth,->](8.1,2.4) --++ (0.1,0.5);
  }
  

  \draw (10.8,3) --++ (0.5,0.8) --++(-1,0) --++ (0.5,-0.8) ;
  \draw (10.8,3.4) node {\small $P_8$}; 
  \draw (9.9,2.8) node [right] {\small $\ju{D_8\!}{s\!:\! S'_8}$};
  \draw [dotted] (10.8,2.6) --++(0,-0.21);
\red{
    \draw (12.2,2.3) node {$ \posPr{a \ct a_2}$};
    \draw [>=stealth,->] (11.5,2.4) --++ (-0.3,0.2); 
  }      
\techrep{
  \draw (10.56,3.67) node {\ssz $\clubs$}; 
  \draw (10.5,4.6) node {$\posPr{a\ct 8\ct \al_{\sclub}}$} ;
  \red{  \draw [->,>=stealth] (10.4,4.4) --++ (0.15,-0.6) ;}
}  
 \draw (7.5,-0.8) node [right] {\bf Subderivation typing the reduct};
}

\end{tikzpicture}\\[-4ex]
  \caption{Subject Reduction\techrep{ and Residuals}}
  \label{fig:SR-residuals-interface}
  $\phd$\\[-7ex]
\end{figure*}

All this is illustrated by the left part of Fig.\;\ref{fig:SR-residuals-interface}: under the same hypotheses, we assume that $a\in \supp{P}$ is such that $\ovla=b$ (thus, $a$ is the position of a judgment typing the redex to be fired) and that there are exactly 2 $\ax$-rules typing $x$ above $a$, using axiom tracks 2 and 7. Notice that the $\ax$-rule typing $x$ on track 7 (assigning $S_7$) must be above $a\ct 1\ct 0$, so that its position is of the form $a\ct 1\ct 0\ct a_7$. Likewise for the other $\ax$-rule assigning $S_2$ on track 2. We omit  $\ax$-rules right-hand sides. We also
indicate the position of a judgment between angle brackets \eg $\posPr{a\ct 1\ct 0 \ct a_7}$ means that judgment $\ju{x\!:\!(7\ct S_7)}{x\!:\!S_7}$ is at position $a\ct 1\ct 0\ct a_7$.

By typing constraints, there must be two argument derivations typing $s$ with types isomorphic to $S_2$ and $S_7$. We assume that those two argument derivations are on track 5 and 8 and conclude with $\ju{D_k}{s:S'_k}\;(k=5,8)$ where \eg $S_2\equiv S'_8$ and $S_7\equiv S'_5$. 
If moreover, $S_2\equiv S_7$, then the $\ax$-rule \#2 can be replaced by $P_5$ as well as by $P_8$: there is a reduction choice. In all cases, the type of $\rsx$ may change \eg if $x:S_7$ corresponds to the head variable of $r$, then $S_7$ is replaced by an isomorphic type $S_5$ (or $S_8$), so $T$ may also be replaced by an isomorphic $T'$.

We now use \textbf{(root) type isomorphisms} (\S\;\ref{ss:tracks-sC}, Def.\;\ref{def:type-iso-sC}) retrieve some determinism and  represent particular reduction choices in system $\ttSh$. Let $P$ be a hybrid derivation:
\begin{itemize}
\item Let $a\in \suppat{P}$. A \textbf{root interface} (resp. an \textbf{interface}) at position $a$ is a root isomorphism (resp. a sequence type isomorphism) from $\Lst^P(a)$ to $\Rst^P(a)$.
\item Let $b\in \supp{t}$ such that $t(b)=\arob$. 
  A \textbf{total (root) interface} at position $b$ is a family of (root) interfaces for all $a\in \supp{P}$ s.t. $\ovla=b$.
\item A \textbf{total interface} is the datum for an interface $\phi_a$ for all $a\in \suppat{P}$.
\end{itemize}
For all $a\in \suppat{P}$, we write $\Inter^P(a)$ for the set of interfaces at position $a$ in $P$.\techrep{ Notice that a root interface interface is a function (between root tracks) that can be extended to an interface isomorphism and conversely, every interface isomorphism $\phi$.

} 
An \textbf{operable derivation} is a hybrid derivation
endowed with a total interface.\techrep{
  If $P$ is an operable derivation whose interface is $(\phi_a)_{a\in \suppat{P}}$, we usually only write $\rho_a$ for $\Rt(\phi_a)$ (so that $\rho_a$ is a root interface) and we set $\bisuppL{P}=\set{(a\cdot 1,k\cdot c)\in \bisupp{P}\,|\,k\in \Nmzo}$ and $\bisuppR{P}=\set{(a\cdot k,c)\in \bisupp{P}\,|\,k \in \Nmzo}$. For all $\p=(a\cdot 1,k\cdot c)\in \bisuppL{P}$, we just write $\phi(\p)$ for $(a\cdot k',c')$ with $k'\in \Nmzo,\; c'\in \bbN^*$ and $k'\cdot c'=\phi_a(k\cdot c)$.
  }\\

Assuming that $S_2\equiv S_7$ in Fig.\;\ref{fig:SR-residuals-interface}, we have seen above that the two $\ax$-rules typing $x$ could be indifferently replaced by $P_5$ or $P_8$.
But if $a$ is endowed with the root interface $\rho_a$ s.t. $\rho_a(2)=8$ and $\rho_a(7)=5$, then the $\ax$-rule typing $x$ on track 2 \textit{must} be replaced by $P_8$ and the other one by $P_5$, as on the right part of the figure.\techrep{

  \begin{remark}
Still if $S_2\equiv S_7$, the root interface $\rho_a$ defined by $\rho_a(2)=5$ and $\rho_a(7)=8$ is licit, but it would produce another derivation typing $\rsx$, obtained from the right part of Fig.\;\ref{fig:SR-residuals-interface} by swapping the two inner triangles.
    \end{remark}
}

Compare this process with Fig.\;\ref{fig:deter-srS-lipics-1}: implicitly, a trivial derivation $P$ is endowed with a trivial interface ($\phi_a$ is the \textit{identity} from $\Lst^P(a)$ to $\Rst^P(a)$ for all $a\in \suppat{P}$).
%
%
\ignore{
\begin{remark}
  We can now see the advantages of sequences over multisets and lists: sequences enable both syntax-direction and rigidity, whereas multisets lack of rigidity (\S\;\ref{ss:getting-rigid}). Moreover, only sequences give rise to a trivial rigid system (without isomorphisms or permutations), whereas subject reduction cannot hold with list types without a permutation rule or type isomorphisms as in \cite{TsukadaAO17}.
\end{remark}
}


\subsection{\paper{Representation Lemma}\techrep{Residuation}}
\label{s:rep-lemma-discussion}

We now explain why \textit{operable} derivations capture every $\scrR$-reduction path.
As seen in \S\;\ref{s:encoding-red-choices}, a total \textit{root} interface$(\rho_a)_{\ovla=b}$ (we write $\ovla=b$ for  $a\in \supp{P},\, \ovla=b$) at position $b$ is enough to formally capture the notion of \textit{reduction choice} used implicitly to define a derivation $P'$ typing $t'$ from a derivation $P$ typing $t$ when $t\rewb{b}t'$.
This allows us to define a suitable notion of \textbf{residuals} for \textbf{positions} (the residual of $\al \in \supp{P}$, if it exists, is a $ \al'\in \supp{P'}$ that may be denoted $\Res^{\rho}_b(\al')$ since it depends both on $b$ and $(\rho)_{\ovla=b}$).

Now, if instead of endowing $P$ with a total \textit{root} interface $(\rho_a)_{\ovla=b}$ at position $b$, we endow it with a total interface $(\phi_a)_{\ovla=b}$ at position $b$,  we are able to define a notion of \textbf{residuals} for \textbf{right bipositions} (\S\;\;\paper{\ref{ss:system-S-lipics-sC}}), as we did for system $\ttS$ (\S\;IV.D and Fig.\;1 in  in \cite{VialLICS17}): the residual of $(\al,c)$ may be then denoted $\Res^{\phi}_b (\al,c)$
 Interestingly, in that case, \textbf{residuation} can also be defined for interfaces: more precisely, if $P$ is endowed with a total interface at position $b$ and $\al\in \suppat{P}$ is such that $\ovlal\neq b$, then $\al$ has a residual $\al':=\Res^\phi_b(\al)$ w.r.t. $\phi$ and there is a bijection $\ResI^\phi_b(\al)$ from $\Inter^P(\al)$ to $\Inter^{P'}(\al')$.

Thus, every interface $\psi'_{\al'}$ at position $\al'$ in the derivation $P'$ typing the reduct $t'$ may be seen as the residual $\ResI^\phi_b(a)(\psi_\al)$ of some interface $\psi_\al$ at position $\al:=(\Res^{\phi}_b)^{-1}(\al')$ in the derivation $P$. Thus, any reduction choice inside a reduct $t'$ of $t$ can be endowed directly in $P$ (and not only in $P'$) provided we define interfaces for every redex that is fired to obtain $t'$. By induction, this allows us to endow directly inside an interface of $P$ a sequence of reduction choices till step $n$ along a reduction sequence $t=t_0\rewb{b_0} \ldots \rewb{b_{n-1}} t_n\ldots$. This is done gradually, step by step, for all $n\in \bbN$. This explains why:

\begin{lemma}[Representation]
  \label{l:rep-reduction-choices}
Every sequence of reduction choices of length $\leqs \om$ in a\ignore{ quantitative} derivation $\Pi$ can be built-in in an operable derivation $P$ representing
(collapsing on) $\Pi$.
\end{lemma}

Lemma\;\ref{l:rep-reduction-choices} is not a challenge to prove, because the hybrid application $\apph$ and the interface isomorphisms \techrep{have allowed}\paper{allow} us to relax system $\ttS$ into something closer to system $\scrR$ (although being still rigid), which is equivalent to the system of \cite{TsukadaAO17} with sequences instead of lists.

\techrep{
  \subsection{Residual Derivation}
  \label{s:res-deriv}

  For the remainder of Sec.\;\ref{s:reduction}, we prove all the claims above and we define properly residuals and some associated notions. We still assume that $\trb=\lxrs$, $t\rewb{b} t'$ and $P\tri \juCtt$. The hybrid derivation $P$ comes along with  with the usual associated notations  \eg $\ttC$ for $\ttC^P$, $\ttT$ for $\ttT^P$, $\ttpos$ for $\ttpos^P$ (see Sec.\;\ref{s:bisupp} and \ref{s:quant-deriv}).

  We assume that $P$ is endowed with total root interface $(\rho_a)_{\ovla=b}$ at position $b$. Using this root interface, we will now build a hybrid derivation $P'$ concluding with $\ju{C}{t:T'}$ with $T\equiv T'$, as it was announced at the beginning of Sec.\;\ref{s:reduction}.\\

\noindent \textbf{Conventions on metavariables $a$ and $\al$} The letter $a$ will stand for a position of $P$ that corresponds to the root of the redex (\ie $a\in \supp{P}$ and $\ovla=b$) and the letter $\al$ for other positions in $\supp{P}$ or even in $A$. We set $\Axl{a}=\AxP_{a\cdot 1\ct 0}(x)$ and $\Trl{a}=\set{\trP{\al_0}\,|\,\al_0\in \Axl{a}}=\Rt(\ttT^P(a\cdot 1))$. Thus, $\Axl{a}$ is the set of positions of the redex variable (to be substituted) above $a$ and $\Trl{a}$ is the set of the axiom tracks that have been used for them. 
For instance, in Fig.\;\ref{fig:SR-residuals-interface}, $\Axl{a}=\set{a\ct 1\ct 0 \cdot a_2,\; a\ct 1\ct 0 \ct a_7}$ and $\Trl{a}=\set{2,7}$.


First, we define the residual position $\Res_b(\al)$ for each $\al \in \supp{P}$ except when $\al$ is of the form $a,~ a\cdot 1$ or $a\cdot 1\ct 0\cdot a_k$ 
(for some $a$ satisfying $\ovla=b$). We begin with discussing  the symbols $\hearts$ and $\clubs$ in Fig.\ref{fig:SR-residuals-interface}.
 In Fig.\;\ref{fig:SR-residuals-interface}, $\hearts$ represents a judgment nested in $P_r$. Thus, its position must be of the form $a\cdot 1\ct 0\cdot \al_{\sheart}$. After reduction, the $\app$-rule and $\abs$-rule at positions $a$ and $a\cdot 0$ have been destroyed and the position of this judgment $\hearts$ will be $a\cdot \al_{\sheart}$.
We set then $\Res_b(a\cdot 1\ct 0\cdot \al_{\sheart})=a\cdot \al_{\sheart}$.

Likewise, $\clubs$ represents a judgment nested in the argument derivation $P_8$ on track 8 w.r.t. $a$. Thus, its position must be of the form $a\cdot 8\cdot \al_{\sclub}$ where $a\cdot 8$ is the root of $P_8$. After reduction, $P_{8}$ will replace the $\ax$-rule typing $x$ on track $\rho^{-1}_a(8)$ \ie 2, so its root will be at $a\cdot a_2$ (by definition of $a_2)$. Thus, after reduction, the position of judgment $\clubs$ will be $a\cdot a_{2}\cdot \al_{\sclub}$. We set then $\Res_b(a\cdot 7\cdot \al_{\sclub})=a\cdot a_{2}\cdot \al_{\sclub}$.

\begin{itemize}
\item Paradigm $\clubs$: if $\al=a\cdot \kR \cdot \al_0$ where $\ovl{a}=b$ and 
$\kR\in \Arg(a)$, then $\Res_b(\al)=a\cdot a_{\kL}\cdot \al_0$ with $\kL=\rho_a^{-1} (\kR)$.
\item Paradigm $\hearts$: if $\al=a\cdot 1\ct 0\cdot \al_0$ where $a \in \supp{P},\ovla=b$ and 
$\al_0 \neq a_k$, then $\Res_b(\al)=a\cdot \al_0$.
\item Outside the redex: if $b\nleqslant \ovlal$, then $\Res_b(\al)=\al$
\end{itemize}

We set $A'=\set{\Res_b(\al)\,|\, \al\in A}=\codom{\Res_b}$ and we call $A'$ the \textbf{residual support} of $P$ (w.r.t. reduction at position $b$ and root interface $(\rho_a)_{\ovla=b}$).  By case analysis, we check that $A'$ is a tree and that $\Res_b$ is an \textit{injection} from $\dom{\Res_b}$ to $A'$. Moreover, we set $\Apheart=\set{\Res_b(\al)\,|\, \al\in \dom{\Res_b},\, \ovlal \geqslant b\cdot 1\ct 0}$, $\Apclub=\set{\Res_b(\al)\,|\, \al \in \dom{\Res_b},\, \ovlal \geqslant b\cdot 2}$ and $\Apout=\set{\al\in A\,|\, \ovlal \ngeqslant b}$, so that $A'$ is the disjoint union of $\Apheart,\; \Apclub$ and $\Apout$.\\

\begin{rmk}[Induction and reduction]
\label{rmk:induction-split}
  Assume that $\Res_b(\al_i)=\al'_i~(i=1,\,2)$ and $\al_1\leqslant \al_2'$.
\begin{itemize}
\item If $\al'_1\in \Apclub$ , then $\al'_2\in \Apclub$.
  \item If $\al'_1 \in \Apheart$, then $\al'_2\in \Apheart\cup \Apclub$.
\item If $\al'_1\in \Apout$, then $\al'_2 \in \Apout \cup \Apheart \cup \Apclub=A'$.
\end{itemize}
So, a 001-induction on $A'$ should be be split in three 001-inductions: first, one on $\Apclub$, then one on $\Apheart$, then, one on $\Apout$. See for instance the proof of Lemma\;\ref{lem-def:QRes-bip-hybrid}.
\end{rmk}

\noindent Conversely, we check that the converse injection $\Res^{-1}_b$ from $A'$ to $\dom{\Res_b}$ satisfies:
\begin{itemize}
\item If $b\nleqslant \ovlalp$, we set $\Res^{-1} _b(\al')= \al'$.
\item If $\al'=a\cdot \al_0 '$ where $a\in \supp{P},\ovla=b$ but there is no $k$  s.t. $b\geqslant a\cdot a_k$, we set $Res^{-1} _b(\al')=a\cdot 1\ct 0 \cdot \al_0 '$.
\item If $\al'=a\cdot a_k\cdot \al_0 '$ where $a\in \supp{P},\ovla=b$, we set 
$\Res^{-1} _b(\al')=a\cdot \rho_a(k)\cdot \al_0 '$.
\end{itemize}

\noindent By case analysis:

\begin{lemma}
\label{l:A-prime-desc}
For all $\al'\in A'$ and $\al$ such that $\Res_b(\al)=\al'$:
\begin{itemize} 
\item $\ovlalp \in \supp {t'}$
\item $t'(\ovlalp)= t'(\ovlal )$.
\item $\childp{\al'}=\child{\al}$ (where $\ttchild=\ttchild^A$ and $\ttchild'=\ttchild^{A'}$, see Sec.\;\ref{s:tracks}).
\end{itemize}
\end{lemma}

\noindent We set $\Aparob=\set{a\in A'\,|\, t'(a')=\arob}$. By Lemma\;\ref{l:A-prime-desc}, $\Aparob=\set{\Res_b(a)\,|\, a\in \suppat{P}}$.\\

\subsection{Residual Types and Contexts}
\label{s:res-types-and-contexts}



In this paragraph, we define the residual derivation $P'$ of $P$ (w.r.t. reduction at position $b$ and the root interface $(\rho_a)_{\ovla=b}$). The residual support $A'$ will be the support of $P'$ (\ie $A'=\supp{P'}$), but we must also define the contexts $\ttC^{P'}(\al')$ and types $\ttT^{P'}(\al')$ for all $\al'\in A'$. Since $P'$ is not built yet, we will respectively denote these contexts and types $\ttC'(\al')$ and $\ttT(\al')$ when we define them.\\

We assume that $t$ satisfies Barendregt convention \ie for all $y\in \TermV$, $\ly$ occurs at most once in $t$ and the sets of free variables and of bound variables of $t$ are disjoint.\\

Let $\Ax$ and $\Ax'$ the respective sets of leaves of $A$ and $A'$. For all $\al'\in A'$ and $y\in \TermV$, $y\neq x$, we set $\Axpalp(y)=\set{ \al_0' \in A' \, | \, \al_0'\geqslant \al'~ \text{and}~  \Res_b^{-1}(\al)\in \Ax(y)}$\ignore{ and $\Ax'(y)=\Ax'(\epsi)(y)$}.  

We observe that, for all $\al'\in \Ax'$, $\Res_b^{-1}(\al') \in \Ax$. Then, we set $\trp{\al'}=\trP{\Res_b^{-1}(\al')}$ and, for all $\al' \in A'$ and $y\in \TermV,y\neq x$, $\ttC'(\al')(y)=(\trp{\al_0'}\cdot  \ttT(\Res_b^{-1}(\al_0')) )_{\al_0' \in \Axpalp(y)}$. This definition is sound, because, if $\al_1,~ \al_2\in \Ax(\al')(y)$ for some $\al'\in A'$ and $x\in \TermV$, then $\trp{\al_1}=\trp{\al_2}$ implies $\al_1=\al_2$ (case analysis).\\

When $t'(\ovlalp)=\ly$ (with $\al'\in A'$), a case analysis shows that $\Ax(\al')(y)=\set{\Res_b(\al_0) ~ | ~ \al_0  \in \Ax(\al)(y)}$ where $\al'\in A'$. Thus, in that case, $\ttC(\al')(y)=(\trp{\al_0'}\cdot \ttT(\Res_b^{-1}(\al_0')) )_{\al_0' \in \Axpalp(y)} =(\trp{\al_0}\cdot \ttT(\al_0))_{\al_0 \in \Ax(\al)(y)}=C(\al)(y)$.\\


\noindent By a $001$-induction on $\al' \in A'$, we define now $\ttT'(\al')$: 
\begin{itemize}
\item When $\al'\in \Ax'$, $\ttT'(\al')=\ttT(\Res_b^{-1}(\al'))$.
\item When $t'(\ovlalp)=\ly$, we set $\ttT'(\al)=C(\al\cdot 0)(y)\rew \ttT'(\al\cdot 0)$.
\item When $t'(\ovlalp)=\arob$, we set $\ttT'(\al')=
\Hd(\ttT'(\al'\cdot 1))$.\\
\end{itemize}

We define then $P'$ as the labelled tree  s.t. $\supp {P'}= A'$ and for all $\al'\in A',~ P'(\al')=\ju{\ttC'(\al')}{t'\rstr{\al'}:\ttT'(\al')}$ (so that $\ttC'=\ttC^{P'}$ and $\ttT'=\ttT^{P'}$ as expected). We intend to prove that $P'$ is a correct hybrid derivation typing $t'$.


As hinted at  Sec.\;\ref{s:reduction}, we must check that a type $\ttT(\al)$ (where $\al \in A$) may only be replaced with a type $\ttT'(\al')$ (where $\al'=\Res_b(\al)$) that is isomorphic to $\ttT(\al)$ in the residual derivation typing $t'$ that we are going to define.\\

In order to check that, it is convenient to extend residuation into \textit{quasi-residuation}: namely, we define \textbf{quasi-residual} $\QRes_b(\al)$ for any $\al \in A$ such that $\ovlal\neq b\cdot 1$ by setting $\QRes_b(\al)=\Res_b(\al)$ when $\Res_b(\al)$ is defined, $\QRes_b(a)=a$ and $\QRes_b(a\cdot 1\cdot 0\cdot a_k)=a\cdot a_k$ when $\ovl{a}=b$ and $k\in \RedTr(a)$.


\begin{rmk}
\begin{itemize}
\item  We do not necessarily have $t(\al)=t'(\al')$ or
$\child(\al)=\child'{\al'}$
  when $\al'=\QRes_b(\al)$ (compare with Lemma\;\ref{l:A-prime-desc}) and $\QRes_b$ is usually not injective. For instance, if $t=(\lx.y)y,\ t'=y,\ b=\epsi=\al=\al'$, then $t\rewb{b}t'$, $\al'=\QRes_b(\al)$ but $t(\al)=\arob \neq y=t'(\al')$ and $\child(\al)$ contains at least $1$ but $\childp{\al'}$ is empty.
\item However, quasi-residuals will be useful to define the isomorphisms $\Res_{b|\al}, ~ \ResR_{b|\al}$ and $\ResL_{b|\al}$ below.
\end{itemize}
\end{rmk}

\begin{lemma}
\label{lem-def:QRes-bip-hybrid} 
This lemma is also a definition: that of the \textbf{quasi-residuation}.
 \begin{itemize}
\item For all $\al'\in A',~ \ttT'(\al')\equiv \ttT(\al)$ where $\al=\Res_b^{-1}(\al')$. Besides, if $\al' \in \Ax'$ or $\al' 
\geqslant a\cdot a_k$ (for $\ovl{a}=b$), then $\ttT'(\al')=\ttT(\al)$.
\item More precisely, if $P$ is endowed with an interface $(\phi_a)_{\ovla=b}$ at position $b$ (extending the root-interface $(\rho_a)_{\ovla=b}$), then, for all $\al \in A$ and $\al'\in A'$ such that $\QRes_b(\al)=\al'$, we can define a type isomorphism $\QRes_{b|\al}:\,\ttT(\al)\rew \ttT'(\al')$ by  001-induction on $\al'$.
\item When $\Res_b(\al)=\al'$, we write $\Res_{b|\al}$ instead of $\QRes_{b|\al}$. Moreover, $\Res_{b|\al}$ is the identity if $\al'\in \Ax'$ or $\al'\geqslant a\cdot a_k$ for some $a\in A,\ovla=b$ and $k\in \AxTr(a\cdot 1\cdot 0)(x)$.
 \end{itemize}
\end{lemma}


\begin{proof}
We proceed by 001-induction on $\al'\in A':=\supp{P'}$ and split the cases as suggested in Remark\;\ref{rmk:induction-split}. 
  \begin{itemize}
\item Paradigm $\clubs$: $\Res_b(\al)=\al'$ and $\al'\geqslant a\cdot a_k$
(for some $a\in A,\ovla=b$ and $k\in \RedTr(a)$).
\begin{itemize}
\item Subcase $t(\ovlal)=y$: here, $t(\ovlal)=y\neq x$ 
and $\ttT'(\al')=\ttT(\al)$.
\item Subcase $t(\ovlal)=\lambda y$: $\al\cdot 0\in A$, $\Res_b(\al\cdot 0)=\al'\cdot 0$ and by IH, we have  $\ttT'(\al' \cdot 0)=\ttT(\al\cdot 0)$  and $\Res_{b|\al\cdot 0}$ is the identity $\ttid_{\ttT(\al\cdot 0)}$.\\
Since $\ttT(\al)=\ttC(\al\cdot 0)(y)\rew \ttT(\al\cdot 0)$ and 
$\ttT'(\al')= \ttC(\al\cdot 0)(y)\rew \ttT'(\al'\cdot 0)$, we also have
$\ttT(\al)=\ttT'(\al)$ and we set $\Res_{b|\al}=\ttid_{\ttT(\al)}$.
\item Subcase $t(\ovlal)=\arob$: $\al\cdot 1\in A$,  
$\Res_b(\al\cdot 1)=\al'\cdot 1$ and by IH, we have $\ttT'(\al\cdot 0)=\ttT(\al'\cdot 0)$ and $\Res_{b|\al\cdot 1}$ is the identity $\ttid_{\ttT(\al\cdot 1)}$.\\
Moreover, $\ttT(\al)=\Hd(\ttT(\al\cdot 1))$ and $\ttT'(\al')=\Hd(\ttT'(\al'))$.
So $\ttT(\al)=\ttT(\al')$ and we set $\Res_{b|\al}=\ttid_{\ttT(\al)}$.\\
\end{itemize}

\item Paradigm $\hearts$: $\al\geqslant a\cdot 1 \cdot 0$ and $\al'\geqslant a$ 
(for some $a\in A,\ovla=b$):
\begin{itemize}
\item Subcase $\al=a\cdot 1 \cdot 0\cdot a_{\kL}$ and $\al'=a\cdot a_{\kL}$: $\al'=\Res_b(a\cdot \kR)$ (where $\kR=\rho_a(\kL)$) and by $IH$, $\ttT'(\al')=\ttT(a\cdot \kR)$. Moreover, since $\ttT(\al)=\Lst(\al)\rstr{\kL}$, we can set $\QRes_{b|\al}=\phi_{a|\kL}$.
\item Subcase $t(\ovlal)=y\neq x$: $\Res_b(\al)=\al'$ and $\ttT'(\al')=\ttT(\al)$.
\item Subcase $t(\ovlal)=\lambda y$: $\al\cdot 0\in A$,
$\Res_b(\al\cdot 0)=\al'\cdot 0$: we set $\Res_{b|\al}=\ttC(\al \cdot 0)(y)\rew \QRes_{b|\al\cdot 0}$. 
\item Subcase $t(\ovlal)=\arob$: $\al\cdot 1\in A$, $\Res_b(\al\cdot 1)=\al'\cdot 1$: we set $\Res_{b|\al} =\Hd(\QRes_{b|\al\cdot 1})$.\\
\end{itemize}

\item Outside the redex: $\ovlal\ngeqslant b$: 
\begin{itemize}
\item Subcase $\al'\in \Ax'$: here, $t(\ovlal)=y\neq x$ and 
$\ttT'(\al')=\ttT(\al)$.
\item Subcase $\al=\al'=a$ (for some $a\in A,\ovla=b$): $a=\Res_b(a\cdot 
1 \cdot 0)$ and by IH, we have an type isomorphism $\QRes_{b|a\cdot 1 \cdot 0}:\,\ttT(a\cdot 
1 \cdot 0)\rew \ttT(a)$. Since $\ttT(a\cdot 1 \cdot 0)=\ttT(a)$, we can set $\QRes_{b|a}
=\QRes_{b|a\cdot 1 \cdot 0}$.
\item Subcase $t(\ovlal)=\lambda y$: $\al\cdot 0\in A$,
$\QRes_b(\al\cdot 0)=\al'\cdot 0$ and we set $\Res_{b|\al}=\ttC(\al \cdot 0)(y) \rew \QRes_{b|\al\cdot 0}$
\item Subcase $t(\ovlal)=\arob$: $\al\cdot 1\in A$,
$\QRes_b(\al\cdot 1)=\al'\cdot 1$: we set $\Res_{b|\al}=\Hd(\QRes_{b|\al\cdot 1})$.
\end{itemize}
\end{itemize}
  \end{proof}

\begin{remark}
It is far easier to define the residual of a biposition for a derivation of $\ttS$: if $P$ is trivial,  whenever $\alpha':= \Res_b(\alpha)$ is defined, the residual biposition of $\p:=(\al,\gam)\in \bisupp{P}$ is $\Res_b(\p)=(\alpha',\,\gamma)$, as in \cite{VialLICS17}.
\end{remark}

\subsection{Residual Interface}
\label{s:res-interface}

We notice that if $\al\in \suppat{P}$ and $\ovlal\neq b$, then $\Res_b(\al)$ is defined. So, for $\al'\in \Aparob$ (Sec.\;\ref{s:res-deriv}), we set $\Lst'(\al')=\Tl(\ttT'(\al'\cdot 1))$, $\ArgTr'(\al')=\set{ k \geqslant 2\, | \, \al' \in A'}$ and $\Rst'(\al')=(k\cdot \ttT'(\al'\cdot k))_{k \in \ArgTr'(\al)}$. We 
write then $\Inter'(\al')$ for the set of sequence type isomorphisms from 
$\Lst'(\al')$ to $\Rst'(\al')$.\\

Assume that $\al'\in \Aparob$. Let us write $\al=\Res_b^{-1}(\al')$, so that $\al \in \suppat{P}$, $\al\cdot 1\in A,~ \ovlal\neq b\cdot 1,~ \QRes_b(\al\cdot 1)=\al'\cdot 1$ and $\child(\al) =\child'(\al')$. Thanks to Lemma\;\ref{l:A-prime-desc}:
\begin{itemize} 
\item Since $\QRes_{b|\al\cdot 1}$  is a type isomorphism from $\ttT(\al\cdot 1)$ to $\ttT'(\al' \cdot 1)$ and $\Lst'(\al')=\Tl(\ttT'(\al' \cdot 1))$, then $\ttT'(\al'\cdot 1)$ is an arrow type (since $\ttT(\al\cdot 1)$ is)
and  we define the sequence type isomorphism $\ResL_{b|\al}$ by $\ResL_{b|\al}= \Tl(\Res_{b|\al\cdot 1})$. 
\item We can define $\ResR_{b|\al}$ by $\ResR_{b|\al}(k\cdot \gamma)=k\cdot \Res_{b|\al\cdot k}(\gamma)$. It is a sequence type isomorphism from $\Lst(\al)$ to $\Lst'(\al')$.
\end{itemize}

Thus, for each application node $\al\in \suppat{P}$ such that $\ovlal\neq b$, the residual $\al'=\Res_b(\al)$ is defined and we can define a bijection 
$\ResI_{b|\al}$ from $\Inter(\al)$ to $\Inter'(\al')$ by
$\ResI_{b|\al}(\phi)=\ResR_{b|\al}\circ \phi \circ \ResL_{b|\al}^{-1}$,
so that the following diagram is commuting:

$$  \begin{array}{ccc}
\Lst(\al) \phantom{ \ResL_{b|\al} }
        & \stackrel{\phi}{\xrightarrow{\hspace*{1cm}}}
        & \Rst(\al) \phantom{ \ResR_{b|\al} } \\
\Bigg\downarrow \ResL_{b|\al} & & \Bigg\downarrow \ResR_{b|\al} \\
\Lst'(\al') \phantom{ \ResL_{b|\al} }
        & \stackrel{\ResI_{b|\al}(\phi)}{\xrightarrow{\hspace*{1cm}}}
        & \Rst'(\al') \phantom{ \ResR_{b|\al} }
\end{array} $$

It means that the set of interface isomorphisms at position $\al'=\Res_b(\al)$ in $P'$ can also be seen as the residual bijective image of the set of interface isomorphisms at position $\al$ in $P$. This observation is pivotal to prove the Representation Lemma in the next subsection.\\

Assume $P$ is endowed with a complete interface $(\phi_a)_{a\in \supp{P}}$ (\ie $P$ is an operable derivation). For all $\al'\in \Aparob$, we set $\phi'_{\al'}=\ResI_{b|\al}(\phi_{\al})$, where $\al=\Res_b^{-1}(\al')$. Notice again that we can retrieve $\phi_\al$ from $\phi'_{\al'}$ since $\ResI_{b|\al}$ is a bijection.
We have enough to ensure:

\begin{proposition}
  This proposition is also a definition.
  \begin{itemize}
    \item The labelled tree $P'$ defined at end of Sec.\;\ref{s:res-types-and-contexts} is a hybrid derivation.
    \item When $P$ is endowed with an interface $(\phi_\al)_{\al\in \suppat{P}}$, then for all $\al'\in A'$, $\phi'_{\al'}$ is an interface isomorphism at pos. $\al'$.
\item Thus, the family $(\phi'_{\al'})_{\al'\in\suppat{P'}}$ is a total interface for $P'$. 
We call it the \textbf{residual interface} of $(\phi_\al)_{\al \in \supp{P}}$ after firing the redex at position $b$. When $P'$ is endowed with the
residual interface, it is an operable derivation
\end{itemize}
\end{proposition}

\noindent This entails, as expected:

\begin{corollary}[Pseudo Subject Reduction]
  \label{corol:pseudo-subj-red-Sh}
 If $t\rew t'$ and $\tri_{\ttSh} \juCtt$, then $\tri_{\ttSh} \ju{C}{t':T'}$ for some $T'\equiv T$.
\end{corollary}

Thus, if needed, we can apply a new $\beta$-reduction in $t'$ according
to this residual interface without having to define a new one. It allows
us to define deterministically the way we perform reduction (inside a derivation) in any reduction sequence of length $\ell \leqslant \omega$. Now, we can assert that, for instance, $A'=\supp{P'},\ \Aparob=\suppat{P'}$, for all $\al'\in \Aparob,\ \Lst'(\al')=\Lst^{P'}(\al'),\ \Rst'(\al')=\Rst^{P'}(\al')$. To sum up:

\begin{rmk} 
\begin{itemize}
\item We need only a total root interface at position $b$ to define the hybrid derivation $P'$.
\item If we have a total interface at position $b$, we may define the isomorphisms $\Res_{b|\al}:\, \ttT(\al)\rew \ttT'(\al')$  (resp. $\QRes_{b | \al}:\, \ttT(\al)\rew \ttT'(\al')$ ) for all $\al$ such that $\al':=\Res_b(\al)$ is defined (resp. such that  $\al':=\QRes_b(\al)$) is defined, as well as $\ResI_{b|\al}$ for all $\al \in \suppat{P}$ s.t. $\ovlal\neq b$.
\item It allows us to choose other interfaces (at positions different from $b$)
\textit{after} firing the redex at position $b$, as suggested in the end of Sec.\;\ref{s:rep-lemma-discussion}. This observation is only one we need to prove the Representation Lemma.
\end{itemize}
\end{rmk}
}

\techrep{
\subsection{Proof of the Representation Lemma}
\label{s:rep-lemma-proof}


We prove here that every sequence of subject-reduction steps that we 
perform "by hand" -- so called a \textbf{reduction choice sequence} -- starting from a derivation $\Pi$ can be built-in inside an operable derivation.\\

%
%

Let $\Pi \rhd \juGtt$ be a derivation, $P\tri \juCtt$ a hybrid derivation
collapsing on $\Pi$ and $t=t_0 \rewb{b_0} t_1\rewb{b_1} t_2  \rewb{b_2} \ldots \rewb{b_{i-1}} t_i \rewb{b_i} \ldots $ a sequence of reduction of length $\ell \leqslant \omega$ (when $\ell=\omega$, we do not need to assume strong convergence \cite{KennawayKSV97}).

We write $\rs$ for the sequence $(b_i)_{i< \ell}$ and  $\rs_n$ for the sequence $(b_i)_{i<n}$ for all $n<\ell$. If we perform reduction on $P$ along with $\rs$, we get a sequence of hybrid derivations $P_0$ (with $P_0=P$), $P_1,\; P_2,...$ such that $P_i$ concludes with  $\ju{C}{t_i:T_i}$ for some $T_i\equiv T$.

More precisely, for each step $i<\ell$ of $\rs$, we have to 
choose a root-interface $(\rho_a ^i)_{\ovla=b_i}$ at position $b_i$ in $P_i$ (typing $t_i$) corresponding to our reduction choice step, then to reduce $P_i$ 
\wrt $(\rho_a^i)_{\ovla=b}$, which yields a new hybrid derivation $P_{i+1}$. We proceed by induction on $i$.\\

Those reduction choices are heuristically made step-by-step. This raised 
the following question (Sec.\;\ref{s:encoding-red-choices}): is the notion of operable derivation 
expressive enough ? That is: can we endow $P$ with a complete interface, 
such that performing $\rs$ on $P$ follows exactly our 
step-by-step choices of substitutions? The answer is positive, as stated in Lemma\;\ref{l:rep-reduction-choices}. We will prove that now.\\

We set $A_i=\supp {P_i}$ for all $i<\ell$ and we define by induction on $i<\ell$ a partial function $\Res_{\rs(i)}$ from $A$ to $A_i$:
\begin{itemize}
\item $\Res_{\rs(0)}$ is the identity on $A$.
\item $\Res_{\rs(i+1)}=\Res^{\rho^i}_{b_i} \circ \Res_{\rs(i)}$, 
where $\Res^{\rho^i}_{b_i}$ is the residual function on  $A_i=\supp{P_i}$ defined w.r.t. reduction at position $b_i$ and the root-interface $(\rho_a^i)_{\ovla=b_i}$ (see Sec.\;\ref{s:res-deriv}).
\end{itemize}

For all $i<\ell$, let $A_{\rs(i)}$ denote the domain of $\Res_{(i)}$. Thus, $A_{\rs(0)}=A_0=\supp{P}$, $(A_{\rs(i)})_{i<\ell}$ is a decreasing sequence (\wrt $\subeq$) and, by induction on $i$, $\Res_{\rs(i)}$ is a bijection from $A_{\rs(i)}$ to $A_i$. We write $\Res_{\rs(i)}^{-1}$ for the converse bijection from $A_i$ to $A_{\rs(i)}$.


To lighten notations, we write $A_{(i)}$ and $\Res_{(i)}$ instead of $A_{\rs(i)}$ and $\Res_{\rs(i)}$. We also set  $A_{(i)}^{\arob}=A_{(i)}\cap \suppat{P}$. Thus, $\Res_{(i)}$ induces a bijection from $A_{(i)}$ to $A_i$.\\

Now, for all $i< \ell$, we chose an interface $(\phi_a^i)$ at position $b_i$
in $P_i$ such that $\Rt(\phi_a^i)= \rho_a^i$ for all $a\in A_i,\ \ovla=b_i=
b_i$.

We define by induction on $i$ a type isomorphism  $\Res_{(i)|\al}$ from 
$\ttT(\al)$ to $\ttT_i(\Res_{(i)}(\al))$ for all $\al \in A_{(i)}$ and a bijection $\ResI_{(i)|\al}$ from $\Inter(\al)$ to $\Inter_i(\Res_{(i)}(\al))$ for all $\al \in A_{(i)|\al}$ by:
\begin{itemize}
\item  $\Res_{(0)|\al}$ and $\Res_{(i)|\al}$ are respectively the identity functions on $\ttT(\al)$ and $\Inter(\al)$.
\item $\Res_{(i+1)|\al}=\Res_{b_i|\al_i} \circ \Res_{(i)|\al}$, where $\al_i=\Res_{(i)}(\al)$ and $\Res_{b_i|\al_i}:\, \ttT_i(\al_i)\rew \ttT_{i+1}(\al_{i+1})$ (with $\al_{i+1}=\Res_{(i+1)}(\al)$) is the residual type isomorphism (in the sense of Sec.\;\ref{s:res-types-and-contexts}) \wrt the interface $(\phi_a^i)$ at position $b$ in $P_i$.\\
We set likewise $\ResI_{(i+1)|\al}=\ResI_{b_i|\al_i} \circ  \ResI_{(i)|\al}$, where $\ResI_{b_i|\al_i}$ is the bijection (in the sense of Sec.\;\ref{s:rep-lemma-proof}) \wrt the interface $(\phi_a^i)$ at position $b$ in $P_i$.\\
\end{itemize}


\noindent To conclude, let $a\in \suppat{P}$. There are two cases:
\begin{itemize}
\item $\Res_{(i)}(a)$ is defined for all $i< \ell$. In that case, we choose
an arbitrary $\phi_a \in \Inter(\al)$.
\item There is a unique $0\leqslant i<\ell $ such that 
$\al_i=\Res_{(i)}(\al)$ is defined, but $\Res_{(i+1)}(\al)$
is not. In that case, $\ovl{\al_{i}}=b_i$ and we have already
chosen an interface isomorphism $\phi_{\al_i}^i\in \Inter_i(\al_i)$
(that extends $\rho_{\al_i}^i$). We set then $\phi_a=\ResI_{(i)|a}^{-1}(\phi_{\al_i}^i)$
\end{itemize}

By construction, the complete interface $(\phi_a)$ emulates the reduction w.r.t. the family $(\rho_a^i)_{\ovla=b_i}$. Thus:

\begin{lemma*}  
Every reduction choice sequence in a quantitative
derivation $\Pi$ can be built-in in an operable derivation representing
$\Pi$.
\end{lemma*}

\ignore{ 
\subsection{Residuals of Left Bipositions}

For technical reasons, it will be convenient to also define residuation for left bipositions, although, as we are going to see now, this notion does not work as well as in the right case.

\fubla{
reprendre les answers de lics
}

Let us assume that $\trb=\lxrs,\ t\rewb{b}t'$ (so that
$\tprb=\rsx,\; P\tri \juCtt$). We use the same notation conventions as in \ref{s:res-deriv} \eg the metavariable $a$ may only denote position $a\in \bisupp{P}$ such that $\ovla=b$ and so on. It is once again more convenient to assume that $t$ satisfies Barendregt convention.

For all $\p=(\al,y,k\cdot c)\in \bisupp{P}$:
\begin{itemize}
\item If $y=x$, then $\QRes_b(\p)$ is not defined.
\item If $y\neq x$\ldots
  \begin{itemize}
  \item \ldots and $\al'=\QRes_b(\al)$ is defined \ie $\ovlal\neq b\cdot1$, then we set $\QRes_b(\p)=(\QRes_b(\al),x,k\cdot c)$. We may write $\Res_b(\p)$ instead of $\QRes_b(\p)$ if $\Res_b(\al)$ is defined.
  \item \ldots and $\QRes_b(\al)$ is not defined \ie $\al=a\cdot 1$ with $\ovla=b$, then we set $\QRes_b(\p):=(a,x,k\cdot c)=\Res_b(a\cdot 1 \cdot 0,x,k\cdot c)$
  \end{itemize}
\end{itemize}
}

}


\section{\paper{Representation Theorem}\techrep{Isomorphisms of Derivations}}


\techrep{\label{s:iso-derivations}
  In this section, we explain why the proof of the Theorem\;\ref{th:rep-R-S-naive} may be related to the notion of isomorphisms of derivations and then, we will describe those isomorphisms in different suitable ways, till we introduce the notion of relabelling of derivation in Sec.\;\ref{s:relab-deriv}.

  As we will see in Sec.\;\ref{s:edge-threads-techrep}, some edges inside a derivations may be identified (because they correspond to a same type moving through the judgments). This motivates the notion of \textbf{threads}, that are formally defined in Sec.\;\ref{s:asc-pol-inv-threads-hybrid}.

}\paper{\label{s:rep-th-paper}
In this final section, we sketch the proof of the main result of this article: every $\scrR$-derivation is the collapse of an $\ttS$-derivation. First, we need  to render this statement more handleable
\ignore{
\subsection{Isomorphisms of Operable Derivations}
\label{s:iso-op-deriv-short}} by considering \textit{isomorphisms of derivations}.}

Since a hybrid derivation $P$ is a tree of $\bbN^*$ that is labelled with rigid judgments, and for all $a\in \supp{P}$, $\ttT^P(a)$ is also a labelled tree of $\bbN^*$, it is easy to define the notion of isomorphism from one hybrid derivation $P_1$ to another $P_2$\paper{ by using suitably the notion of 01-isomorphism (Sec.\;\ref{ss:tracks-sC})}.\techrep{ Namely, an isomorphism $\Psi$ from $P_1$ to $P_2$ is the datum of $\Psi_{\ttsupp}$, a 01-isomorphism from $\supp{P_1}$ to $\supp{P_2}$ and, for all $\ax$-rule of $a\in \supp{P_1}$, a type isomorphism from $\ttT^{P_1}(a)$ to $\ttT^{P_2}(\Psi_{\ttsupp}(a))$. We check in Sec.\;\ref{s:iso-op-deriv} that isomorphisms of hybrid derivations behave as expected (\eg a judgment of $P_1$ is mapped on a judgment of $P_2$ that represents the same $\scrR$-judgment).} Actually, two hybrid derivations $P_1$ and $P_2$ are isomorphic iff they collapse on the same $\scrR$-derivation.\techrep{\\}

From Sec.\;\ref{s:encoding-red-choices}, we recall that an operable derivation is a hybrid derivation $P$ that is endowed with a total interface $(\phi_a)_{a\in \suppat{P}}$. This motivates the notion of \textbf{isomorphism of operable derivations}. Informally, if $P_1$ and $P_2$ are operable derivations, then an isomorphism $\Psi$ of hybrid derivations from $P_1$ to $P_2$ is actually an isomorphism of operable derivations if $\Psi$ commutes with the interfaces of $P_1$ and $P_2$. \ie $P_1$ and $P_2$ are are \textit{operably} isomorphic if they collapse on the same $\scrR$-derivation and encode the same $\scrR$-reduction path.

Now, since we remember that a trivial derivation is endowed with identity interfaces, Theorem\;\ref{th:rep-R-S-naive} is a consequence of this one:

\begin{theorem}
\label{th:rep}
Every operable derivation is isomorphic to a trivial derivation.
\end{theorem}

With Lemma\;\ref{l:rep-reduction-choices}, this theorem means that any $\scrR$-derivation $\Pi$ and any sequence of reduction choices w.r.t. $\Pi$ can be encoded by a $\ttS$-derivation $P$, as expected from the Introduction.\techrep{
 Let us now properly define isomorphisms of derivation (Sec.\;\ref{s:iso-op-deriv}) and give different way to handle them. Most notably (Sec.\,\ref{s:relab-deriv}) we extend  the notion of 01-relabelling to define that of \textit{relabelling of derivations}, that allow to describe more lightly isomorphisms of derivations and, prospectively, to prove Theorem\;\ref{th:rep} above.}

\techrep{


\subsection{Isomorphisms of Operable Derivations}
\label{s:iso-op-deriv}


Let $P_1$ and $P_2$ be two hybrid derivations collapsing one the same $\scrR$-derivation $\Pi$ (thus, intuitively, $P_1$ and $P_2$ should be isomorphic as hybrid derivations). We set $A_i=\supp{P_i}$ and we write $\ttC_i,\; \ttT_i,\; \tttr_i,\; \ttpos_i$ for $\ttC^{P_i},\; \ttT^{P_i},\; \tttr^{P_i},\; \ttpos^{P_i}$ and so on. We write $\Ax_i$ for the set of leaves of $A_i$ ($i=1,2$).\\

\noindent A \textbf{hybrid derivation isomorphism} $\Psi$ from $P_1$ to  $P_2$ is  the date of:
\begin{itemize}
\item 
$\Psisu$ a 01-isomorphism from  $A_1$ to $A_2$. We often write $\Psi$ instead of $\Psisu$.
\item For each $a_1\in \Ax_1$, a type isomorphism $\Psi_{a_1}$ from $\ttT_1(a_1)$ to $\ttT_2(a_2)$, where $a_2=\Psi(a_1)$.
\end{itemize}

\noindent This is enough to define many useful isomorphisms, using the typing rules of System $\ttSh$:

\begin{itemize}
\item Since $\Psisu$ is a tree isomorphism, we check that it induces a bijection from $\Ax_1$ to $\Ax_2$.\\
So, for all $a_1\in A^1$, we can set $\Psitr(a_1)=\tttr_2(\Psi(a_1))$. Thus, $\Psitr(a_1)$ is the axiom track (in $P_2$) of the axiom rule $a_2$ corresponding to $a_1$ \textit{via} $\Psi$.\\

\item Since $\Psisu$ induces a bijection from $\Ax_1(a_1)(x)$ to $\Ax_2(a_2)(x)$ for all $a_1\in A_1,\ a_2=\Psi (a_1)$ and $x\in \TermV$, we can define a context isomorphism $\Psi_{a_1,x}$ from  $\ttC_1(a_1)(x)$ to $\ttC_2(a_2)(x)$ by $\Psi_{a_1,x}(k_1\cdot\gam_1)=\Psitr(a_{01})\cdot \Psi_{a_{01}}(\gam_1)$ for all $\gam_1\in \supp{\ttT_1(a_1)}$ , where $a_{01}=\ttpos_1(a_1,x,k_1)$.\\
\item By 001-induction, this allows us to define a type isomorphism $\Psi_{a_1}$ from $\ttT_1(a_1)$ to $\ttT_2(a_2)$ for any $a_1\in A_1$ and $a_2=\Psi(a_1)$ (not only for $a_1\in \Ax_1$ and $a_2\in \Ax_2$).
\begin{itemize}
\item $\Psi_{a_1}$ is already defined when $a_1\in \Ax_1$.
\item If $t(a_1)=\lambda x$, we set $\Psi_{a_1}=\Psi_{a_1\cdot 0,x}\rew \Psi_{a_1\cdot 0}$
\item If $t(a_1)=\arob$, we set $\Psi_{a_1}=\Hd(\Psi_{a_1\cdot 1})$.\\
\end{itemize} 
\item We assume here that $a_1\in \suppat{P_1}$ and $a_2=\Psi(a_1)$.
\begin{itemize}
\item We set $\Psi_{a_1}^{\Lst}=\Tl (\Psi_{a_1})$.
\item Since $\Psisu$ induces a bijection from $\ArgTr_1(a_1)$ to
$\ArgTr_2(a_2)$, we define a sequence type isomorphism $\Psi_{a_1}^{\Rst}$ from  $\Rst_1(a_1)$ to $\Rst_2(a_2)$ by $\Psi_{a_1}^{\Rst}(k_1\cdot \gam_1)=k_2\cdot \gam_2$, where  $a_2\cdot k_2=\Psisu(a_1\cdot k_1)$ and $\gam_2=\Psi_{a_1 \cdot k_1}(\gam_1)$.\\
\end{itemize}
\item We define now $\Psi$ on bisupports (with $a_2=\Psisu(a_2)$).
\begin{itemize}
\item Left case: $\Psi(a_1,\,x,\, k_1\cdot \gam_2)=(a_2,\,x,\, 
\Psi_{a_1,x}(k_1\cdot \gam_1))$
\item Right case: $\Psi(a_1,\,\gam_1)=(a_2,\, \Psi_{a_1}(\gam_1))$.\\
\end{itemize}
\end{itemize}

Let $P_1$ and $P_2$ be two operable derivation typing the same term $t$. Their interface isomorphisms are written $(\phi_{i,a})_{a \in \suppat{P_i}}$ ($i=1,\,2$). An \textbf{operable derivation isomorphism} is a hybrid derivation isomorphism $\Psi$ from $P_1$ to $P_2$ such that for all $a_1\in \suppat{P_1}$ and $a_2=\Psi(a_1)$, the following diagram is commuting:
$$\begin{array}{ccc}
\Lst_1(a_1) \phantom{\Psi a}
        & \stackrel{\phi_{1,a_1}}{\xrightarrow{\hspace*{1cm}}}
        & \Rst_1(a_1)  \phantom{\Psi a} \\
\Bigg\downarrow \Psi_{a_1}^{\Lst} 
	& & \Bigg\downarrow \Psi_{a_1}^{\Rst} \\
\Lst_2(a_2)  \phantom{\Psi a}
        & \stackrel{\phi_{2,a_2}}{\xrightarrow{\hspace*{1cm}}}
        & \Rst_2(a_2) \phantom{\Psi a}
\end{array}$$\\

\subsection{Resetting an Operable Derivation}
\label{s:bisupp-reset}


Let $P$ an operable derivation. We reuse the notations $A,\; \ttC,\; \ttT,\,\Ax,\; \tttr,\; \ttpos$ and $(\phi_a)_{a \in \suppat{P}}$. In the last section, we have defined,  the notion of isomorphism from one operable derivation $P_1$ to another $P_2$. Now, using 01-resettings (Sec.\;\ref{s:tracks}) and in particular, of type resetting (Sec.\;\ref{s:rigid-types-short}), we want to \textit{build}, given only one operable derivation $P$, a derivation $P_0$ that is isomorphic to $P$. This will serve the purpose of proving that every operable derivation is isomorphic to a trivial derivation.\\

\noindent A \textbf{resetting} of $P$ is given by the data of:
\begin{itemize}
  \item $\Psisu$ a 01-resetting of $A$ \ie a function from $A$ to $\mathbb{N}^*$ inducing a 01-isomorphism from $A$ to its codomain $A_0$. We write
$\Ax_0$ for the set of leaves of $A_0$.
\item For each $a\in \Ax$, $\Psi_a$ a type resetting of $\ttT(a)$.
\item $\Psitr$ an \textit{injection} from $\Ax$ from $\Nmzo$.
\end{itemize}
We will use $\Psitr$ to define new track values. Theoretically, we only need to avoid track conflict, but the injectivity of $\Psitr$ will actually serve to produce a derivation $P_0$ isomorphic to $P$ such that two distinct axiom rules of $P_0$ do not use the same axiom track, which of course ensures that there is no track conflict.\\

We may now define step by step the  derivation $P_0$ such that $\Psi$ actually defines an operable derivation isomorphism from $P$ to $P_0$.
\begin{itemize}
\item For all $a_0\in \Ax_0$, we set $\ttT_0(a_0)=\Psi_a(\ttT(a))$, where $a=\Psisu^{-1}(a_0)$. Thus, $\Psi_a$ induces a type isomorphism from $\ttT(a)$ to $\ttT_0(a_0)$
\item Let $\Ax_0(\Psi(a))(x)=\set{\Psisu(a_0)\,|\,al_0\in \Axal(x)}$ for all $a\in A$ and $x\in \TermV$.\\ 
Since $\Psisu$ induces a bijection from $\Ax$ to $\Ax_0$, we can 
set $\tttr_0(\Psi(a))=\Psitr(a)$ for all $a\in \Ax$ and $
\AxTr_0(\Psi(a))(x) =\Psitr(\Axal(x))$\\
Since $\Psitr$ is an injection whose domain is $\Ax$, we can define 
$\ttpos_0:\,\codom (\Psitr)\rew \Ax'$ with $\ttpos_0(k_0)=a_0$, where 
$a_0\in \Ax_0$ is the unique leaf of $A_0$ such that $\tttr_0(a_0)=k_0$ 
(by the injectivity of $\Psitr$ discussed above, the function $\ttpos_0$ requires only one argument, contrary to $\ttpos$).\\
\item We define then $\ttC_0(\Psi(a))(x) =(\tttr_0(\Psi(a_0))\cdot  \ttT(\Psi(a_0)))_{a_0\in \Axal(x)}$. So we can write $\Psi_{a,x}$ for the 
context isomorphism from $\ttC(a)(x)$ to $\ttC_0(a_0)(x)$ such 
that $\Psi_{a,\, x}(k \cdot \gam)=\Psitr(a_0) \cdot 
\Psi_{a_0} (\gam)$, where $a_0=\pos{a,\,x,\,k}$.\\
\item We can now define a type $\ttT_0(\Psi(a))$ and a type
isomorphism $\Psi_{a}$ from $\ttT(a)$ to $\ttT_0(\Psi(a))$ for all $a\in A$ by $001$-induction.
\begin{itemize}
\item $\ttT_(a_0)$ and $\Psi_{a}$ are already defined when 
$a\in \Ax$.
\item If $t(a)=\lambda x$, we set $\ttT_0(a_0)=\ttC_0(a_0)(x)\rew \ttT_0(a_0\cdot 0)$ and
$\Psi_{a}=\Psi_{a\cdot 0,x}\rew \Psi_{a \cdot 0}$.
\item If $t(a)=\arob$, we set $\ttT_0(a_0)= \Hd(\ttT_0(a_0))$ and $\Psi_{a}=\Hd(\Psi_{a \cdot 1})$.\\
\end{itemize}
\item We assume here that $t(a)=\arob$ and
$a_0=\Psi(a)$. We set then $\ArgTr_0(\Psi(a))=\set{k_0\geqslant 2\, |\, \exists k\in \ArgTr (a),\; a_0 \cdot k_0=\Psi(a \cdot k)}$,
$\Lst_0(a)=\Tl(\ttT_0(a_0 \cdot 1))$ and $\Rst_0(a_0)=
(k_0\cdot \ttT_0(a_0\cdot k_0))_{k_0\in \ArgTr_0(a_0)}$.
\begin{itemize}
\item We set $\Psi_{a}^{\Lst}=\Tl (\Psi_{a})$.
\item We define a sequence type isomorphism $\Psi_{a}^{\Rst}$ from $\Rst(a)$ to $\Rst_0(a_0)$ by
$\Psi_{a}^{\Rst}(k\cdot \gam)=k_0\cdot \gam_0$, where $a_0\cdot k_0
=\Psi(a_1 \cdot k_1)$ and $\gam_0=\Psi_{a \cdot k}(\gam_0)$.
\end{itemize}
We can set $\phi_{0,a_0}=\Psi_a^{\Rst} \circ \phi_a \circ
(\Psi_a^{\Lst})^{-1}$. Thus, $\phi_{0,a_0}:\, \Lst_0(a_0)
\rew \Rst_0(a_0)$ is a sequence type isomorphism.
\end{itemize}

\noindent The following proposition stems from the previous constructions:

\begin{proposition} 
\label{prop:deriv-reset}
  With the above notations, let $P_0$ be the labelled tree such that $\supp {P_0}=A_0$ and $P_0(a_0)=\ju{\ttC_0(a_0)}{ t\rstr{a_0}:\ttT_0(a_0)}$.\\
Then $P_0$ is a hybrid derivation and $(\phi_{0,\,a_0})$ is a complete
interface that makes $P_0$ isomorphic to $P$ as an operable derivation \via the isomorphism $\Psi$. Thus, we may naturally denote $P_0$ by $\Psi(P)$.
\end{proposition}
}

\techrep{\subsection{Relabelling a derivation}
\label{s:relab-deriv}}
\paper{\subsection{Getting a Trivial Derivation considering Tracks Threads}
\label{ss:track-threads-and-representability}

For the remainder of this paper, we explain why Theorem\;\ref{th:rep} above holds. \techrep{why every operable derivation is isomorphic to a trivial one (Theorem\;\ref{th:rep}).}

}
\techrep{
\begin{center}
\begin{figure}
  \begin{tikzpicture}

\draw (0,2.8) node [right] {\small \bf Abstraction rule};

\trans{-2}{0}{
  \draw (1.5,2.25) --++ (4.2,0);
  \draw (5.6,2.25) node [right] {\small $\ax$};
  \draw (1.5,2.3) node [below right] {$\ju{x:(5\ct S_5)}{x:S_5}\posPr{\,\pos{5}\,}$};
}
  \draw [dashed] (1.4,1.8) --++ (0,-0.2) --++(0.7,-0.5) --++ (0,-0.2);

  \draw (0,1) node [below right] { $\ju{C;\,x:\sSk}{t:T}\!\posPr{\,a\ct 0\,}$}
   ;
  
  \draw (0.1,0.4) --++(4.4,0);
  \draw (4.4,0.45) node [right] {$\abs$};
  
  \draw (0,0.4) node [below right] {$\ju{C }{\lambda x.t:\sSk\rew T }\!\posPr{a}$ };

\trans{5}{3.2}{ 
  
\draw (0,-0.4) node [right]{\small \bf Application rule};
  
 \draw (0,-0.7) node [below right] {$\ju{C}{t: \sSk \rew T}  \posPr{\,a\ct 1\,} \hspace{0.4cm} (\ju{D_k}{u:S'_k}  \posPr{\,a\ct k\,} )_{\kK'}$
  } ;

  \draw (0,-1.25) --++(7.65,0);
  \draw (7.77,-1.25) node [right] {\small $\apph$};
  
  \draw (1.77,-1.2) node [below right] {$\ju{C \uplus(\uplus_{\kK} D_k)}{tu:T} \posPr{\,a\,}$ };

\draw (0.5,-2.2) node [right] {endowed with an interface $\phi_a:\sSk \rew \sSpkp$};  
}
\end{tikzpicture}\\[-4ex]
  \caption{Ascendance and Polar Inversion}
  \label{fig:asc-hybrid}
\pierre{BIEN PLACER FIGURE}
\end{figure}
\end{center}
  }

We recall that tracks are numbers that label edges (of types or of derivations).
 Let $P$ be an operable derivation.  We want to find a trivial derivation $P_0$ that is isomorphic to $P$ (as an operable derivation). But roughly speaking, defining an isomorphism of operable derivation whose domain is $P$ is a matter of giving new values to the tracks which occur in $P$. These new values must be chosen appropriately to (1) respect the typing rules of $\ttSh$ (2) respect the interface of $P$ and yield a trivial derivation. \\

In System $\ttSh$, tracks 0 and 1 are special (they are dedicated to the premise of the $\abs$-rule or the left-premise of the $\apph$-rule and also to the target of $\rew$) and their value is fixed by 01-isomorphisms. But the value of tracks $\geqs 2$ 
may be changed: we say that they are \textit{mutable}. We informally write $\Edg{P}$ for the set of edges nested in $P$ whose tracks are mutable\techrep{. An element $\tte\in \Edg{P}$ is called a \textbf{mutable edge}, that may be of 3 natures: 
\begin{itemize}
\item The edges of the source of arrows nested in types (\textbf{inner mutable edges}).
\item The edges leading to an argument derivation in some $\apph$-rule (\textbf{argument edge}).
\item The \textbf{axiom edges} which are labelled with axiom tracks in contexts.
\end{itemize}
}
\techrep{

  The set of argument edges corresponds to $\suppmut{P}$ (Sec.\ref{s:tracks}). By analogy with $\ttsuppmut$, we define 
  the \textbf{mutable bisupport} of $P$ by $\bisuppmut{P}=\set{(a,c)\in \bisupp{P}\, | \, c\in \suppmut{\ttT^P(a)}}\cup \set{(a,x,\ell\cdot c)\in \bisupp{P}\,|\,c\in \suppmut{\ttC^P(a)(x)}}$. Implicitly: 
  \begin{itemize}
\item  A right biposition $\p=(a,c\cdot k)\in \bisuppmut{P}$ stands for the edge from $c$ to $c\cdot k$ in the type $\ttT^P(a)$. Its label $\lab{\p}$ is then $k$.
\item  A left biposition $(a,x,\ell\cdot c\cdot k)\in \bisuppmut{P}$ stands for the inner edge from $c$ to $c\cdot k$ in the sequence type $\ttC^P(a)(x)$. We set then $\lab{\p}=k$.
  \item A left biposition $\p=(a,x,\ell)$ stands for the ``axiom edge'' from $\epsi$ to $\ell$ w.r.t. $\ttC^P(a)(x)$ (we recall that $\epsi\notin \supp{\ttC^P(a)(x)}$ since $\ttC^P(a)(x)$ is a forest). We set then $\lab{\p}=\ell$.
  \end{itemize}
Thus, the set of mutable edges  $\Edg{P}$ may be identified to $\bisuppmut{P}\cup \suppmut{P}$.\\
  

However, as it has been noticed before (the hybrid construction of Sec.\;\ref{s:hybrid-cons-techrep} makes use of that), in a hybrid derivation $P$, every type and context is determined by the types and axiom tracks given in the axiom leaves (if the forget about outer argument tracks). So resetting $P$ (hopefully, into a \textit{trivial} derivation) is mainly a matter of giving good track values to the mutable edges in the axiom leaves (and not to every mutable edge in $P$). We invoke the notion of 01-relabelling from Sec.\;\ref{s:tracks} and of type relabelling from Sec.\;\ref{s:rigid-types-short}\\

\begin{definition}
  \label{def:relabelling}
 A \textbf{relabelling} $\Relab$ of an operable derivation $P$ is given by:
\begin{itemize}
\item A relabelling $\Relab_{\arg}$ of $\suppmut{P}$.
\item For all $a \in \Ax$, a relabelling $\Relab_a$ of $\ttT^P(a)$.
\item An \textit{injection} $\Relabtr$ from $\Ax$ to $\Nmzo$.
\end{itemize}
\end{definition}

\noindent When such a $\Relab$ is given, we reuse the construction of the resetting induced by a relabelling of Sec.\;\ref{s:tracks}.
\begin{itemize}
\item We define a function $\Psisu^\Relab$ as the 01-resetting of $\supp{P}$ induced by $\Relabarg$.
\item For all $a\in \Ax$, we define $\Psi^\Relab_a$ as the 01-resetting of $\ttT^P(a)$ induced by $\Relab_a$.
\item We set $\Psitr^\Relab=\Relabtr$.
\end{itemize} 

\noindent Thus, $\Psi^{\Relab}$ is a resetting of $P$ in the sense of Sec.\;\ref{s:bisupp-reset}. We then set $P^{\Relab}=\Psi^{\Relab}(P)$ (Prop.\;\ref{prop:deriv-reset}). Thus, $P^{\Relab}$ is a derivation that is isomorphic to $P$.\\

We can now define $\Relab_a:\ \suppmut{ \ttT(a)}   \rew \Nmzo$
for all $a \in \supp{P}$ by $\Relab_a(\gam\cdot k)=\lab {\Psi_a(\gam\cdot k)}$ and then, for all $a \in \suppat{P}$, we can set:
\begin{itemize}
\item For all $k\cdot \gam \in \suppmut{ \Lst(a)}$,  $\Relab_a^{\Lst}(k\cdot \gam)=\lab{\Psi^{\Lst}(k \cdot \gam)}$.
\item For all $k \cdot \gam \in \suppmut{ \Rst(a)}$, $\Relab_a^{\Rst}(k \cdot \gam)=\lab{ \Psi^{\Rst}(k \cdot \gam)}$.
\end{itemize}

}


We now discuss the moves of a type (and its edges) inside a hybrid derivation  by looking at Fig.\;\ref{fig:two-bro-threads-verbose}, before explaining in what this concerns the trivialization of hybrid derivations \via the notions of brotherhood and consumption below.

Some occurrences of 8 are colored in blue or in red: they all correspond to the label of an edge that ``moves'' inside the derivation. For instance, each red (\resp blue) occurrence of 8 can be identified to the one just below \via the typing rules of system $\ttSh$: we say that the former is the \textbf{ascendant} of the latter. The same can be said about the colored occurrences of 9, 2 and 7. We call a series of ascendants an \textbf{ascendant (edge) thread} \eg the set of the red (\resp blue) occurrences of $8$ correspond to an ascendant thread (\resp to another).

Moreover, since the $\abs$-rule ``calls'' all the assigned types of the bound variable, the top \textit{blue} occurrence of 8 is called by the constructor $\lx$ and correspond to the top \textit{red} occurrence of 8. We say that the former occurrence of $8$ is the \textbf{polar inverse} of the latter.

Ascendance and polar inversion induce a congruence between mutable edges, that we also denote $\equiv$ and that identifies labelled edges \wrt the moves of the types inside $P$. Intuitively, two labelled edges $\tte_1$ and $\tte_2$ are congruent iff the typing rules of $\ttSh$ constrain them to be labelled with the same track. We denote by $\Thrtr{P}$ the quotient set of $\Edg{P}$ by $\equiv$ and an element  $\theta\in \Thrtr{P}$ is called a \textbf{mutable edge thread} or simply a \textbf{thread}. A thread is composed of at most two ascendant threads. An occurrence of a thread is said to be \textbf{negative} if its top ascendant is in an $\abs$-rule and \textbf{positive} in every other case (top ascendant in an $\ax$-rule or argument edge). In Fig.\;\ref{fig:two-bro-threads-verbose}, the set of colored occurrences of 8 (\resp of 9) correspond to a thread: the blue ones are negative (they ascend to an abstraction) and the red ones are positive (the ascend to an $\ax$-rule). The colored occurrences of 2, 3, 5 and 7 are all positive. We  denote these threads $\theta_i$ ($i=2,3,5,7,8,9$) without ambiguity.  

We are interested in finding an isomorphism fulfilling Theorem\;\ref{th:rep} and points (1) and (2) at the beginning of this section. The discussion above explains how to capture point (1) (respecting the typing rules). In other words:

\begin{observation}
\label{obs:iso-deriv-sC}
Defining an isomorphism from $P$ is about specifying a new value $\Val{\theta}$ for the track of each mutable edge thread $\theta$ of $P$.
\end{observation}

\subsection{Brotherhood and Consumption}
\label{ss:bro-consum-lipics-sC}

We now explain how to capture point (2) of \S\;\ref{ss:track-threads-and-representability} \ie we also want to respect the interface $(\phi_a)_{a\in \suppat{P}}$ of $P$ while obtaining a trivial interface (\ie using only identity isomorphisms).

Still in Fig.\;\ref{fig:two-bro-threads-verbose}, every colored occurrence of 8 occurs beside a colored occurrence of 9 with the same polarity (each pair of 8 and 9 is nested in the same sequence type): we say that\techrep{ the threads} $\theta_8$ and $\theta_9$ are \textbf{brother threads}. Likewise, the orange threads $\theta_3$ and $\theta_5$ are brothers, as well as the purple threads $\theta_2$ and $\theta_3$. \\

We endow now the derivation of Fig.\;\ref{fig:two-bro-threads-verbose} with an interface: in the $\apph$-rule typing $((\ly x.xy)z)(a\,x)$, there are two possible interfaces $\phi_{1}$ from $\typetwo$ to $\typetwobis$: one ``maps'' 8 onto 2 and 9 onto 7 and the other 8 onto 7 and 9 onto 2. We consider the second case (8 into 7,9 into 2). Likewise, at the root $\apph$-rule, there are two interfaces $\phi_{\epsi}$ from $(8\ct \tv,9\ct \tv)$ to $(3\ct \tv,5\ct \tv)$. We assume that $\phi_{\epsi}$ maps 8 to 3 and 9 to 5.

We associate to the interface the relation of \textbf{consumption}: since,
intuitively, $\phi_1$ associates to the \textit{blue} track 8 (on the left-hand side) the track 2 (on the right hand-side), we say that the thread $\theta_8$ (\resp $\theta_2$) is \textbf{left-consumed} (\resp \textbf{right-consumed}) at position 1. Moreover, at position 1, the consumed occurrence of 8 is \textit{negative} (blue), so we say that $\theta_8$ is left-consumed negatively, and we write $\theta_8 \lomin \rrewc{1} \roplus \theta_2$ (since the consumed occurrence of 2 is positive). Likewise,  $\theta_9 \lomin \rrewc{1} \roplus \theta_7$, 
 $\theta_8\loplus \rrewc{\epsi}\roplus   \theta_3$, $\theta_9\loplus \rrewc{\epsi}   \roplus \theta_5$. We then say \eg that $\theta_9$ and $\theta_7$ \textbf{face each other} at pos. 1.\\



\begin{figure}
\hskip -3.1cm 
{\fnsz  
\infer[\apph]{
\infer[\apph]{
\infer[\apph]{ 
\infer[\abs]{ 
\infer[\abs]{ 
\infer[\apph]{ 
\infer[\ax]{\phd}{\ju{5}{x:\typetwored}}
\infer[\ax]{\phd}{\ju{3}{y:\tv}\trckblck{2}} 
}{
\ju{\ldots}{x\,y:\typeonered}
}}{\ju{\ldots}{\lx.x\,y:(5\ct \typetwoblue) \rewsh \typeonered }
}}{\ju{}{\ly x.x\,y:(7\ct \tv)\rewsh(5\ct \typetwoblue)\rewsh \typeonered}
}
\infer[\ax]{\phd}{\ju{4}{z: \tv}\trckblck{3} }
}{
\ju{\ldots}{(\ly x.x\,y)z:(5\ct \typetwoblue) \rewsh \typeonered}
}
 \infer[\apph]{
\infer[\ax]{\phd}{\ju{2}{a:\est\rewsh  \typetwobispurple}}}{\ju{\ldots}{a\,x:\typetwobispurple}\trckblck{6}}
}{
\ju{\ldots}{((\ly x.xy)z)(a\,x):\typeonered}
}
\infer[\ax]{\phd}{\ju{4}{b: \tv\orange{\trckblck{3}}}}\
\infer[\ax]{\phd}{\ju{9}{b: \tv \orange{\trckblck{5}}}}
}{
\ju{\ldots}{(((\ly x.xy)z)(a\,x))b:\tv'}
}}
\caption{Two Brother Threads}
\label{fig:two-bro-threads-verbose}
\vspace{-0.5cm}
\end{figure}
\ignore{\small  
\infer[\apph]{
\infer[\apph]{
\infer[\apph]{ 
\infer[\abs]{ 
\infer[\abs]{ 
\infer[\apph]{ 
\infer[\ax]{\phd}{\ju{5}{x:\typetwored}}\hspace{0.2cm}
\infer[\ax]{\phd}{\ju{3}{y:\tv}\trckblck{2}} 
}{
\ju{\ldots}{x\,y:\typeonered}
}}{\ju{\ldots}{\lx.x\,y:(5\ct \typetwoblue) \rewsh \typeonered }
}}{\ju{}{\ly x.x\,y:(7\ct \tv)\rewsh(5\ct \typetwoblue)\rewsh \typeonered}
}
\hspace{0.2cm}
\infer[\ax]{\phd}{\ju{4}{z: \tv}\trckblck{3} }
}{
\ju{\ldots}{(\ly x.x\,y)z:(5\ct \typetwoblue) \rewsh \typeonered}
}\hspace{0.2cm}
 \infer[\apph]{
\infer[\ax]{\phd}{\ju{2}{a:\est\rewsh  \typetwobispurple}}}{\ju{\ldots}{a\,x:\typetwobispurple}\trckblck{6}}
}{
\ju{\ldots}{((\ly x.xy)z)(a\,x):\typeonered}
}\hspace{0.2cm}
\infer[\ax]{\phd}{\ju{4}{b: \tv\orange{\trckblck{3}}}}\
\infer[\ax]{\phd}{\ju{9}{b: \tv \orange{\trckblck{5}}}}
}{
\ju{\ldots}{(((\ly x.xy)z)(a\,x))b:\tv'}
}
}
\ignore{
\hspace*{0.1cm} We have represented  edge brother threads in a $\ttSh$-derivation. The bipositions of two brothers threads are red and blue. One brother thread, denoted $\theta^{\ttB\ttR}_8$ is composed of edges labelled with 8 and the other, denoted $\theta^{\ttB\ttR}_9$ of edges labelled with 9. The positive (\resp negative) occurrences of these two threads are colored in red (\resp blue). The top blue occurrence of 8 (\resp 9) corresponds to the polar inverse of the top red one. Indeed, the two top red (\resp blue) occurrences of 8 and 9 are in the conclusion of an $\ax$-rule (\resp an $\abs$-rule) and are thus positive (\resp negative). Besides, two brother threads are represented in purple, one is labelled with 2 (denoted $\theta^{\ttP}_2$) and the other $\theta^\ttP_7$ with 7. Both threads only have positive occurrences. Finally, the orange argument tracks (3 and 5) of the $\ax$-rules typing the variable $b$ correspond to two singleton brother thread $\theta^\ttO_2$ and $\theta^\ttO_8$.

\hspace*{0.1cm} Note that, in the $\apph$-rule typing $((\ly x.xy)z)(a\,x)$, there are two possible interfaces  $\phi_{1}$ from $\typetwo$ to $\typetwobis$: abusively, one relabels 8 into 2 and 9 into 7 and the other 8 into 7 and 9 into 2. We consider the second case (8 into 7,9 into 2). Likewise, at the root $\apph$-rule, there are two interfaces $\phi_{\epsi}$ from $(8\ct \tv,9\ct \tv)$ to $(3\ct \tv,5\ct \tv)$. We assume that $\phi_{\epsi}$ maps 8 to 3 and 9 to 5.

\hspace*{0.1cm} The brother threads labelled with 8 and 9 are consumed at two places: positively at pos. $\epsi$ (root of the derivation) and negatively at pos. 1 (application of $a\,x$). These two brother threads are left-consumed positively facing the two orange brother threads (which are \textit{right}-consumed positively) and they are right-consumed negatively facing the two purple brother threads (which are also right-consumed positively). More precisely, due to our choice of $\phi_{\epsi}$ and $\phi_1$, the thread labelled with 8 (\resp with 9) is left-consumed positively facing the orange thread labelled with 2 (\resp with 8) ad it is left-consumed negatively facing the purple thread labelled with 7 (\resp with 2): this is formally denoted by $\theta^{\ttB \ttR}_8\loplus \rrewc{\epsi} \theta^\ttO_2$, $\theta^{\ttB \ttR}_9\loplus \rrewc{\epsi}   \theta^\ttO_8$, 
$\theta^{\ttB \ttR}_8\lomin \rrewc{1} \theta^\ttP_3$, $\theta^{\ttB \ttR}_9\lomin \rrewc{1}   \theta^\ttO_5$. 
}

Now, to obtain Theorem\;\ref{th:rep}, we want a trivial interface up to relabelling the threads. A derivation is trivial when its interfaces are only identities. In the light of Obs.\;\ref{obs:iso-deriv-sC}, we remark that, to obtain a trivial derivation, we must assign the same \textit{new} track value to \textit{any pair of threads that face each other}. As a consequence, if two threads $\theta$ and $\theta'$ are congruent modulo $\rrew :=\cup_a \rrewc{a}$, they should also be assigned the same track value \eg $\theta_5$ and $\theta_7$ must receive the same new label, even if they have no direct relation, because $\theta_9\rrew \theta_7$ and $\theta_9\rrew\theta_5$.


\ignore{
\newcommand{\typeone}{\mathord{(8\ct \tv,9\ct \tv)\rewsh \tv\secu}}
\newcommand{\typetwo}{\mathord{(3\ct \tv,8\ct \tv')\rewsh \typeone}}

$P_1=$
\infer[\apph]{ 
\infer[\abs]{ 
\infer[\abs]{ 
\infer[\apph]{ 
\infer[\ax]{\phd}{\ju{5}{x:\typetwo}}\hspace{0.2cm}
\infer[\ax]{\phd}{\ju{3}{y:\tv}\trck{3}}\hspace{0.1cm}  
\infer[\ax]{\phd}{\ju{7}{y:\tv'}\trck{8}} }{
\ju{\ldots}{x\,y:\typeone}
}}{\ju{\ldots}{\lx.x\,y:(5\ct \typetwo) \rewsh \typeone }
}}{\ju{}{\ly x.x\,y:(3\ct \tv,7\ct \tv')\rewsh(5\ct \typetwo)\rewsh \typeone}
}
\hspace{0.2cm}
\infer[\ax]{\phd}{\ju{2}{z: \tv}\trck{3} }
\hspace{0.1cm}
\infer[\ax]{\phd}{\ju{3}{z: \tv'}\trck{7} }
}{
\ju{\ldots}{(\ly x.x\,y)z:(5\ct \typetwo) \rewsh \typeone}
}
\vspace{0.5cm}

$P_2=$
\infer[\apph]{
\infer[\apph]{P_1\tri\ju{\ldots}{(\ly x.x\,y)z:(5\ct \typetwo) \rewsh \typeone}
\hspace{0.2cm} \infer[\apph]{
\infer[\ax]{\phd}{\ju{2}{a:\est\rewsh  \typetwo}}}{\ju{\ldots}{a\,x:\typetwo}\trck{5}}
}{\ju{\ldots}{((\ly x.xy)z)(a\,x):\typeone}}\hspace{0.2cm}
\infer[\ax]{\phd}{\ju{4}{b:\tv\trck{8}}}\
\infer[\ax]{\phd}{\ju{9}{b:\tv\trck{9}}}
}{
\ju{\ldots}{(((\ly x.xy)z)(a\,x))b}
}

%


$P_1=$\infer[\apph]{ 
\infer[\apph]{ 
\infer[\abs]{ 
\infer[\abs]{ 
\infer[\apph]{ 
\infer[\ax]{\phd}{\ju{5}{x:\typetwo}}\hspace{0.2cm}
\infer[\ax]{\phd}{\ju{3}{y:\tv}\trck{3}}\hspace{0.1cm}  
\infer[\ax]{\phd}{\ju{7}{y:\tv'}\trck{8}} }{
\ju{\ldots}{x\,y:\typeone}
}}{\ju{\ldots}{\lx.x\,y:(5\ct \typetwo) \rewsh \typeone }
}}{\ju{}{\ly x.x\,y:(3\ct \tv,7\ct \tv')\rewsh(5\ct \typetwo)\rewsh \typeone}
}
\hspace{0.2cm}
\infer[\ax]{\phd}{\ju{2}{z: \tv}\trck{3} }
\hspace{0.1cm}
\infer[\ax]{\phd}{\ju{3}{z: \tv'}\trck{7} }
}{
\ju{\ldots}{(\ly x.x\,y)z:(5\ct \typetwo) \rewsh \typeone}
}
\infer[\apph]{
\infer[\ax]{\phd}{\ju{2}{a:\est\rewsh  \typetwo}}}{
\ju{\ldots}{a\,x:\typetwo}\trck{5}
} 
}{
\ju{\ldots}{((\ly x.xy)z)(a\,x):\typeone
}}} 

To prove Theorem\;\ref{th:rep}, we must then prove that we can assign a value $\Val{\theta}$ to each edge thread $\theta \in \Thrtr{P}$, such that, if $\thL\rrewc{a} \thR$ for some $a\in \suppat{P}$, then $\Val{\thL}=\Val{\thR}$. This assignation $\ttVal$ must be \textbf{consistent} \techrep{\ie satisfy the two following points:
\begin{itemize}
\item No track conflict should arise from $\ttVal$ (in the contexts of $\apph$-rules).
\item Two brother threads should not be reassigned the same value \eg  
 $\theta_8$ and $\theta_9$ (or $\theta_2$ and $\theta_7$) cannot receive the same new track value (\ie $\Val{\theta_8}\neq \Val{\theta_9}$ must hold).
\end{itemize}}
\paper{
 \ie two brother threads should not be reassigned the same value \eg  in Fig.\;\ref{fig:two-bro-threads-verbose},
 $\theta_8$ and $\theta_9$ (or $\theta_2$ and $\theta_7$) cannot receive the same new track value (\ie $\Val{\theta_8}\neq \Val{\theta_9}$ must hold).
}

\techrep{Actually, the first condition (no track conflict) is easy and we}\paper{We} can prove (Proposition\;\ref{prop:if-no-bro-then-th-rep} below) that such a good assignation $\ttVal$ exists iff there is no proof  showing that there are two brother threads that should be given an equal track value\techrep{\pierre{(based upon the interface $(\phi_a)_{a\in \suppat{P}}$)}}. Such a proof would be called a \textbf{brother chain} and would have the form 
$\theta_0 \lrrewc{a_0} \theta_1\lrrewc{a_1}\ldots \lrrewc{a_{n-1}} \theta_n$, where $\theta_0$ and $\theta_n$ are two brother threads and $\lrrewc{a}$ is the symmetric closure of $\rrewc{a}$. This would imply that $\Val{\theta_0}$ and $\Val{\theta_n}$ \textit{must} be equal, which is illicit ($\theta_0$ and $\theta_n$ are brothers). In other words, a brother chain corresponds to a proof of contradiction in \paper{the}\techrep{an \textit{ad hoc}} first order theory $\calT_P$ whose set of constants is $\Thrtr{P}$ and whose axioms are $\Val{\theta_1}=\Val{\theta_2}$ for all $\theta_1,\theta_2$ s.t. $\theta_1 \rrewc{a} \theta_2$ for some $a\in \suppat{P}$ and $\Val{\theta_1}\neq \Val{\theta_2}$ for all brother threads $\theta_1,\theta_2$. Thus, Proposition\;\ref{prop:if-no-bro-then-th-rep} can be interpreted as a \textbf{completeness property} for theory $\calT_P$. It is pivotal to obtain Theorem\;\ref{th:rep}:

\begin{proposition}
\label{prop:if-no-bro-then-th-rep}
If there is no brother chains \wrt an operable derivation $P$\techrep{ collapsing on a $\scrR$-derivation $\Pi$}, then there is a trivial derivation $P_*$ isomorphic to $P$\techrep{ (and thus also collapsing on $\Pi$}.
\end{proposition}

\paper{
To develop the final stages of this theorem, we also need the following lemmas:

\begin{lemma}[Uniqueness of Consumption]
  \label{l:uniqueness-consumption-hybrid}
  Let $P$ be an operable derivation, $\oast \in \set{\oplus,\omin}$ and $\theta \in \ThrE{P}$. Then, there is a most one $\theta'$ such that ($\theta \loast  \rrew \theta'$ or $\theta \loast\werr \theta'$).
\end{lemma}
\ignore{
Brother threads always face brother threads:

\begin{lemma}
\label{l:consum-brother-thr-same-pos}
Let $\theta_1$ and $\theta_2$ be two brother edges and $a\in \suppat{P}$. Then $\theta_1$ is consumed at position $a$ iff $\theta_2$ is. Moreover:
\begin{itemize}
\item If $\theta_1\rrewc{a} \theta^{\Rst}_1$, then, there is $\theta^{\Rst}_2$, brother with $\theta^{\Rst}_1$, such that $\theta_2\rrewc{a}\theta^{\Rst}_2$.
\item If $\theta^{\Lst}_1\rrewc{a} \theta_1$, then, there is $\theta^{\Lst}_2$, brother with $\theta^{\Lst}_1$, such that $\theta^{\Lst}_2\rrewc{a}\theta_2$.
\end{itemize}
\end{lemma}}

We define the \textbf{applicative depth} of a thread $\theta$ as the maximal applicative depth of a judgment in which it occurs \eg in Fig.\;\ref{fig:two-bro-threads-verbose}, $\ad{\theta_8}=\ad{\theta_9}=0$ and $\ad{\theta_7}=0$.
 Consumption causes an increase of applicative depth, provided the left thread is consumed positively.

\begin{lemma}[Monotonicity]
\label{l:consum-and-ad}
If $\thL\loplus\rrew \thR$, then $\ad{\thL}<\ad{\thR}$.
\end{lemma}

}

\techrep{

For now, let us properly define threads and consumption.

\subsection{Ascendance, Polar Inversion and Threads}
\label{s:asc-pol-inv-threads-hybrid}



For the discussion below, we need to recall  that, given a type $T$ and a sequence type $\sSk$, a position $c\in \supp{T}$ corresponds to the position $1\cdot c\in \supp{\sSk\rew T}$, since $T$ occurs in this arrow type right-hand side. And if $k\in K$, position $c$ in $S_k$ corresponds to position $k\cdot c$ in $\sSk$.\\

\noindent \textbullet~Assume that we find an $\abs$-rule at position $a$ in $P$ as in Fig\;\ref{fig:asc-hybrid}: the judgment $\ju{C;x:\sSk}{t:T}$ (pos. $a\ct 0$) yields $\ju{C}{\lx.t:\sSk\rew T}$ below (pos. $a$). The occurrence of $T$ in the conclusion of the rule is intuitively the same as that in its premise: we say the former is the \textbf{ascendant} of the latter, since it occurs above in the typing derivation. Moreover, the context $C$ of the judgment typing $\lx.t$ comes from $C$ in the judgment typing $t$ and the source $\sSk$ of the arrow type of $\lx.t$ comes from $x:\sSk$ in the context of the premise.
Likewise, in the $\app$-rule, the occurrence of $T$ in $\ju{C\uplus_{\kK} D_k}{t\,u:T}$ stems from that of premise $\ju{C}{t:\sSk\rew T}$: the first occurrence of $T$ is also the ascendant of $T$ in the conclusion of the rule. Moreover, any part of the context $C\uplus(\uplus_{\kK}D_k)$ of the judgment typing $t\,u$ comes from a part of $C$ or of some $D_k$ in one of the premise.

Recalling the correspondence between $c$ and $1\ct c$ above, the relation of ascendance is defined by:
\begin{itemize}
\item For all $a\in \supp{P}$ such that $t(a)=\arob$:
\begin{itemize}
\item For all $c\in \suppmut{\ttT(a)}$, $(a,c)\rewa (a\cdot 1,1\cdot c)$.
\item For all $x\in \TermV$ and $k\cdot c\in \suppmut{\ttC(a)(x)}$, $(a,x,k\cdot c)\rewa (a\cdot \ell,x,k\cdot c)$ where $\ell\geqs 1$ is the unique integer such that $k\cdot c\in \supp{\ttC(a\cdot \ell)(x)}$.
\end{itemize}
\item For all $a\in \supp{P}$ such that $t(a)=\lx$ for some $x$:
\begin{itemize}
\item For all $c\in \bbN^*$ such that $1\cdot c\in \suppmut{\ttT(a)}$, $(a,1\cdot c)\rewa (a\cdot 0,c)$.
\item For all $y\in \TermV,\,y\neq x$ and all $k\cdot c\in \suppmut{\ttC(a)(y)}$, $(a,y,k\cdot c)\rewa (a\cdot 0,k\cdot c)$.
\item For all $c\in \bbN^*$ and $k\geqslant 2$ such that $k\cdot c\in \suppmut{\ttT(a)}$, $(a,k\cdot c)\rewa (a\cdot 0,x,k\cdot a)$.
\end{itemize}
\end{itemize}
Relation $\p_1\rewa  \p_2$ means that $\p_2$ is the ascendant of $\p_1$ \ie $\p_1$ and $\p_2$ are corresponding pointers to the same edge in the conclusion and the (left) premise of the rule at some position $a$.
For instance, in $\Pex$, $(\epsi,\,\epsi)\rewa (0,\,1)\rewa (0\ct 1,\,\epsi)$ (those 3 bipositions point to type symbol $o'$, from the judgment concluding $\Pex$ to the axiom rule where it was created (position $0\ct 1$)). \\

\noindent \textbullet~ Let us have another look at the $\ax$-rule in Fig.\;\ref{fig:asc-hybrid}. We may identify the occurrence of $S_5$ in the context (placed on track 5 in the singleton sequence type $(5\cdot S_5)$) and that in the right part of the judgment. We then define the relation of \textbf{polar inversion} by:
\begin{itemize}
\item For all $a\in \supp{P}$ such that $t(a)=x$ for some $x$, for all $c\in \ttT(a)$, $(a,x,k\cdot c)\rewp (a,c)$, with $k=\trP{a}$.
\end{itemize}
Relation $\p_1\rewp \p_2$ means that $\p_1$ is the \textit{(negative) polar inverse} of $\p_2$. For instance, in $\Pex$, $a=0\cdot 3$ points to an $\ax$-rule typing $x$, and $(0\cdot 3,x,2)\rewp (0\cdot 3,\epsi)$ (those 2 bipositions point to type symbol $\tv'$).\\

As planned in Sec.\;\ref{s:edge-threads-techrep}, let us call an \textbf{mutable edge thread} or simply a \textbf{thread} an equivalence class of relation $\equiv$ on $\Edg{P}$. The quotient set $\Edg{P}/\equiv$ is denoted $\ThrE{P}$. If $\tte\in \Edg{P}$, we denote by $\ttthr^P_{\ttE}(\tte)$ or just by $\thr{\tte}$ the thread of $\tte$ \ie the unique $\theta\in \ThrE{P}$ such that $\tte \in \theta$. In that case, we say that $\theta$ \textbf{occurs} at $\tte$ and we write $\theta:\tte$ or $\tte:\theta$.

The lemma below states that two edges of the same thread are labelled with the same track, as expected from Sec.\;\ref{s:edge-threads-techrep}:

\begin{lemma}
Let $\tte_1,\tte_2\in \ThrE{P}$ such that $\tte_1\equiv \tte_2$. Then $\lab{\tte_1}=\lab{\tte_2}$.
\end{lemma}

\noindent And more generally:

\begin{lemma}
\label{l:relab-and-iden}
Let $\Relab$ be a relabelling of $P$. Then, for all $\tte_1,\tte_2\in \Edg{P}$ such that $\tte_1\equiv \tte_2$. Then $\Relab(\tte_1)=\Relab(\tte_2)$.
\end{lemma}

\begin{proof}
By induction, since the lemma is true when we replace $\equiv$ by $\rewa$ and $\rewp$. 
\end{proof}

Let $\idena$ be the reflexive, transitive, symmetric closure of $\rewa$. Notice that $\rewa$ is functional: if $\p_1 \rewa \p_2$, we write $\p_2=\asc(\p_1)$. Notice also that $\asc$ is injective. Thus, $\p_1 \idena \p_2$ iff $\exists i\geqslant 0,~ \p_2=\asc^i(\p_1)$ or $\p_1=\asc^i(\p_2)$.

Given a right biposition $\p=(a,c)$ (resp. a left biposition $(a,x,k\cdot c)$) in $\bisupp{P}$, we call $a$ the \textbf{outer} and $c$ (resp. $k\cdot c$) the \textbf{inner position} of $\p$. We then write $\out{\p}=a$ and the \textbf{applicative depth} (extending the notion of Sec.\;\ref{s:tracks}) of $\p$ is defined by $\ad{\p}=\ad{\out{\p}}$ .


\subsection{Consumption}
\label{s:consum-interface}


The notion of \textbf{consumption} is associated with rule $\apph$. Assume $t(a)=\arob,\,\tra=u\,v$ with $u:\sSk \rew T$ for all $\kK$ and $v:S'_k$ for all $\kK'$ for all $\kK$ with $\phi_a:\sSk \trew \sSpkp$ as in Fig.\;\ref{fig:asc-hybrid} so that $u\,v$ can be typed with $T$. The types $S_k\ (\kK)$ and $S'_k\ (\kK')$ occur in the premises of the judgment typing $t\,u$, however, they are absent in this judgment. We say that they have been \textbf{consumed}. Intuitively, for all $\kK$, the sequence type isomorphism $\phi_a$ given by the interface of $P$ will identify every edge $\tte$ of $S_k$ with some edge $\tte'$ of some $S'_{k'}$.
Formally, we set:
\begin{itemize}
\item For all $a,c,c'\in \bbN^*,\ k,k'\in \Nmzo$ such that $(a\cdot 1,k\cdot c)\in \bisuppmut{P}$, $\phi_a(k\cdot c)=k'\cdot c'$
:
\begin{itemize}
\item If $c\neq \epsi$, $(a\cdot 1,k\cdot c)\rewc{a} (a\cdot k', c')$
\item $(a\cdot 1,k)\rewc{a} a\cdot k'$.
\end{itemize} 
\end{itemize}  
Indeed, the premise concluding with $u:\sSk \rew T$ is at position $a\ct 1$. The $c\in \supp{S_k}$ corresponds to position $k \ct c$ in  $\supp{\sSk\rew T}$. The interface $\phi_a$ will map $S_k$ onto the type $S'_{k'}$ (with $\rho_a=\Rt(\phi_a)$ and 
$k'=\rho_a(k)$, so that $S'_{k'}\equiv S_k$), that occurs in the judgment $\ju{D_{k'}}{v:S'_{k'}}$ at position $a\cdot k'$.
But the edge of $S_k$ that ends at some position $c$ will be mapped by $\phi_a$ on the edge of $S_k'$ that ends at $c':=\phi_a\rstr{k}(c)$. Moreover, biposition $(a\cdot 1,k)$, that stands for the \textit{inner} edge joining $(a\cdot 1,\epsi)$ with $(a\cdot 1,k)$, will be mapped on the \textit{argument} edge joining $a$ with $a\cdot k'$, denoted simply by $a\cdot k'$. For instance, in the trivial derivation $\Pex$, there is an $\app$-rule at position 0, $\phi_0$ is the identity and $(0\ct 1,8)\rewc{0} 0\ct 8,\; (0\ct 1,3)\rewc{0} 0\ct 3$.

We set $\rew\,=\cup \set{\rewc{a}\,|\, a\in \suppat{P}}$ and write $\wer$ for the symmetric relation. Implicitly, the relations $\rewc{a}$ and $\rew$ both depend on the interface $\phi$ of $P$.\\

We consider now the extension of relation $\rew$ modulo $\iden$. We write $\theta_1 \rrewc{a} \theta_2$ if $\exists \p_1,\p_2,~ \theta_1=\thr{\p_1},\ \theta_2=\thr{\p_2},\ \p_1\rewc{a} \p_2$. Thus, $\theta_1\rrewc{a} \theta_2$ iff
$\theta_1:\p_1\rewc{a} \p_2:\theta_2$ for some $\p_1,\p_2$. In that case, we say that $\theta_1$ (resp. $\theta_2$) has been \textbf{left-consumed} (resp. \textbf{right-consumed}) at biposition $\p_1$ (resp. $\p_2$). 

\subsection{Brother Chains}
\label{s:brother-chains}

Let $\Relab$ be a relabelling of $P$ and 
assume that $\phi(\p_\Lst)=\p_\Rst$. Thus, the edges $\p_{\Lst}$ and $\p_{\Rst}$ are identified by the interface.
We recall that $\lab{\p_{\Lst}}$ is the label of the edge $\p_{\Lst}$, $\lab{\p_{\Rst}}$ that of $\p_{\Rst}$ and that $\Relab(\p_{\Lst})$ and $\Relab(\p_{\Rst})$ (that is $\Relab(\phi_a(\p_{\Lst}))$) will be their new labels assigned by $\Relab$. According to the discussion of \ref{s:edge-threads-techrep}, if we want $\Relab$ to produce a trivial derivation $P^{\Relab}$ (we say then that $\Relab$ \textbf{trivializes} $P$), we need $\p_{\Lst}$ and $\p_{\Rst}$ to be assigned the same label. Indeed (we reuse the notation $\ttbisuppL$ introduced in Sec.\;\ref{s:encoding-red-choices}):

\begin{lemma}
\label{l:getting-a-trivial-derivation}
Let $P$ be an operable derivation $P$ and $\Relab$  a relabelling of $P$.
If, for all $\p \in\bisuppL{P}$, $\Relab(\theta)=\Relab(\phi(\theta))$, then $P^{\Relab}$ is a trivial derivation.
\end{lemma}

\begin{proof}
Let $\Psi$ the resetting induced by $\Relab$. Let $a_0\in \suppat{P^{\Relab}}$ and $a=\Psi^{-1}(a_0)$. From Sec.\;\ref{s:res-interface}, we recall that the interface  $\phi^\Relab_{a_0}$ of $P^{\Relab}$ is defined by
$\phi^{\Relab}_{a_0}=\Psi^{\Rst}_{a} \circ \phi_a \circ (\Psi^{\Lst}_a)^{-1}$.

Let $k_0\cdot c_0\in \suppmut{\Lst^{\Relab}(a_0\cdot )}$. We set $\p_0=(a_0,k_0\cdot c_0)$ and $\p=\Psi^{-1}(\p_0)$. Since $\Relab(\theta)=\Relab(\phi(\theta))$, we have $\lab{\phi_{a_0}(k_0\cdot c_0)}=\lab{k_0\cdot c_0}$. 

By induction on $|k_0\cdot c_0|$, we show then that, for all $k_0\cdot c_0 \in \supp{\Lst^{\Relab}(a_0)}$, $\phi_{a_0}(k_0\cdot c_0)=k_0\cdot c_0$.
\end{proof}

Given an operable derivation $P$ with the usual notations, we consider now the following first order theory $\calT_P$, whose sets of constants is $\ThrE{P}$, whose unique function symbol is $\ttVal$ and that holds the axioms $\Val{\theta_1}=\Val{\theta_2}$ for all $\theta_1,\theta_2 \in \ThrE{P}$ such $\exists a\in \suppat{P},\ \theta_1\rrewc{a} \theta_2$. Intuitively, $\Val{\theta_1}$ stands for $\Relab(\theta)$ (\ie $\Relab(\tte)$ for any $\tte \in \theta$) where $\Theta$ is a relabelling that trivializes $P$.\\

We say that $\tte_1,\tte_2\in \Edg{P}$ are \textbf{brother edges}  if they have a node in common \ie  either $\tte_1=a\cdot k_1,\; \tte_2=a\cdot k_2$ or $\tte_1=(a,c\cdot k_1),\; \tte_2=(a,c\cdot k_2)$ or $\tte_1=(a,x,\ell\cdot c\cdot  k_1),\; \tte_2=(a,x,\ell\cdot c \cdot k_2)$ for some $k_1\neq k_2$.

\begin{rmk}
\label{rmk:brother-e}
Actually, we strengthen this definition by declaring that $\tte_1$ and $\tte_2$ are also brother edges if $\tte_i=(a_i,t(a_i),\trP{a_i})$ ($i=1,2$) where $a_1$ and $a_2$ are two distinct axiom rules. Thus, ``axiom edges'' are regarded as brothers, even if they type two distinct variables. We do so to match the definition of resetting and relabelling given in Sec.\;\ref{s:bisupp-reset} and Sec.\;\ref{s:relab-deriv}.
\end{rmk}

Of course, two brother edges cannot be relabelled with the same track value (see Definition\;\ref{def:relabelling} and Sec.\;\ref{s:tracks}). Likewise, we say that $\theta_1,\theta_2\in \ThrE{P}$ are \textbf{brother threads} if $\theta_1:\tte_1$ and $\theta_2:\tte_2$ for some brother edges $\tte_1,\; \tte_2$.

Actually, the only thing that could go wrong is that the interface $\phi$ of $P$ may syntactically constrain two brother threads to be relabelled with the same edge in a prospective relabelling that would trivialize $P$. Indeed, if this does not happen, we can prove that such a relabelling exists:

\begin{proposition}
\label{prop:brother-chains-must-not-exists}
Let $P$ be an operable derivation. If $\calT_P$ cannot prove an equality of the form $\Val{\theta_1}=\Val{\theta_2}$ for some brother threads $\theta_1$ and $\theta_2$, then $P$ is isomorphic to a trivial derivation.
\end{proposition}

\begin{proof}
Under the assumption of the statement, we define an equivalence relation $\sim$ on $\ThrE{P}$ by
$\theta_1\sim \theta_2$ iff $\calT_P$ proves that $\Val{\theta_1}=\Val{\theta_2}$(\ie $\calT_P\vdash \Val{\theta_1}=\Val{\theta_2}$).

Since $\ThrE{P}$ is a countable set, let $\tti$ be an injection from $\ThrE{P}/\sim$ to $\Nmzo$. We define then a relabelling $\Relab$ of $P$ by:
\begin{itemize}
\item $\Relabarg(a)=\tti(\thr{a})$ for all $a\in \suppmut{P}$.
\item $\Relab_a(c)=\tti(\thr{a,c})$ for all $a\in \Ax$ and $c\in \suppmut{\ttT(a)}$.
\item $\Relabtr(a)=\tti(\thr{a,t(a),\trP{a}})$ for all $a\in \Ax$.
\end{itemize}
Since no brother edges are assigned the same value by hypothesis, this definition matches the clauses of definition\;\ref{def:relabelling} and $\Relab$ is indeed a relabelling of $P$. Notice that $\Relabtr$ is injective because of Remark\;\ref{rmk:brother-e}.
By Lemma\;\ref{l:getting-a-trivial-derivation}, $P^{\Relab}$ is a trivial derivation, that is isomorphic to $P$, thanks to Sec.\;\ref{s:relab-deriv}.
\end{proof}

According to Proposition\;\ref{prop:brother-chains-must-not-exists}, in order to prove Theorem\;\ref{th:rep}, we must now check that, for any operable derivation $P$, there is not proof in $\calT_P$ that two brother threads $\theta_1$ and $\theta_2$ should be assigned the same value. Such a proof would be called a \textbf{brother chain} and could be written $\theta_0\lrrewc{a_0} \theta_1\lrrewc{a_1}\theta_2\ldots \ldots \theta_{n-1} \lrrewc{a_{n-1}} \theta_n$ where $\theta_0$ and $\theta_n$ are two brother threads and $\lrrewc{a_i}$ is the symmetric closure of $\rrewc{a_i}$. We must then prove that brother chains \textit{do not} exist.


}


\techrep{\section{Representation Theorem }

We intend to prove here Theorem\;\ref{th:rep}:

\begin{theorem*}
Every operable derivation is isomorphic to a trivial derivation.
\end{theorem*}

Since Lemma\;\ref{l:rep-reduction-choices} ensures that every way of performing
finite or not sequences of subject-reduction can be built-in inside an operable,
it establishes that the ``rigid'' framework $\ttS$ does not cause any loss of
expressivity compared to type system $\scrR$ resorting to multiset constructions.\\

}
\paper{
\subsection{Collapsing Brother Threads}
\label{ss:syntactic-polarity-and-collapsing-strat}}

\techrep{
By Proposition\;\ref{prop:brother-chains-must-not-exists} and the discussion that follows, it is sufficient to prove that brother chain do not exist.
We start by giving an outline of the final stages of the proof. 
}
We prove now the main theorem: according to Prop.\;\ref{prop:if-no-bro-then-th-rep}, in order to prove Theorem\;\ref{th:rep}, we must prove that brother chains \textit{do not} exist. For that, we assume \textit{ad absurdum} that there is a brother chain \paper{$\theta_0 \lrrewc{a_0} \theta_1\lrrewc{a_1}\ldots \lrrewc{a_{n-1}} \theta_n$}\techrep{$\theta_0\lrrewc{a_0} \theta_1\lrrewc{a_1}\theta_2\ldots \ldots \theta_{n-1} \lrrewc{a_{n-1}} \theta_n$} associated to an operable derivation $P$.\\

\noindent 
\textbf{Argument 1 (normal brother chains do not exist):} Using Lemma\;\ref{l:uniqueness-consumption-hybrid}, it is easy to prove that, if no thread is left-consumed negatively in the chain (in that case, we say that the chain is a  \textbf{normal brother chain}), then it is of the form $\theta_0\loplus \rrewc{a_1}\ldots \rrewc{a_{n-1}} \theta_n$ and that $a_1\;,a_2,\;\ldots$ are resp. nested in an argument of $a_0,\;a_1,\ldots$. By Lemma\;\ref{l:consum-and-ad}, this entails that $\ad{a_0}<\ad{a_1}<\ad{a_2}\ldots$ Since $\theta_0\rrewc{a_0}$ and  $\rrewc{a_{n-1}}\theta_n$ implies that the occurrences of the thread $\theta_0$ are above or below $a_0\cdot 1$ whereas some edges of $\theta_n$ are above $a_{n-1}\cdot k$ (for some $k\geqs 2$). From that and $\ad{a_0}<\ad{a_n\cdot k}$, we deduce that $\theta_0$ and $\theta_n$ cannot be brother (roughly because two brother threads have the same applicative depth). Contradiction.\\

\noindent 
\textbf{Argument 2 (if they existed, brother chains could be normalized):}
The idea is that if $\thL \lomin \rrewc{a} \thR$, then either $\tra$ is a redex and we have $\Res_b(\thL)=\Res_b(\thR)$ (\ie $\thL$ and $\thR$ are \textbf{collapsed} by the reduction step) or there is a redex ``inside'' $\thL$. When we reduce it, the ``height'' of $\thL$ will decrease. More precisely, the 2nd case is associated to the notion of \textbf{redex tower}\paper{, which is more or less a \textit{finite} nesting of redexes, that can---more importantly---be collapsed in a \textit{finite} number of steps. A case of negative left-consumption of a sequence type $\sSk$ (which is the domain of the abstraction $\lx.u$), coming along with a redex tower, is represented  in Fig.\;\ref{fig:redex-tow-hybrid} (by lack of space, we write $\lam 1,\lam 2,\ldots$ instead of $\lx_1,\lx_2,\ldots$ and $(*)$ for matterless sequence types).
The sequence type $\sSk$ of negative polarity is ``called'' by $\lx$ at the top corner of the first term and consumed at the bottom $\app$-rule. 

The initial redex tower is reduced\techrep{ in 3 steps, so that}\paper{,} its height decreases and finally, the (residuals $S'_k$ of the) types $S_k$, that were left-consumed negatively, are destroyed in the final term $u_3[v/x]$.
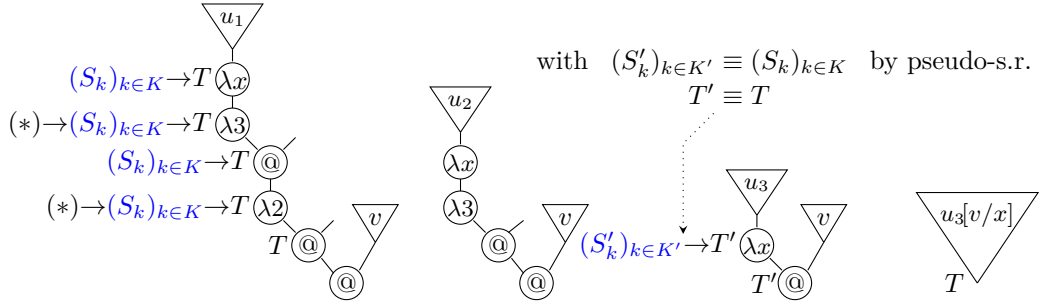
\begin{figure}
\begin{tikzpicture}


\ignore{  
  \draw (0,4.5) --++ (-0.3,0.5) --++ (0.6,0) --++ (-0.3,-0.5);
  \draw (0,4.8) node {\small $u$};

  \draw (0,4.34) -- (0,4.5);
  
  \draw (0,4.1) node {\small $\lx$};
  \draw (0,4.1) circle (0.22);
  \draw (-1.2,4.1) node {$\blue{\sSksh}\rew T$};

  \draw (0,3.74)-- (0,3.86);
  
  \draw (0,3.5) node {\small $\lambda 3$};
  \draw (0,3.5) circle (0.22);
  \draw (-1.5,3.5) node {$(*)\!\rewsh \blue{\sSksh}\! \rewsh T$};
  
  \draw (0.17,3.33) -- (0.33,3.17);
  \draw (0.83,3.33)--  (0.67,3.17);
  
  \draw (0.5,3) node {$\arob$};
  \draw (0.5,3) circle (0.22);
  \draw (-0.7,3) node {$\blue{\sSksh}\rew T$};
  
  \draw (0.5,2.62) -- (0.5,2.78);

  \draw (0.5,2.4) node {\small $\lambda 2$};
  \draw (0.5,2.4) circle (0.22);
  \draw (-1.05,2.4) node {$(*)\rewsh \blue{\sSksh} \rewsh T$};

  \draw (0.5,2.02) -- (0.5,2.18);

  \draw (0.5,1.8) node {\small $\lambda 1$};
  \draw (0.5,1.8) circle (0.22); 
  \draw (-1.4,1.8) node {$(*)\rewsh(*)\rewsh \blue{\sSksh} \rewsh T$};
  
  \draw (0.67,1.63) -- (0.87,1.43);
  \draw (1.33,1.63) -- (1.17,1.43);
  
  \draw (1,1.3) node {$\arob$};
  \draw (1,1.3) circle (0.22);
  \draw (-0.8,1.3) node {$(*)\!\rew \blue{\sSksh}\! \rew T$};
  
  \draw (1.17,1.13) -- (1.33,0.93);
  \draw (1.83,1.13) -- (1.67,0.93);
  
  \draw (1.5,0.8) node {$\arob$};
  \draw (1.5,0.8) circle (0.22);
  \draw (0.2,0.8) node {$\blue{\sSksh}\rew T$};
  
  \draw (1.67,0.63) -- (1.83,0.43);

  \draw (2,0.3) node {$\arob$};
  \draw (2,0.3) circle (0.22);
  \draw (1.4,0.3) node {$T$};
  \draw (-1,0.2) node [below] {\small The $\blue{S_k}$ are consumed there}; 
  \draw [>=stealth,->] (0.9,-0.1) -- (1.6,0.2);
  
  \draw (2.17,0.47) --++ (0.21,0.37);
  \draw (2.38,0.84) --++ (-0.3,0.5) --++ (0.6,0) --++ (-0.3,-0.5);
  \draw (2.38,1.14) node {$v$};
  \draw (3,0.8) node {$\sSksh$};
  
}

  \draw (3.2,3.4) --++ (-0.4,0.6) --++ (0.8,0) --++ (-0.4,-0.6);
  \draw (3.2,3.8) node {\small $u_1$};

  \draw (3.2,3.24) -- (3.2,3.4);

  \draw (3.2,3) node {$\lx$};
  \draw (3.2,3) circle (0.22);

  \draw (3.05,3) node [left] {$\blue{\sSk}\rewsh T$};
  
  \draw (3.2,2.62) -- (3.2,2.78);

  \draw (3.05,2.4) node [left] {$(*)\rewsh\blue{\sSk}\rewsh T$};  
  \draw (3.2,2.4) node {\small$\lambda 3$};
  \draw (3.2,2.4) circle (0.22);

  \draw (3.53,2.07) -- (3.37,2.23);
  \draw (3.87,2.07) -- (4.03,2.23);

  \draw (3.55,1.9) node [left] {$\blue{\sSk}\rewsh T$};
  \draw (3.7,1.9) node {$\arob$};
  \draw (3.7,1.9) circle (0.22);

  \draw (3.7,1.52)-- (3.7,1.68);

  \draw (3.55,1.3) node [left] {$(*)\rewsh\blue{\sSk}\rewsh T$};
  \draw (3.7,1.3) node {\small $\lambda 2$};
  \draw (3.7,1.3) circle (0.22); 

  \draw (4.03,0.97) -- (3.87,1.13);
  \draw (4.37,0.97) -- (4.53,1.13);

  \draw (4.05,0.8) node [left] {$T$};
  \draw (4.2,0.8) node {$\arob$};
  \draw (4.2,0.8) circle (0.22);

  \draw (4.37,0.63) -- (4.53,0.43);

  \draw (4.7,0.3) node {$\arob$};
  \draw (4.7,0.3) circle (0.22);
   
  \draw (4.87,0.47) --++ (0.21,0.37);
  \draw (5.08,0.84) --++ (-0.3,0.5) --++ (0.6,0) --++ (-0.3,-0.5);
  \draw (5.08,1.14) node {$v$};

  \draw (6.2,2.3) --++ (-0.4,0.6) --++ (0.8,0) --++ (-0.4,-0.6);
  \draw (6.2,2.7) node {\small $u_2$};

  \draw (6.2,2.14) -- (6.2,2.3);

  \draw (6.2,1.9) node {\small $\lx$};
  \draw (6.2,1.9) circle (0.22);

  \draw (6.2,1.52) -- (6.2,1.68);
  
  \draw (6.2,1.3) node {\small $\lambda 3$};
  \draw (6.2,1.3) circle (0.22);
  
  \draw (6.53,0.97) -- (6.37,1.13); 
  \draw (6.83,0.97) -- (7.03,1.13);
  
  \draw (6.7,0.8) node {$\arob$};
  \draw (6.7,0.8) circle (0.22); 
  
  \draw (6.87,0.63) -- (7.03,0.43);
  
  \draw (7.2,0.3) node {$\arob$};
  \draw (7.2,0.3) circle (0.22);

  \draw (7.37,0.47) --++ (0.21,0.37);
  \draw (7.58,0.84) --++ (-0.3,0.5) --++ (0.6,0) --++ (-0.3,-0.5);
  \draw (7.58,1.14) node {$v$};


  \draw (10.1,1.2) --++ (-0.4,0.6) --++ (0.8,0) --++ (-0.4,-0.6);
  \draw (10.1,1.6) node {\small $u_3$};

  \draw (10.1,1.04) -- (10.1,1.2);

  \draw (8.8,0.8) node {$\blue{\sSpkp}\rewsh T'$};
  \draw (10.1,0.8) node {\small $\lx$};
  \draw (10.1,0.8) circle (0.22);

  \draw (10.43,0.47) -- (10.27,0.63);
  
  \draw (10.6,0.3) node {$\arob$};
  \draw (10.6,0.3) circle (0.22);
  \draw (10.2,0.3) node {$T'$};

 \draw (10.77,0.47) --++ (0.21,0.37);
  \draw (10.98,0.84) --++ (-0.3,0.5) --++ (0.6,0) --++ (-0.3,-0.5);
  \draw (10.98,1.14) node {$v$};

  \draw (10.7,3)node {\parbox{7cm}{with $\begin{array}[t]{c}\sSpkp \equiv \sSk\\ T'\equiv T\end{array}$ by pseudo-s.r.}};
  \draw [dotted,->,>=stealth] (9.53,2.55) --++ (-0.4,-0.4) --++ (0,-1.2);
  

  \draw (13,0.3) --++ (-0.8,1.2) --++ (1.6,0) --++ (-0.8,-1.2);
  \draw (13,1.2) node {\small $u_3\![v/x]$};
  \draw (12.7,0.3) node {$T$};
  \end{tikzpicture}
\vspace{-0.3cm}
\caption{Collapsing a Redex Tower}
\label{fig:redex-tow-hybrid}
\vspace{-0.5cm}
\end{figure}

Argument 2 relies on definitions and techniques that we cannot develop here (see \S\;13.5 of \cite{VialPhd} for full details). Assume that $t\rewb{b} t'$ and that $P$ types $t$.
\begin{itemize}
  \item We prove that, if two edges $\tte_1$ and $\tte_2$ are in the same thread $\theta$, then their residuals $\Res_b(\tte_1)$ and $\Res_b(\tte_2)$ are in the same thread of the reduced derivation $P'$. This allows us to define the \textit{residual of the thread} $\theta$ as the unique thread of $P'$ containing $\Res_b(\tte_i)$.
  \item We prove that residuation for threads preserves brotherhood and consumption etc.
\end{itemize}
Anyway, the fact that  brother chains (if they existed) could be normalized, along with Argument 1, concludes the proof that brother chains do not exist. Thus, by Proposition\;\ref{prop:if-no-bro-then-th-rep}, Theorem\;\ref{th:rep} is proved, as well as Theorem\;\ref{th:rep-R-S-naive}.
}

\techrep{
\subsection{Syntactic Polarity}

We set, for all $\p\in \bisuppmut{P}$, $\Asc{\p}=\asc^i(\p)$, where $i$ is maximal (\ie $\asc^i(\p)$ is defined, but not $\asc^{i+1}(\p)$). Thus, $\Asc{\p}$ is the highest ascendant of $\p$ and is either located in the context of an $\ax$-rule or in its right-hand side, motivating the notion of (syntactic) polarity for bipositions:

\begin{definition}
\label{d:syntactic-polarity-hybrid}
  \begin{itemize}
\item Let $\tte\in \Edg{P}$.
\begin{itemize}
\item If $\tte=\p \in \bisuppmut{P}$, the \textbf{polarity}  of $\p$ is \textbf{positive} ($\Pol{\p}=\oplus$) if $\Asc{\p}=(a,c)$ for some $a\in \Ax$ and $c\in \mathtt{P}$ and it is negative ($\Pol{\p}=\omin$) if $\Asc{\p}=(a,x,k\cdot c)$ for some $a\in \Ax$ and $x\in \TermV$ such that $t(a)=x$.
\item If $\tte=a\in \suppmut{P}$, then $\Pol{\tte}=\oplus$.
\end{itemize} 
\item If $\thr{\p}=\theta$ and $\Pol{\p}=\oplus/\ominus$, we say that $\theta$ occurs positively/negatively at $\p$.
\item If $\theta$ is left/right-consumed at $\p$ and $\Pol{\p}=\oplus$ (resp. $\Pol{\p}=\ominus$), we say that $\theta$ is left/right-consumed positively (resp. negatively) at biposition $\p$.  
\end{itemize}
\end{definition}

Then, we write for instance $\theta_1 \loplus \rrewc{a} \romin \theta_2$
to mean that $\theta_1$  is left-consumed positively and $\theta_2$ is right-consumed negatively in the $\apph$-rule at position $a$.

Since $\rewp$ also defines an injective function 
 and $\p_1\rewp \p_2$ implies that $\p_1$  and $\p_2$  do not have ascendants:

\begin{lemma}
  \label{l:form-iden-hybrid}
For all $\p_1,\p_2\in \bisuppmut{P}$, we have $\p_1\iden \p_2,\, \Pol{a_1,c_1}=\oplus$ and $\Pol{a_2,c_2}=\omin$ iff $\Asc{a_2,c_2}\rewp \Asc{a_2,c_2}$.
\end{lemma}

Since a consumed biposition does not have a descendant, Lemma\;\ref{l:form-iden-hybrid} imply:

\begin{lemma}[Uniqueness of Consumption]
  \label{l:uniqueness-consumption-hybrid}
  Let $\oast \in \set{\oplus,\omin}$ and $\theta \in \ThrE{P}$. Then, there is a most one $\theta'$ such that ($\theta \loast  \rrew \theta'$ or $\theta \loast\werr \theta'$).
\end{lemma}


\subsection{Referents and Residual Relabelling}
\label{s:ref-and-res-relab}

Let $\Relab$ be a relabelling of an operable derivation $P\tri \juCtt$ (coming along with the usual notations) and $\tte\in \Edg{P}$. By Lemma\;\ref{l:relab-and-iden}, the value of $\Relab(\tte)$ may be related to an instance of $\Relabout,\, \Relabtr$ or $\Relab_a$ for some $a\in \Ax$. We define the \textbf{referent} of $\tte$ as the $a\cdot k \in \suppmut{P}$ or the $a\in \Ax$ or the mutable right biposition $(a,c\cdot k)$ such that $\Relab(\tte)$ is respectively equal to $\Relabout(a\cdot k)$ or $\Relabtr(a)$ or $\Relab_a(c\cdot k)$. There are several cases:
\begin{itemize}
  \item If $\tte=a\cdot k\in \suppmut{P}$, then $\thr{\tte}=\set{\tte}$ and we set $\refe{\tte}=\tte=a\cdot k$, so that $\Relab(\tte)=\Relabout(\tte)$.
  \item If $\tte=\p\in \bisupp{P}$:
  \begin{itemize}
  \item If $\Asc{\p}=(a,c\cdot k)$ for some $a\in \Ax$ and $c\cdot k\in \suppmut{\ttT(a)}$, then we set $\refe{\p}=\Asc{\p}=(a,c\cdot k)$, so that, by Lemma\;\ref{l:relab-and-iden}, $\Relab(\p)=\Relab_a(c\cdot k)$.
  \item If $\Asc{p}=(a,x,k)$ for some $a\in \Ax,\; x=t(a),\; k=\trP{a}$, then we set $\refe{\p}=\Asc{\p}=(a,x,k)$. Notice that that $\Relab(\p)=\Relabtr(a)$.
  \item If $\Asc{p}=(a,x,k\cdot c\cdot \ell)$ for some $a\in \Ax,\; t(a)=x,\; k=\trP{a}$ and $c\cdot \ell \in \suppmut{\ttT(a)}$, then we set $\refe{\p}=(a,c)$, so that $\Asc{p}\rewp \refe{\p}$, $\p \equiv \refe{\p}$ and $\Relab(\p)=\Relab(\refe{\p})=\Relab_a(c)$.
  \end{itemize} 
\end{itemize}
We denote by $\refe{P}$ the \textbf{set of referent edges} of $P$ and we use the meta variable $\ttr$ to denote referent edges.\\

Notice that if $\tte_1\iden \tte_2$, then $\refe{\tte_1}=\refe{\tte_2}$ and $\tte\iden \refe{\tte}$. Thus, the referent of a mutable edge $\tte$ may be seen as the unique representative of $\thr{\tte}$ that allows to directly compute $\Relab(\tte)$ from $\Relab$ as given in Definition\;\ref{def:relabelling}. We extend the notation $\ttref$ and for all thread $\theta\in \ThrE{P}$, we set $\refe{\theta}=\refe{\tte}$ for any $\tte\in \ThrE{P}$. If $\refe{\theta}=(a,x,k)$ (resp. $\refe{\theta}=a\cdot k\in \suppmut{P}$, resp. $\refe{\theta}=(a,c\cdot k)$), we say that $\theta$ is an \textbf{axiom thread} (resp. \textbf{argument thread}, resp. \textbf{inner thread}). 

A bit abusively, if $\refe{\tte}=(a,x,k)$ (where $a\in \Ax,\;x=t(a)$ and $k=\trP{a}$), we identify $\refe{\tte}$ with $a$ so that we may write that $\Relab(\tte)=\Relab(\refe{\tte})$ for all $\tte\in \Edg{P}$. We also abusively identify the set $\tilAx=\set{(a,x,k)\,|\, a\in \Ax,\; x=t(a),\;  k=\trP{a}}$ with  the set $\Ax$.\\

Lemma\;\ref{l:negative-thread-a-one-k} states that axiom threads occur at edges of some particular forms:

\begin{lemma}
\label{l:negative-thread-a-one-k}
\begin{itemize}
\item  If $(a,1^i\cdot k)\in \Edg{P}$ and $\Pol{a\cdot 1,1^i\cdot k}=\ominus$, then $\thr{a\cdot 1,1^i \cdot k}$ is an axiom thread. 
\item If $(a\cdot 1,k\cdot c)\in \Edg{P}$ and 
$(a\cdot 1,k\cdot c)\rewc{a} \eR$ with $k\in \Nmzo$ and $\thr{a\cdot 1,k\cdot c}$ is an axiom thread \ie $c=\epsi$ and $\eR$ is an argument edge \ie $\eR=a\cdot k'$ with $k'=\rho_a(k)$.
\end{itemize} 
\end{lemma}

\begin{proof}
\begin{itemize}
 \item By 001-induction on $a$. We split the cases according to whether $\asc{a,1^i\cdot k}$ is a right biposition or a left one (in the second case, $i=0$).
 \item By induction on $i$, we prove that $\p_i\asc^i(a\cdot 1,k\cdot c)$ is defined iff $\p^0_i=\asc^i(a\cdot 1,k)$ and that moreover, the respective inner positions $c^0_i$ of $\p^0_i$ and $c_i$ of $\p_i$ satisfy $c_i=c^0_i\cdot c$. Thus,
$\thL$ is an axiom thread \ie $\Asc{a\cdot 1,k\cdot c}$ is of the form $(a_*,x,k')$ only if $c=\epsi$ (then $k'=k'$).
\end{itemize}
\end{proof}

Lemma\;\ref{l:consumption-ax-and-arg-threads}  below means that axiom threads may only occur  at the root of sources of arrows and that consumption, when the left-hand side is negative, identifies some axiom threads with argument threads:

\begin{lemma}
\label{l:consumption-ax-and-arg-threads}
  \begin{itemize}
  \item If $\thL \rrew \thR$ and $\thL$ is an axiom thread, then $\thR$ is an argument thread.
  \item If $\thL \lomin \rrew \thR$ and $\thR$ is an argument thread, then $\thL$ is an axiom thread.
  \end{itemize}
\end{lemma}

\begin{proof}
\begin{itemize}
\item Assume that $\thL:\eL \rew{a} \eR:\thR$ and $\thL$ is an axiom thread.
Then, by Lemma\;\ref{l:negative-thread-a-one-k}, $\eL=(a\cdot 1,k)$ for some $k\in \Nmzo$ and $\eR=a\cdot k'$ with $k'=\rho_a(k)$. Thus, $\thR$ is the argument thread $\set{\eR}$.
\item Assume that $\thL:\eL \lomin \rew{a} \eR:\thR$ and $\thR$ is an argument thread. This implies that $\eR=a\cdot k'$ for some $k'\in \Nmzo$ and then, that
$\eL=(a\cdot 1,k)$ with $k=\rho^{-1}_a(k')$. Then, by Lemma\;\ref{l:negative-thread-a-one-k}, $\thL$ is an axiom thread.
\end{itemize}
\end{proof}

Let $\theta\in \ThrE{P}$ that is not an axiom thread, the \textbf{applicative depth} $\ad{\theta}$ of $\theta$ is defined by $\ad{\theta}=\ad{\refe{\theta}}$ (see Sec.\;\ref{s:asc-pol-inv-threads-hybrid}).

\begin{lemma}
\label{l:brother-ref}
Let $\tte_1,\;\tte_2 \in \Edg{E}$ be two brother edges and $\ttr_i=\refe{\tte_i}\; (i=1,2)$ their referents. Then $\ttr_1$ and $\ttr_2$ are brother threads.\\
\end{lemma}

\begin{proof}
Straightforward by induction on $\idena$.
\end{proof}

As a consequence, since two \textit{brother} referent edges that are not axiom edges have the same applicative depth, two brother threads that are not axiom threads have the same applicative depth and two threads are brother iff their referent are brother edges.

Moreover:

\begin{lemma}
\label{l:consum-brother-thr-same-pos}
Let $\theta_1$ and $\theta_2$ be two brother edges and $a\in \suppat{P}$. Then $\theta_1$ is consumed at position $a$ iff $\theta_2$ is. Moreover:
\begin{itemize}
\item If $\theta_1\rrewc{a} \theta^{\Rst}_1$, then, there is $\theta^{\Rst}_2$, brother with $\theta^{\Rst}_1$, such that $\theta_2\rrewc{a}\theta^{\Rst}_2$.
\item If $\theta^{\Lst}_1\rrewc{a} \theta_1$, then, there is $\theta^{\Lst}_2$, brother with $\theta^{\Lst}_1$, such that $\theta^{\Lst}_2\rrewc{a}\theta_2$.
\end{itemize}
\end{lemma}

\begin{lemma}
\label{l:consum-and-ad}
If $\thL\loplus\rrew \thR$, then $\ad{\thL}<\ad{\thR}$.
\end{lemma}

\begin{proof}
Assume that $\thL:\eL \loplus \rewc{a} \eR:\thR$. Excluding the case where $\thR$ is an argument thread (which is easy), we have
$\eL=(a\cdot 1,k\cdot c)$ and $\eR=(a\cdot k',c')$ for some $k,\,k'\in \Nmzo,\ c,\c'\in \bbN^*$ such that $\phi_a(k\cdot c)=k'\cdot c'$.

Since $\thR:(a\cdot k',c')\leqslant \refe{\thR}$, we have $\ad{\thR}\geqs \ad{a\cdot k'}=\ad{a}+1>\ad{a}$.

Since $\thL$ occurs positively, $\refe{\thR}=\Asc{a\cdot 1,k\cdot c}$ and $\ad{\refe{\Asc{a\cdot 1,k\cdot 1}}}=\ad{a\cdot 1}=\ad{a}$, so that $\ad{\thL}=\ad{a}<\ad{\thR}$.
\end{proof}

Let us assume that $P\tri \juCtt$ is an operable derivation (with the usual notations), $\trb=\lxrs$, $t\rewb{b} t'$ (so that $\tprb=\rsx$), $P\rewb{b} P'$ (so that $P'$ is a residual operable derivation of $P$) and that $\Relab$ is a relabelling of $\Relab$.

Let us observe now that $\Relab$ induces a relabelling $\Relab'$ of $P'$ \ie we will define a relabelling $\Relab'$ such that the following diagram commutes:
$$
\begin{array}{ccc}
P\phantom{\!\scriptstyle b} & \stackrel{\Psi^{\Relab} }{\longrightarrow} & P^{\Relab}\phantom{\!\scriptstyle b} \\
\downarrow {\!\scriptstyle b} & & \downarrow  {\!\scriptstyle b}\\
P'\phantom{\!\scriptstyle b} & \stackrel{\Psi^{\Relab'} }{\longrightarrow} & (P')^{\Relab'}\phantom{\!\scriptstyle b} 
\end{array}
$$

\begin{itemize}
\item If $\tte'=\al'\cdot k\in \suppmut{P'}$, then there is a unique $\al \in \suppmut{P}$ such that $\Res_b(\al\cdot k)=\al'\cdot k$ and we set $\Relabpout(\al')=\Relabout(\al)$.
\item If $\tte'=\al'\in \Ax'$, then there is a unique $\al \in \Ax$ such that $\Res_b(\al)=\al'$ and we set $\Relabptr(\al')=\Relabtr(\al)$.
\item If $\tte'=(\al',c\cdot k) \in \bisuppmut{P'}$ with $\al'\in \Ax'$, then there is a unique such that $\Res_b(\al)=\al'$. Moreover, by Lemma\;\ref{lem-def:QRes-bip-hybrid}, position $\al$  satisfies $\Res_b(\al,c\cdot k)=(\al',c\cdot k)$ and we set $\Relabpref_{\al'}(c)=\Relabref_{\al}(c)$.
\end{itemize}
Using again Lemma\;\ref{lem-def:QRes-bip-hybrid}, it is easy to check that $\Relab'$ matches Definition\;\ref{def:relabelling} and is a licit relabelling of $P'$. The commutation of the diagram above is easy to verify: we check that both paths define the same derivation support $A'$ and that for all leaf $\al'$ of $A'$, that the types are equal. We conclude then by using an easy 001-induction.

\subsection{Edges and Residuation}
\label{s:res-edges}


Let us observe now that a distinction should be made between residuation for nodes and residuations for edges. Consider an operable derivation $P\tri \juCtt$ (coming along with the usual notations) and assume that $\trb=\lxrs$, $t\rewb{b} t'$ (so that $\tprb=\rsx$), $P\rewb{b} P'$ so that $P'$ is a residual operable derivation of $P$. We use the same metavariable conventions as in Sec.\;\ref{s:res-deriv}. 
If $a\in \bbN^*$ and $k'\in \Nmzo$ are such that $\ovla=b$ and $a\cdot k'\in \supp{P}$, then $P(a)$ is a judgment of the form $\ju{D_{k'}}{s:S_{k'}}$, more precisely, $a\cdot k$ points to a  \textit{node} of $P$ labelled with a judgment typing the argument $s$ of the reduced redex. After reduction, this judgment will be located at position $a\cdot a_k$ where $k=\rho_a^{-1}(k')$. Thus, it was natural to define $\Res_b(a\cdot k')$ as $a\cdot a_k$ in Sec.\;\ref{s:res-deriv} (paradigm $\clubs$). This point of view that used throughout Sec.\;\ref{s:reduction} to define residuation for derivation.

But, as seen in Sec.\;\ref{s:relab-deriv}, we also use $a \cdot k$ to denote the \textit{edge} of $P$ from $a$ to $a\cdot k$. This is an edge that joins the $\app$-rule of the redex to an argument derivation of the redex. When we reduce the redex at position $b$, this edge is destroyed and should \textit{not} have a residual. This leads us to make a distinction between residuation for nodes (as in Sec.\;\ref{s:reduction}) and residuation for edges.
\ignore{We then define a partial function $\ResE$ from $\Edg{P}$ to $\Edg{P'}$ as follows: \pierre{TROP CHIANT}
\begin{itemize}
 \item Assume that $\al \cdot k\in \suppmut{P}$.
   \begin{itemize}
    \item If $\ovlal\neq b$, then $\al'=\Res_b(\al)$ is defined and we set $\ResE_b(\al\cdot k)=\al'\cdot k$.
    \item If $\al=a$ with $\ovla=b$, then $\ResE_b(a\cdot k)$ is left undefined.
    \end{itemize}
 \item Assume that $t(\al)=\ly$:
  \begin{itemize}
    \item If $\ovlal\neq b\cdot 1$, then 
    \item
  \end{itemize}
 \item
\end{itemize}
} Residuation for mutable edges should be then defined properly. However, 
relabelling matters only for \textit{referent edges}. Thus, it is enough for our purpose to define a residuation $\Resref$ that associates, to each referent edge of $P$, a referent edge of $P'$.


The definition of the residual $\Resref(\ttr)$ of a referent edge $\ttr$ is suggested by residual relabelling from Sec.\;\ref{s:ref-and-res-relab}. This residual is defined except when $\ttr$ is of the form $a\cdot a_k\in \Ax$ or of the form $a\cdot k\in  \suppmut{P}$ with $\ovla=b$ \ie $\ttr$ is not defined when it is an axiom edge associated to the variable of the redex or an argument edge associated to the application of the redex.
\begin{itemize}
\item Assume that $\ttr=\al \cdot k\in \suppmut{P}$.
\begin{itemize}
\item If $\ovlal\neq b$, then $\al':=\Res_b(a\cdot k)$ is defined and we set $\Resref_b(\ttr)=\Res_b(\al\cdot k)=\al'\cdot k\in \suppmut{P'}$.
\item If $\ovlal=b$ \ie $\al$ may be written $a$: $\Resref_b(\ttr)$ is left undefined.
\end{itemize}
\item Assume that $\ttr=(\al,y,k)$ with $\al \in \Ax,\; y=t(a),\; k=\trP{a}$ \ie abusively, $\ttr=\al\in \Ax$.
\begin{itemize}
 \item If $t(\al)\neq x$ (assuming Barendregt convention\footnote{If Barendregt convention does not hold, we must write $\ovl{\bind^{P}(\al)}\neq b\cdot 1$ instead of $t(\al)\neq x$ (see Sec.\;\ref{s:bisupp} for notation $\bind^{P}(\al)$).}), then $\al'=\Res_b(\al)$ is defined and we set $\Resref_b(\ttr)=(\al',y,k)$.
 \item If $t(\al)=x$ \ie $\al$ is an axiom rule typing the variable of the redex, then $\al$ is of the form $a\cdot 10 \cdot a_k$ for some $a\in \suppat{P},\; \ovla=b$ and $k\in \Trl{a}$ and $\Resref_b(\ttr)$ is left undefined.
\end{itemize}
\item Assume that $\ttr=(\al,c\cdot k)$ with $\al \in \Ax$ and $c\cdot k\in \suppmut{\ttT(\al)}$.
  \begin{itemize}
    \item If $t(\al)\neq x$, then $\al'=\Res_b(\al)$ is defined and we set $\Resref_b(\ttr):=\Res_b(\ttr)=(\al',c\cdot k)$, which is 
    \item If $t(\al) = x$, then $\Res_b(\al)$ is not defined but $\al':=\QRes_b(\al)$ is. However, $\QRes_b(\al)$ is not necessarily in $\Ax'$, so $\QRes_b(\ttr)$ is not necessarily an inner referent and we set $\Resref_b(\ttr)=\refp{\QRes_b(\ttr)}=\refp{\QRes_b(\al,c\cdot k)}$.
  \end{itemize}
\end{itemize}


\begin{lemma}
\label{l:pos-and-res-hybrid}
Assume that  $\al\in \supp{P}$ and $\ovlal \neq b\cdot 1$ and that $k\in \Ax_{\al}(y)$ for some $y\in \TermV$ with $y \neq x$ (\ie $\ovl{\bind^P(\al)}\neq b\cdot 1)$.\\
Then $\Res_b(\pos{\al,\,y,\,k})=\posp{\QRes_b(\al),\, y,\, k}$.
\end{lemma}

\begin{proof}
This stems immediately from the monotonicity of $\QRes_b$ \wrt the prefix order and the fact that, for any $\al\in A$ and $k\geqslant 2$, there is at most one $\al_0  \in \Ax_{\al}(y)$ s.t. $\tr{\al_0}=k$.
\end{proof}

\begin{lemma}
\label{l:ref-and-res}
Let $(\al,\, c)\in \bisuppmut{P}$  such that $\ovlal\neq
b\cdot 1$.\\
Then $\refp{\QRes_b (\al,\,c)}=\Resref_b(\refe{\al,\, c})$.
\end{lemma}

\begin{proof}
By $001$-induction on $\al$.\\

\begin{itemize}
\item When $\al\in \Ax$, it stems from the definition of 
$\Resref_b$, whether $t(\al)=x$ (\ie $\ovl{\bind^P(\al)}=b\cdot 1$)or not.\\

\item Assume $t(\al)=\ly$. Since 
$\ovlal\neq b\cdot 1$, then $y\neq x$, $\al'=\Res_b(\al)$ is defined, $y\neq x$ and $t'(\ovl{\al'})=\ly$. Of course, $\ovlal \cdot 0=b\cdot 1$ is impossible, so the induction hypothesis may be applied to $\al\cdot 0$.
\begin{itemize}
\item Subcase $c=1\cdot c_0$: we write $\Res_b(\al,\, 1\cdot c)=(\al',\, 1\cdot c_0')$ where $c'_0$ is the unique position such that. 
$1\cdot c_0'=\Res_{b|\al}(c)$, so that we have $\QRes_b(\al\cdot 0,\, c_0)=(\al'\cdot 0,\, c'_0)$. \\
By definition,
$\refe{\al,\, 1\cdot c_0}=\refe{\al\cdot 0,\, c_0}$ and $\refp{\al',\,1\cdot c'_0}=\refp{\al'\cdot 0,\,c'_0}$, 
By IH, $\refp{\QRes_b(\al\cdot 0,\,c_0)}=\Resref_b(\refe{\al\cdot 0, \,c_0})$. So $\refp{\al'\cdot 0,\, c_0'}=\Resref_b(\refp{\al\cdot 0,\, c_0})$. Thus, $\refp{\al',\,1\cdot c'_0}=\Resref_b(\refe{\al,\, 1\cdot c_0})$, as expected.
\item Subcase $c=k\cdot c_0$ with $k\in \Ax_{\al\cdot 0}(y)$: we must prove that
$\refp{\Res_b (\al,\,k\cdot c_0)}=\Resref_b(\refe{\al,\, k\cdot c_0})$ and 
we set  $\al_0=\pos{\al\cdot 0,\, y,\, k}$. Thus, $\refe{\al,k}=(\al_0,y,k)$
and  $\refe{\al,\, k\cdot c_0}= (\al_0,\, c_0)$ for any $c_0 \in \suppmut{\ttT(\al)}$.\\
Since $y\neq x$, we have $\Resref_b(\al_0)=\al'_0$ (\ie $\Resref_b(\al_0,y,k)=(\al'_0,y,k)$) and 
 $\Resref_b(\al_0,\,c_0)=(\al'_0,\,c_0)$ for all $c_0\in \suppmut{\ttT(\al_0)}$, where $\al'_0=\Res_b(\al_0)$.\\
On the left-hand side, $\Res_b(\al,\,k\cdot c_0)=(\al',\,k\cdot c_0)$ (by
definition of $\Res_{b|\al}$) and $\refp{\al',\,k\cdot \gam}=
\posp{\al'\cdot 0,\,y,\,k\cdot \gam}=(\eta',\,\gam)$, according to the
Lemma\;\ref{l:pos-and-res-hybrid}. The desired equality is proven.\\
\end{itemize}

\item Assume $t(\al)=\arob$. We set $(\al',\,c')=\QRes_b(\al,\, c)$.
\begin{itemize}
\item Subcase $\ovl{\al}\neq b$: then $\ovlal\cdot 1\neq  b=1$ and the induction hypothesis can be applied to $\al\cdot 1$.
Since $\ovl{\al}\neq b$,  we actually have $\Res_b(\al)=\al'$,
$t'(\al')=\arob$ and $\QRes_b(\al\cdot 1,\, 1\cdot c)=(\al'\cdot 1,\, 1\cdot c')$ by definition of $\QRes_b$.\\
By definition of referents, $\refe{\al,\,c}=\refe{\al\cdot 1,\,1\cdot c}$
and $\refp{\al',\,c'}=\refp{\al'\cdot 1,\,1\cdot c'}$.\\
By IH, $\refp{\QRes_b(\al\cdot 1,\, 1\cdot c)}=\Resref_b(\refe{\al\cdot 1, \,1\cdot c})$. So $\refp{\QRes_b(\al,\,c)}=\Resref_b (\refe{\al,\,c})$ as expected.
\item Subcase $\ovl{\al}=b$: by IH,  $\refp{\QRes_b(\al\cdot 10,\,c)}=\Resref_b(\refe{\al\cdot 10,\,c})$. Since $\refe{\al,\,c}= \refe{\al\cdot 10,\,c}$ and $\QRes_b(\al\cdot 10,\,c)=(\al,\,c')= \QRes_b(\al,\, c)$, the desired equality is true for $(\al,c)$.
\end{itemize}
\end{itemize}
\end{proof}

We may define \textbf{residuation for threads}: for all $\theta\in \ThrE{P}$, we set $\Res_b(\theta)=\Resref_b(\refe{\theta})$ ($\Res_b(\theta)$ is not defined when $\Resref_b(\refe{\theta})$ is not). By Lemma\;\ref{l:ref-and-res}, if $(\al,c)\in \theta),\ \ovlal \neq b\cdot 1$ and $\Res_b(\theta)$ is defined, then $\QRes_b(\al,c)\in \Res_b(\theta)$ and moreover:

\begin{corollary}
\label{corol:thr-and-res}
Let $(\al,\, c)\in \bisuppmut{P}$  such that $\ovlal\neq b\cdot 1$.\\
Then $\Res_b(\thr{\al,\, c})=\thrp{\QRes_b (\al,\,c)}$.
\end{corollary}

\subsection{Collapsing Threads}

\begin{lemma}
\label{l:collapsing-brothers}
Let $\theta_1$ and $\theta_2$ two brother threads.
If $\theta'_i=\Res_b(\theta_i)\; (i=1,2)$ are defined, then $\theta'_1$ and $\theta'_2$ are also brother threads in $P'$.
\end{lemma}

\begin{proof}
Consequence of Lemma\;\ref{l:brother-ref} and of the definition of $\Resref$ in Sec.\;\ref{s:res-edges}.
\end{proof}

\begin{lemma}
\label{l:collapsing-edge-threads}
Assume $\thL \rrewc{\al} \thR$.
\begin{itemize}
\item If $\ovlal\neq b$, then $\al':=\Res_b(\al)$ is defined, as well as $\Res_b(\thL)$ and $\Res_b(\thR)$, and $\Res_b(\thL) \rrewc{\al'} \Res_b(\thR)$  (the relation $\rrewc{\al'}$ is meant \wrt $P'$ instead of $P$).
\item If $\ovlal=b$, then $\Res_b(\thL)$ is defined iff  $\Res_b(\thR)$ is.\\
In that case, $\Resref_b(\thL)=\Resref_b(\rR)$.
\end{itemize}
\end{lemma}

\begin{proof}
The hypothesis means that there are there are $\kL\cdot \cL,\; 
\kR\cdot \cR$ such that $\thL:(\al\cdot 1,\kL\cdot \cL)$, $\thR:(\al\cdot \kR,\cR)$ and $\phi(\al\cdot 1,\kL \cdot \cL)=(\al \cdot \kR,\cR)$ \ie $\phi_{\al}(\kL\cdot \cL)=\kR\cdot \cR$ and in particular, $\kR=\rho_a(\kL)$ and $\cR=\phi_a\rstr{\kL}(\cL)$.


\begin{itemize}
\item Case $\ovlal\neq b$:  consequence of the definition of the residual interface (Sec.\;\ref{s:res-interface}) and of Corollary\;\ref{corol:thr-and-res}.
\item Case $\ovlal=b$: we write $a$ instead of $\al$ and $\Resref_b(\thL)$ is not defined iff $\refe{\thL}=(\al_0,\,x,k)$ with $\al_0 \in \Axl{a}$ and $k=\tr{\al_0}$ (see Sec.\ref{s:res-deriv} for notation $\Axl{a}$). Moreover, $\Res_b(\thR)$ is not defined iff $\thR$  is the thread of an argument edge $a\cdot k$ of the redex.  Moreover, $\rewc{a}$ induces a bijection from $\set{(a\cdot 1,k)\in \bisuppmut{P}\,|\,k\in \Nmzo}$ to $\set{a\cdot k'\in \suppmut{P}\,|\,k'\in \Nmzo}$(see Sec.\;\ref{s:consum-interface}) \ie the set of left edges whose referent is a $(a_0,x,k)$ for some $a_0\in \Axl{a}$ to the set of argument edges of $a$, which entails the first part of the lemma. 

We assume now that $\Res_b(\thL)$ and $\Res_b(\thR)$ are both defined. We must then have $\refe{\thL}=(a\cdot 10 \cdot a_{\kL},\cL)$.
Thus, $\Res_b(\thL)=\thrp{\Resref_b(a\cdot 10 \cdot a_{\kL},\, \cL)}=\thrp{\QRes_b(a\cdot 10\cdot a_{\kL},\, \cL)}=\thrp{a\cdot a_{\kL},\, \cR}= \thrp{\Res_b(a\cdot \kR,\, \cR)}=\Res_b(\thr{a\cdot\kR,\, \cR})=\Res_b(\thR)$. The second equality is by definition of $\Res_b$, the second-to-last one by Corollary\;\ref{corol:thr-and-res}.
\end{itemize}
\end{proof}

Assume that we have
$\thL \lomin \rrewc{a} \thR$ for some $\thL,\,\thR \in \ThrE{P}$. We set $\rL=\refe{\thL}$ and $\rR=\refe{\thR}$.
Thus, $\rL=(\al_0,y,k)$ or $\rL=(\al_0,c)$ for some $\al_0\in \Ax$ ,$y=t(\al_0),\; k=\tr{\al_0}$ or  $\rL=(\al_0,\,\, c)$ for some $\al_0 \in \Ax$ and 
$c \in \suppmut{\ttT(\al_0)}$  (in this second case, we set $y=t(\al_0)$).
There is a maximal $\al< \al_0$ such that $t(\al)=\ly$.
We define the \textbf{collapsing strategy \wrt $\thL^{\ominus}$} by induction on $h:=|\al| - |a|$:
\begin{itemize}
\item Case $h=1$: then $\al=a\cdot 1$ and there is a redex in $t$ at position $\ovla$. We fire it and the strategy is completed. By Lemma;\ref{l:collapsing-edge-threads}, $\Res_b(\thL)=\Res_b(\thR)$.
\item Case $h>1$: since $\thL \lomin \rrewc{\al} \thR$, there
is a maximal $a_0 <\al_0 $ such that $t(a_0)=\arob$. Since $a_0$ is maximal, then  $a_0\cdot 0 \in \supp{P}$ \ie
$t\rstr{a_0}$ is a redex. We
fire it,  so that, by Lemma\;\ref{l:collapsing-edge-threads}, $\Res_b(\thL)\rrewc{a}\Res_b(\thR)$, but the height $h$ has decreased by 2, and we go on with the strategy.
\end{itemize}
Let $\rs$ be the sequence of reductions representing
the collapsing strategy.
Lemma\;\ref{l:collapsing-edge-threads} entails that $\Resref_{\rs}(\rL) =\Resref_{\rs}(\rR)$. Thus:

\begin{lemma}
\label{l:collapsing-strategy-hybrid}
If $\thL\lomin \rrew \thR$, then, there is a reduction sequence $\rs$ such that $\Res_{\rs}(\thL)=\Res_{\rs}(\thR)$.
\end{lemma}

Lemma\;\ref{l:collapsing-strategy-hybrid} is extremely important: it entails that we can discard negative left-consumption by a finite sequence of reduction.


\ignore{ 
\subsection{Threads and Residuation}

\begin{lemma}
Let $P,P'$ be operable derivations such that $P \rewb{b} P'$. Let $\tte_1,\;\tte_2\in \Edg{P}$ such that $\tte_1\equiv \tte_2$ \pierre{and $\QRes_b(\tte_1)$ and $\QRes_b(\tte_2)$ are both defined.}
And in that case, $\QRes_b(\tte_1) \equiv \QRes_b(\tte_2)$, where $\equiv$ is meant here \wrt $P'$.
\end{lemma}

\begin{proof}
\pierre{on a besoin d'une notation pour les occurrences de la variable du redex\\
d'autre part, les quasi-residus de bip droites st tjrs definis: le faire remarquer. Pas ici}
We first recall that, for all $\al\in \supp{P}$
\begin{itemize}
 \item If $t(\al)=\arob$ and $\ovlal\neq b$, then $\Res_b(\al)$ is defined.
 \item If $t(\al)=\ly$ and $\ovlal\neq b\cdot 1$, then $\Res_b(\al)$ is defined.
 \item If $t(\al)=y$ and $\al$ may not be written $a\cdot 10\cdot a_k$ with $\ovla=b$ and $k\in \Trl{a}$, then $\Res_b(\al)$ is defined.
\end{itemize}
We also notice that quasi-residuals of \textit{mutable right} bipositions are always defined.

To prove the statement, we reason by induction on $\equiv$ and we just have to handle the cases $\tte_1\rewa \tte_2$ and $\tte_1\rewp\tte_2$.
\begin{itemize}
\item Assume $\al \in \suppat{P},\; c\in \suppmut{\ttT(\al)},\; \p_1=(\al,c)$ and $\p_2=(\al\cdot 1,1\cdot c)$, so that $\pi_1\rewa \p_2$.
\begin{itemize}
   \item Subcase $\ovlal\neq b$: we may have $\ovl{\al\cdot 1}=b$ or $a\cdot 1\in \blut$ or neither of those. By case analysis, we check that $\Res_b(\p_1)\rewa \QRes_b(\p_2)$.
   \item Subcase $\ovlal=b$: then $\ovlal{\al\cdot 1}=b\cdot 1$. Then $\QRes_b(\p_1)=\QRes_b(\p_2)=\p_1$.
\end{itemize}
\item Assume $\al \in \suppat{P},\; t(\al)=\ly,\; 1\cdot c\in \suppmut{\ttT(\al)},\; \p_1=(\al,1\cdot c)$ and $\p_2=(\al\cdot 0,c)$, so that $\p_1\rewa \p_2$.
\begin{itemize}
  \item Assume that $\ovlal\neq b\cdot 1$: we may have $\ovl{\al\cdot 0}=b$ or $\al \cdot 0 \in \blut$ or neither or those. By case analysis, we check that $\QRes_b(\p_1)\rewa \QRes_b(\p_2)$.
  \item Assume that $\ovlal=b\cdot 1$ \ie $\al=a\cdot 1$ with $\ovla=b$: then $\QRes_b(\p_1)=\QRes_b(\p_2)=(a,c)$.
\end{itemize}
\item Assume $\al \in \suppat{P},\; t(\al)=\ly,\; k\cdot c\in \suppmut{\ttT(\al)}$ with $k\geqslant 2$, $\p_1=(\al,k\cdot c)$ and $\p_2=(\al\cdot 0,x,k\cdot c)$, so that $\p_1\rewa \p_2$.
\begin{itemize}
  \item Assume that $\ovlal\neq b\cdot 1$: we may have $\ovl{\al\cdot 0}=b$ or $\al \cdot 0 \in \blut$ or neither or those. By case analysis, we check that $\QRes_b(\p_1)\rewa \QRes_b(\p_2)$.
  \item Assume that $\ovlal=b\cdot 1$ \ie $\al=a\cdot 1$ with $\ovla=b$: \pierre{then $\QRes_b(\p_1)=\QRes_b(\p_2)=(a\cdot a_k,\phi_a\rstr{k}(c))$.}
\end{itemize}
\item
\item Assume that $t(\al)=y$ for some $y\in \TermV$, $k=\tr{a}$, $c\in \supp{\ttT(\al)},\ \p_1=(\al,k\cdot c)$ and $\p_2=(\al,c)$, so that $\p_1\rewp \p_2$.
\end{itemize}
\end{proof}

}

\subsection{Conclusion of the Proof}
\label{s:proof-end-rep-th}

Thanks to Lemmas\;\ref{l:brother-ref},\; \ref{l:collapsing-edge-threads} and the collapsing strategy, if a brother chain $\calC$ existed, then we could produce a brother chain $\calC'$ written
$\theta_0 \lrrewc{a_0} \theta_1\lrrewc{a_1}\ldots \lrab{a_{n-1}} \theta_n$ (where $\theta_0$ and $\theta_n$ are brother threads) such that no thread is left-consumed negatively. We recall from that the chain $\calC'$ is said to be \textbf{normal}. Thus, if we prove Proposition\;\ref{prop:brother-chains-do-no-exist} below, Theorem\;\ref{th:rep} will be established:

\begin{proposition}
\label{prop:brother-chains-do-no-exist}
There is no \textit{normal} brother chain.
\end{proposition}

In order to prove the proposition above, we proceed \textit{ad absurdum} and consider a normal brother chain $\theta_0 \lrab{\phi_{a_0}} \theta_1\lrab{\phi_{a_1}}\ldots \lrab{\phi_{a_{n-1}}} \theta_n$ where $\theta_0$ and $\theta_n$ are brother threads.\\

Assume that there is brother chain $\calC$ that may be written
$\theta_0 \lrrewc{\phi_{a_0}} \theta_1\lrrewc{\phi_{a_1}}\ldots \lrab{\phi_{a_{n-1}}} \theta_n$ where $\theta_0$ and $\theta_n$ are brother threads.

We describe now an algorithm to produce from $\calC$ a normal brother chain $\calC'$ that is minimal.\\

By Lemma\;\ref{l:uniqueness-consumption-hybrid}. There may not be ``crossing'' of the form $\theta_{i-1}  \rrewc{a_{i-1}} \roplus \theta_i \loplus \rrewc{a_i}  \theta_{i+1}$. Moreover, if  there is a crossing
$\theta_{i-1}  \werrc{a_{i-1}} \roplus \theta_i \loplus \rrewc{a_i} \roplus \theta_{i+1}$, then (still by Lemma\;\ref{l:uniqueness-consumption-hybrid}), $a_{i-1}=a_i$ and $\theta_{i-1}=\theta_{i+1}$, so that we may assume that this kind of crossing never not occurs (if a some point these is one, we immediately remove it).

Likewise, we may remove every crossing of the form  $\theta_{i-1}  \werrc{a_{i-1}} \romin \theta_i \lomin \rrewc{a_i}  \theta_{i+1}$, since this implies that $\theta_{i-1}=\theta_{i+1}$.

Since argument threads only occur positively and axiom threads only negatively, Lemmas\;\ref{l:uniqueness-consumption-hybrid}
and \ref{l:consumption-ax-and-arg-threads}
imply that no $\theta_i$ may be an argument thread and that no $\theta_i$ may be an axiom thread that is left-consumed negatively: if not, we would have $n=1$ and $\calC$ would be $\theta_0\rrewc{a_0} \theta_1$ with $\theta_0$ axiom thread and $\theta_1$ (or $\theta_0\werrc{a_0} \theta_1$\ldots), which is not a brother chain !

Resorting to the collapsing strategy (Lemma\,\ref{l:collapsing-strategy-hybrid}), we can reduce $t$ so that we obtain a normal brother chain $\calC'$ associated to a reduct $t'$ of $t$. Notice that, by Lemma\;\ref{l:collapsing-edge-threads}, then length of $\calC'$ will be smaller than that of $\calC$.
The chain $\calC'$ may be written
$\theta'_0 \lrrewc{a'_0} \theta_1\lrrewc{a'_1}\ldots \lrab{a'_{n-1}} \theta'_n$, where $n'\leqs n'$, the $\theta'_i$ are mutable edge threads of $P'$, the residual operable derivation of $P'$ whose subject is $t'$. To lighten the notations, we drop the primes and just write $t'$ for $t$, $\calC$ for $\calC'$, $\theta_i$ for $\theta_i'$, $n$ for $n'$, and $a_i$ for $a'_i$.\\

Due to the first steps of the algorithm, we may only have four kinds of crossings in $\calC$.
$\theta_{i-1} \loplus \rsa{\al_{i-1}} \romin \theta_i \loplus \rsa{\al_i} \theta_{i+1}$ (right-right) or
$\theta_{i-1} \lsa{\al_{i-1}}  \roplus \theta_i\lomin \lsa{\al_i}\roplus \theta_{i+1}$ (left-left) or
$\theta_{i-1} \loplus \rsa{\al_{i-1}} \romin \theta_i \loplus \lsa{\al_i}\roplus  \theta_{i+1}$ (minus-plus) or $\theta_{i-1} \loplus \rsa{\al_{i-1}} \roplus \theta_i \lomin \lsa{\al_i}\roplus \theta_{i+1}$ (plus-minus)\\

We cannot have only crossings of kind right-right. If it were, by Lemma\;\ref{l:consum-and-ad}, we would have $\ad{\theta_0}<\ad{\theta_n}$. But, as it has been observed above, $\theta_0$ and $\theta_n$ cannot be axiom threads. By Lemma\;\ref{l:brother-ref}, we should have $\ad{\theta_0}=\ad{\theta_n}$, which is a contradiction. For the same reason, we cannot have only crossings of kind left-left.\\

Thus, there must be at least a crossing of kind minus-plus or of kind plus-minus. This is easy to see that there can only be one (there cannot be more that one ``change of direction'' in $\calC'$: if not, we would have a crossing on the form $\theta_{i-1} \werr \theta_i \rrew \theta_{i+1}$ and this case is not possible for $\calC'$). In both cases, the chain starts with $\theta_0\loplus \rrewc{a_0} \theta_1$ and ends with $\theta_{n-1} \werrc{a_{n-1}} \theta_n$. Lemmas\;\ref{l:consum-brother-thr-same-pos} and \ref{l:uniqueness-consumption-hybrid} entail that $a_0=a_{n-1}$ and $\theta_1$ and $\theta_{n-1}$ are brother threads. Thus, $\theta_1\lrrewc{a_1} \ldots \lrrewc{a_{n_2}} \theta_{n-1}$ is also a brother chain.
By induction on $i\leqslant n/2$, we show that $\theta_i$ and $\theta_{n-i}$ are brother threads, so that $n$ is even and $i_0:=n/2$ is a natural number and $\theta_{i_0}$ is brother with itself, which is impossible (by Lemma\;\ref{l:brother-ref}).
Thus, $\calC'$ cannot exist and so neither can $\calC$. This concludes the proof of Proposition\;\ref{prop:brother-chains-do-no-exist} and of Theorem\;\ref{th:rep}.


\ignore{

\begin{lemma} \label{lemConsum}
\begin{itemize}
\item If $r\in \Ax\times \{\epsi\} \cup \OAP$, there is at most one
$\al \in A\atmk$ such that $\exists r'\in \Thr{P},~ 
r\lrsa{\al} r'$.\\
In that case, either $\exists !r'\in \Thr{P},~ r \rsa{\al} r'$,
either $\exists ! r' \in \Thr{P},~ r \lsa{\al} r'$.
\item If $r\in \rbp (P)$ and $\circledast\in \{ \oplus,\, \ominus \}$, 
there is at most one $\al \in A\atmk$ such that 
$\exists r'\in \Thr{P},~ r^{\circledast}\lrsa{\al} r'$.\\
In that case, either $\exists !r'\in \Thr{P},~ r{}^\circledast \rsa{\al} r'$,
either $\exists ! r' \in \Thr{P},~ r {}^\circledast \lsa{\al} r'$.
\end{itemize} 
We say then that $r$ (resp. $r^\circledast$) is \textbf{consumed} at pos.
$\al$.
\end{lemma}

\noindent Moreover:

\begin{lemma} \label{lemBrotherRefs}
If $\circledast\in \{\oplus,\,\ominus\}$ and $r_1$ and $r_2$ are brother referents (that are not in $\OAP$), then $r_1^\circledast$ is consumed at pos. $\al$ iff $r_2$ is.\\
In the case where $r_i \in \dom \widetilde{\phi_{\al}}$, 
$\widetilde{\phi_{\al}}(r_1)$ and $\widetilde{\phi_{\al}}(r_2)$ are
brother referents.
\end{lemma}
}

\ignore{ 
\begin{proof}
By downward-induction on $|\alpha|$. We set $\alpha'=\QRes_b(\alpha)$.

\textbullet~ Case $\alpha \in \Ax(y)$: since $y\neq x$, 
$t'(\ovl{\alpha'})=y$ and $\posp{\alpha',\, y,\, k}$.\\
Since $\pos{\alpha,\,y,\,k}=\alpha$, $\Res_b(\pos{\alpha,\, y,\,k})=\alpha'$.
The equality is proven.\\

\textbullet~ Case $\alpha\cdot 0 \in A$: since $\ovl{\alpha}\neq b
\cdot 1$, $t(\ovl{\alpha})=\lambda z$ with $z\neq x,\,y,~ \alpha'
=\Res_b(\alpha),~ t'(\ovl{\alpha'})=\lambda z$ and 
$\QRes_b(\alpha\cdot 0)=\alpha'\cdot 0$.\\
We have $\pos{\alpha,\, y,\, k}=\pos{\alpha\cdot 0,\,y,\, k}$ and 
$\posp{\alpha',\,y,\,k}=\posp{\alpha'\cdot 0,\,y,\,k}$. By IH on
$\alpha'\cdot 0$, $\Res_b(\pos{\alpha\cdot 0,\,y,\, k})
=\posp{\QRes_b (\alpha\cdot 0,\,y,\,k)}$
\ie $\Res_b(\pos{\alpha,\,y,\, k})=\posp{\alpha'\cdot 0,\, y,\, k}$,
so $\Res_b(\pos{\alpha,\,y,\, k})=\posp{\alpha',\, y,\, k}$.\\

\textbullet~ Case $\alpha\cdot 1\in A$ and $\ovl{\alpha}\neq b$:
there is a unique $\ell \geqslant 1$ s.t. $k\in \AxTr(\alpha\cdot \ell,\,y)$.
Moreover, $\alpha'=\Res_b(\alpha),~ t'(\ovl{\alpha'})=\symbol{64}$ and 
$\QRes_{\alpha\ell}=\alpha\cdot \ell$. \\
We have $\pos{\alpha,\,y,\, k}=\pos{\alpha \cdot \ell,\,y,\, k}$
and also, by construction of $P'$, 
 $\posp{\alpha',\, y,\, k}=\posp{\alpha' \cdot \ell,\,y,\,k}$. By IH, 
$\Res_b(\pos{\alpha \cdot \ell,\, y,\, k})=\pos{\alpha'\cdot \ell,\,y,\, k}$ 
since here $\QRes_b(\alpha \cdot \ell)=\alpha' \cdot \ell$. So the
equality is true for $\alpha$.\\

\textbullet~ Case $\alpha\cdot 1 \in A$ and $\alpha=a$ where $a \in \Rep(b)$:
here, $\QRes_b(a)=a$ and there is also a unique $\ell \geqslant 1$ s.t. 
$k\in \AxTr(\alpha\cdot \ell,\, y)$.
\begin{itemize}
\item Subcase $\ell =1$: here, $\pos{\alpha,\,y ,\, k}=\pos(\alpha\cdot 10,\,
y,\, k)$ and by IH, $\Res_b(\pos{a\cdot 10,\, y,\, k})=
\pos{\QRes_b(a\cdot 10),\,y,\, k}$, whence $\Res_b(\pos{\alpha,\,y,\,
k})=\Res_b(\QRes_b(a),\, y,\, k)$ since $\QRes_b (a\cdot 10)=a=\QRes_b(a)$.
\item Subcase $\ell\geqslant 2$: we set $\ell_{\Rst}=\ell$ and $\ell_{\Lst}=
\phi_a^{-1}(\ell_{\Lst})$, so that $\QRes_b(\alpha \cdot \ell_{\Rst})=a \cdot a_{\ell_{\Lst}}$.\\
By IH,  $\Res_b(\pos{a \cdot \ell_{\Rst},\, y,\, k})=\posp{ \Res_b(a \cdot \ell_{\Rst}),\,y,\,k}$, \ie $\Res_b(\pos{a \cdot \ell,\,y,\, k})=\posp
{a\cdot a_{\ell_{\Lst}},\, y,\, k}$. But since $a\leqslant a\cdot a_{\ell_{\Lst}}$, $\posp{a,\,y,\, k}=\posp{a\cdot a_{\ell_{\Lst}},\,y,\, k}$, so $\Res_b(\pos{a,\,y,\, k})=\posp{\QRes(a,\,y,\,k)}$.\\
\end{itemize}
\end{proof}
}

}

\section{Conclusion}

The main results of this article (Theorems \ref{th:rep-R-S-naive}, \ref{th:rep} along with Lemma\;\ref{l:rep-reduction-choices}) state that every (possibly infinitary) multiset-based derivation is the collapse of a trivial derivation (which means sequential and permutation-free) and that, moreover, one can directly represent in the trivial setting  every reduction path in the multiset or operable setting, the latter meaning rigid with isomorphisms allowing us to encode reduction choices. This proves that the \textit{trivial} rigid framework is as expressive as the non-trivial one and advocates in favor of replacing multiset intersection with sequential intersection, since it is no more complicated to use sequences (system $\ttS$ does not use a permutation rule or type isomorphisms) than multisets, while having extra-features (parsing, pointing).

The proof technique of this article has been  actually used \textit{twice} in the companion paper \cite{compUnsoundArxiv}, by considering another first order theory, involving a variant of the notions of threads. This suggests that the method presented here is general and can be adapted again.

One can conjecture that this results applies to \cite{TsukadaAO17} (which resorts to lists of types, instead of sequences), in that it would allow us to ignore the use of isomorphisms in the study of the system (since they can be all assumed to be identities) and thus simplify \textit{proofs} in this latter article and in the framework that it introduces (up to the replacement of lists by sequences).

There are analogies between the sequential intersection framework and indexed linear logic, Girard's Geometry of Interaction or Ludics, or with game semantics (\eg between sequences and copy indices). They should be investigated in a close future.

\bibliographystyle{abbrv}
\bibliography{biblioThese.bib}


\renewcommand{\em}{\it}



\end{document}